\documentclass[11pt,letterpaper]{article}
\usepackage[utf8]{inputenc}
\usepackage[T1]{fontenc}
\usepackage{microtype}

\usepackage[left=1in, right=1in, top=1in, bottom=1in]{geometry}

\usepackage{subcaption}
\usepackage{nicefrac}
\usepackage{amsmath}
\usepackage{amsfonts}
\usepackage{amssymb}
\usepackage{amsthm}
\usepackage{amstext}
\usepackage{thmtools}
\usepackage{relsize}

\usepackage{csquotes}
\usepackage{xcolor}
\usepackage{graphicx}

\usepackage{mathtools}
\usepackage{xspace}
\usepackage{enumerate}
\usepackage[linesnumbered,ruled,vlined]{algorithm2e}

\usepackage{soul}

\usepackage{paralist}
\usepackage{bm}

\usepackage{url}
\usepackage[colorlinks]{hyperref}
\hypersetup{citecolor=green!60!black,linkcolor=blue!70!black,urlcolor=blue!70!black}
\usepackage[capitalize,nameinlink]{cleveref}
\usepackage[skip=5pt plus 3pt minus 1pt,parfill=0pt]{parskip}
\makeatletter
\def\thm@space@setup{\thm@preskip=\parskip \thm@postskip=0pt
}
\makeatother

\definecolor{buy}{RGB}{182,18,49}
\definecolor{sell}{HTML}{82e057}
\definecolor{lightblue}{HTML}{77CDFF}

\usepackage{tikz}
\usepackage{pgfplots}
\usetikzlibrary{decorations.pathreplacing}
\usetikzlibrary{calc,tikzmark}
\usetikzlibrary{positioning,shapes,arrows} 
\usetikzlibrary{matrix}
\tikzset{
	vx/.style={circle,fill,black,minimum size=7pt,inner sep=0pt,outer sep=0pt},
	comp/.style={circle,line width=2pt,draw=black,minimum size=0.6cm,inner sep=0pt, outer sep=1pt},
	pac/.style={circle,line width=0pt,minimum size=0.2cm,inner sep=0pt, outer sep=-2pt},
	edge/.style={draw=black,line width=1.7pt},
	ledge/.style={draw=black!20!white,line width=1.4pt,dashed},
	mypath/.style={draw=black!20!white,line width=1.4pt},
}

\usepackage[
 backend=biber, 
 style=alphabetic, 
 citetracker, 
 hyperref=auto,
maxcitenames=5, sortcites,      sorting=nyt,
 maxbibnames=12, date=year,      isbn=false,     url=false,      doi=false,      eprint=false,   ]{biblatex}
\newbibmacro{string+doiurlisbn}[1]{\iffieldundef{doi}{\iffieldundef{url}{#1
    }{\href{\thefield{url}}{#1}}}{\href{http://dx.doi.org/\thefield{doi}}{#1}}}

\DeclareFieldFormat [article,inbook,incollection,inproceedings,patent,thesis,unpublished]
  {title}{\usebibmacro{string+doiurlisbn}{#1}}

\newtheorem{theorem}{Theorem}
\newtheorem{corollary}[theorem]{Corollary}
\newtheorem{lemma}[theorem]{Lemma}

\newtheorem{fact}[theorem]{Fact}

\newtheorem{proposition}[theorem]{Proposition}
\theoremstyle{definition}
\newtheorem{definition}{Definition}

\newcommand{\opt}{\mathsf{opt}}
\newcommand{\Reduce}{\mathsf{RED}}
\newcommand{\reduce}{\mathsf{red}}
\newcommand{\OPT}{\mathsf{OPT}}
\newcommand{\ALG}{\mathsf{ALG}}

\newcommand{\hG}{\ensuremath{\bm\hat{G}}\xspace}

\newcommand{\cT}{\ensuremath{\mathcal{T}}\xspace}

\newcommand{\cThree}{\ensuremath{\mathcal{C}_3}\xspace}
\newcommand{\cFour}{\ensuremath{\mathcal{C}_4}\xspace}
\newcommand{\cFive}{\ensuremath{\mathcal{C}_5}\xspace}
\newcommand{\cSix}{\ensuremath{\mathcal{C}_6}\xspace}
\newcommand{\cSeven}{\ensuremath{\mathcal{C}_7}\xspace}

\newcommand{\typA}{\ensuremath{\mathsf{A}}\xspace}
\newcommand{\typB}{\ensuremath{\mathsf{B}}\xspace}
\newcommand{\typBi}{\ensuremath{\mathsf{B1}}\xspace}
\newcommand{\typBii}{\ensuremath{\mathsf{B2}}\xspace}
\newcommand{\typC}{\ensuremath{\mathsf{C}}\xspace}
\newcommand{\typCi}{\ensuremath{\mathsf{C1}}\xspace}
\newcommand{\typCii}{\ensuremath{\mathsf{C2}}\xspace}
\newcommand{\typCiisubba}{\ensuremath{\mathsf{C2(i)}}\xspace}
\newcommand{\typCiisubb}{\ensuremath{\mathsf{C2(ii)}}\xspace}
\newcommand{\typCiisubc}{\ensuremath{\mathsf{C2(iii)}}\xspace}
\newcommand{\typCiii}{\ensuremath{\mathsf{C3}}\xspace}
\newcommand{\typesThreeVC}{\cT_3}
\newcommand{\typesTwoVC}{\cT_2}

\DeclarePairedDelimiter{\abs}{\lvert}{\rvert}

\newcommand{\cost}{\mathrm{cost}\xspace}
\newcommand{\credit}{\mathrm{cr}\xspace}

\newcommand{\myrightarrow}[1]{\mathrel{\raisebox{-1.3pt}{$\xrightarrow{\mbox{\relsize{-3} $#1$}}$}}}

\title{Two-Edge Connectivity via Pac-Man Gluing}

\author{
Mohit Garg\thanks{Indian Institute of Science, Bengaluru, India. Supported by a Walmart fellowship.} \and 
Felix Hommelsheim\thanks{Faculty of Mathematics and Computer Science, University of Bremen, Germany.} \and 
Alexander Lindermayr\footnotemark[2]
}
\date{}

\begin{document}

\maketitle
\begin{abstract}
    We study the 2-edge-connected spanning subgraph (2-ECSS) problem: Given a graph $G$, compute a connected subgraph $H$ of $G$ with the minimum number of edges such that $H$ is spanning, i.e., $V(H) = V(G)$, and $H$ is 2-edge-connected, i.e., $H$ remains connected upon the deletion of any single edge, if such an $H$ exists. The $2$-ECSS problem is known to be NP-hard. In this work, we provide a polynomial-time $(\frac 5 4 + \varepsilon)$-approximation for the problem for an arbitrarily small $\varepsilon>0$, improving the previous best approximation ratio of $\frac{13}{10}+\varepsilon$.

    Our improvement is based on two main innovations: First, we reduce solving the problem on general graphs to solving it on structured graphs with high vertex connectivity. This high vertex connectivity ensures the existence of a 4-matching across any bipartition of the vertex set with at least 10 vertices in each part. Second, we exploit this property in a later gluing step, where isolated 2-edge connected components need to be merged without adding too many edges. Using the 4-matching property, we can repeatedly glue a huge component (containing at least 10 vertices) to other components. This step is reminiscent of the Pac-Man game, where a Pac-Man (a huge component) consumes all the dots (other components) as it moves through a maze. These two innovations lead to a significantly simpler algorithm and analysis for the gluing step compared to the previous best approximation algorithm, which required a long and tedious case analysis.
\end{abstract}

\thispagestyle{empty}

\newpage

\thispagestyle{empty}

\tableofcontents

\newpage

\setcounter{page}{1}

\section{Introduction}\label{sec:Introductin}
We study the 2-edge-connected spanning subgraph (2-ECSS) problem, which is a fundamental fault-tolerant network design problem. In the 2-ECSS problem, we are given a graph $G$, and the objective is to compute a connected subgraph $H$ of $G$ with the minimum number of edges such that $H$ is spanning (i.e., $V(H) = V(G)$) and $H$ is 2-edge-connected (i.e., $H$ remains connected upon the deletion of any single edge: for all edges $e\in E(H)$, $H \setminus \{e\}$ is connected), if such an $H$ exists. While computing a (1-edge) connected spanning subgraph with the minimum number of edges---a spanning tree---admits simple polynomial-time algorithms, the 2-ECSS problem is MAX-SNP hard~\cite{CL99,F98}. In particular, assuming $\mathrm{P} \neq \mathrm{NP}$, there is no PTAS for the 2-ECSS problem. Consequently, there has been a quest to obtain good polynomial-time approximations for the problem.

A $2$-approximation for the $2$-ECSS problem is easily achieved by augmenting a DFS tree by picking the highest back edges for each non-root vertex. The first non-trivial approximation achieving a ratio of $\nicefrac 3 2$ was provided by Khuller and Vishkin~\cite{KV94}. Subsequently, the approximation ratio was first improved to $\nicefrac{17}{12}$ by  \textcite{CSS01}, and later to $\nicefrac{4}{3}$ independently by  Hunkenschr\"{o}der, Vempala, and Vetta~\cite{VV00,HVV19} and \textcite{SV14} using very different techniques. Several failed attempts (see Section F of~\cite{GargGA23improved}) were made to improve the approximation ratio of $\nicefrac 4 3$, making it a natural approximation barrier. In a breakthrough, Garg, Grandoni, and Jabal Ameli~\cite{GargGA23improved} obtained an approximation ratio of $\nicefrac{118}{89}+\varepsilon<1.326$. Subsequently, Kobayashi and Noguchi~\cite{kobayashi2023approximation} observed that if a minimum $2$-edge cover that does not contain any isolated triangles can be computed in polynomial-time, then in fact the analysis of~\cite{GargGA23improved} will result in an approximation ratio of $1.3 + \varepsilon$. They were able to show that such a triangle-free $2$-edge cover is polynomial-time computable using a polynomial-time computable maximum triangle-free 2-matching, a result proved in Hartvigsen's PhD thesis~\cite{H84}, which was recently published in a journal~\cite{hartvigsen2024finding} and subsequently improved and simplified by Paluch~\cite{P23}. 
Additionally, PTASs for triangle-free $2$-matching are available~\cite{BGJ24ptas,KN24ptas} (with the proof by~\cite{KN24ptas} being only a few pages), which also result in a $(1.3+\varepsilon)$-approximation for the 2-ECSS problem.
\subsection{Our Result} In this work, we improve the approximation ratio for the 2-ECSS problem to $\nicefrac 5 4 + \varepsilon$. 

\begin{theorem}\label{thm:main}
    For all $\varepsilon>0$, there exists a deterministic polynomial-time approximation algorithm for the $2$-ECSS problem with an approximation ratio of $\nicefrac 5 4 +\varepsilon$.
\end{theorem}

The previous best approximation ratio of $1.3+\varepsilon$~\cite{kobayashi2023approximation} depends on~\cite{GargGA23improved}, a long and highly involved paper consisting of a tedious case analysis, making the result difficult to verify. In this work, we address this shortcoming by providing a simpler, shorter, and modular proof of a better approximation ratio of $\nicefrac 5 4 + \varepsilon$.

\subsection{Our Techniques}
Our proof is broadly based on the well-established template for such problems (as used in some form or another in~\cite{CDGKN20, CCDZ23, GHM23, VV00, HVV19, EFKN09, KN18lp, BGJ23}). It consists of four main steps: reducing the input instance to a structured instance, computing in polynomial time an object $\mathcal{O}$ (a minimum 2-edge cover with no isolated triangles) of the structured instance that is \emph{closely related} to and costs less than an optimal desired object (minimum 2-ECSS), then converting $\mathcal{O}$ into a feasible solution (2-ECSS) by first \emph{covering} the bridges in each connected component, followed by \emph{gluing} the various connected components together while keeping the additional cost of the conversion to a small factor $c$ of the cost of $\mathcal{O}$. This results in a $(1+c)$-approximation. We will elaborate on these steps in detail in the following sections.

In brief, for the first three steps, we closely follow~\cite{GargGA23improved}, which are uncomplicated parts of their paper, along with~\cite{kobayashi2023approximation} for the second step, which helps in eliminating isolated triangles, an important fact without which the gluing would be cumbersome. Our main insight in this part is to make the input even more structured than in~\cite{GargGA23improved} by getting rid of all \emph{large} 3-vertex cuts. This allows us to obtain a \emph{4-matching lemma}, which states that for any vertex partition of a structured input into two parts consisting of at least 10 vertices each, there exists a matching of size 4 going between the two parts. This lemma is then exploited in the gluing step. 

In the gluing step, we first glue a few components to obtain a \emph{huge} component (consisting of at least 10 vertices). Then, with the aid of the 4-matching lemma, we are repeatedly able to glue a huge component with the other components, while keeping $c$ to be \nicefrac{1}{4}. 
This step is reminiscent of the popular Pac-Man game, where a Pac-Man (a huge component) eats up all the dots (other components) as it moves through a maze (\emph{component graph}). An important advantage of our approach is that the gluing step, which has traditionally been considered the most complex and intricate part of the process (as seen in~\cite{GargGA23improved}), is now notably simplified, making it as straightforward as the other steps in the proof. Like in~\cite{GargGA23improved}, we lose an additional $\varepsilon$ approximation factor in the reduction step itself, giving us an approximation ratio of $1 + \nicefrac{1}{4} + \varepsilon = \nicefrac{5}{4} + \varepsilon$.

Most of the proofs for the first three steps have been revisited from~\cite{GargGA23improved} or are slight extensions, and they have been moved to the appendix for completeness. The primary new contribution in the appendix is the preprocessing step involving the elimination of large 3-vertex cuts.

In summary, our technique separates the technical complications of the gluing used in \cite{GargGA23improved} into two modular and independent components, the reduction to graphs without large 3-vertex cuts and the Pac-Man gluing.  We believe that the simplicity and modularity of our approach makes it applicable in and relevant for future work. This streamlined approach transforms the gluing step, once the most challenging part, into an intuitive and manageable procedure.

\subsection{Additional Related Work}
\paragraph{Survivable Network Design. }
The 2-ECSS problem is closely related to augmentation problems, where we are given a graph having a specific structure (e.g., spanning tree, matching, forest, cactus, or disjoint paths) along with a set of additional edges called links. The task is to add the minimum number of links to the structure so that the resulting graph becomes spanning and 2-edge connected. For the (spanning) tree augmentation problem, several better than $2$-approximations are known~\cite{A19,CTZ21,CG18,CG18a,CN13,EFKN09,FGKS18,GKZ18,KN16simple,KN18lp,N03,N21,TraubZ25}, with 1.393 as the currently best known approximation ratio
by \textcite{CTZ21}. For the forest augmentation problem, the first non-trivial approximation was obtained by \textcite{GJT22} with a ratio of 1.9973. For the matching augmentation problem the current best approximation ratio is $\nicefrac{13}8$~\cite{GHM23} by \citeauthor{GHM23}, while other non-trivial approximations are also known~\cite{BDS22,CDGKN20,CCDZ23}. 
All the above problems have a natural weighted version, and a general result of Jain~\cite{J01} provides a $2$-approximation for them. For the weighted 2-ECSS problem, obtaining an approximation ratio better than 2 is a major open problem. For the weighted tree augmentation problem the current best approximation ratio is $1.5+\varepsilon$ by Traub and Zenklusen~\cite{TraubZ25}. The k-ECSS problem for $k\geq 2$ is another natural generalization of the 2-ECSS problem (see for e.g.,~\cite{CT00,GG12}).

\paragraph{Graphic TSP. }
Another closely related problem is the Traveling Salesperson Problem (TSP). 
In metric TSP, one is given an undirected complete graph with weights on the edges that satisfy the triangle-inequality, and the goal is to find a tour through all vertices of minimum total weight.
In a recent breakthrough~\cite{karlin2021slightly}, Karlin, Klein, and Oveis Gharan improved the long-standing best approximation factor of $1.5$ by Christofides~\cite{christofides2022worst} to a $(1.5 - c)$-approximation, where $c > 10^{-36}$.
The algorithm makes use of the Help-Karp LP-relaxation of the problem, which has an integrality gap of $\nicefrac{4}{3}$, indicating
a potential answer for the best-possible approximation guarantee of metric TSP.
The example of this integrality gap is actually an instance of graphic TSP, which is a special case of metric TSP. 
Therein, we are given an undirected unweighted graph $G$ and one has to find an Eulerian multigraph while minimizing the number of edges.
Seb\"{o} and Vygen~\cite{SV14} proved the currently best approximation guarantee of $1.4$ for graphic TSP. They use very similar techniques in the $1.4$-approximation for graphic TSP and the $\nicefrac{4}{3}$-approximation for $2$-ECSS, which highlights the close connection between the two problems. Therefore, our new techniques for 2-ECSS could potentially contribute to further advancements in graphic TSP.

\paragraph{Concurrent Work.} Shortly after the first draft of this work was made publicly available, Bosch-Calvo, Grandoni, and Jabal Ameli~\cite{bosch20245} presented a different $\nicefrac{5}{4}$-approximation. While both works build on the framework of~\cite{GargGA23improved}, the approaches diverge substantially in their key innovations. In particular, \cite{bosch20245} removes the $\varepsilon$ additive term in the preprocessing phase by exploiting an observation used in earlier works~\cite{CDGKN20, CCDZ23, GHM23}.

In contrast, in our approach the $\varepsilon$ term is not eliminated. Instead, it is deliberately used to remove large 3-vertex cuts in the preprocessing phase, which is crucial for establishing the 4-matching lemma. Combined with our Pac-Man-style gluing procedure, this leads to a proof that is modular, conceptually simple, and straightforward to verify—rather than relying on the long, extensive, and intricate case-analysis-based gluing employed in~\cite{bosch20245}.

Subsequently, the two lines of work were merged in a joint paper~\cite{BGGHJL25}, obtaining a $\nicefrac{5}{4}$-approximation where our Pac-Man gluing technique is combined with the preprocessing of~\cite{bosch20245}, which significantly reduces the complexity of the gluing step. However, the 4-matching lemma from this work was not leveraged in the joint paper; incorporating it would further streamline the gluing as compared to~\cite{BGGHJL25}, though it introduces an additional $\varepsilon$ additive term in the approximation ratio.

\subsection{Organization of the Paper}
The paper proceeds as follows.
In \cref{sec:preliminaries} 
we start with some preliminaries consisting of some basic definitions and notations. 
Subsequently, we tackle the four main steps of our approach in the following four sections:
reduction to structured graphs in \cref{sec:structured-graphs}, obtaining a canonical 2-edge cover in \cref{sec:canonical},  bridge covering in \cref{sec:bridge-covering}, and the gluing step, which forms the core of this work, in
\Cref{sec:gluing}. For completeness, the appendix contains detailed proofs of the first three steps, most of which have been revisited from~\cite{GargGA23improved} or are mild generalizations. The main new contribution in the appendix is the elimination of large 3-vertex cuts in the preprocessing. 
Finally, we conclude in \Cref{sec:conclusion}, summarizing our contributions and discussing potential avenues for future research.

\section{Preliminaries}\label{sec:preliminaries}
We use standard graph theory notation. For a graph $G$, $V(G)$ and $E(G)$ refers to its vertex and edge sets, respectively. We use $\deg_G(v)$ to represent the degree of a vertex $v$ in the graph $G$ (when $G$ is clear from the context, we omit the subscript $G$). A pendant vertex in a graph is a vertex that is adjacent to exactly one other vertex. 
For a graph $G$, and a vertex set $S\subseteq V(G)$, $G[S]$ refers to the subgraph of $G$ induced on the vertex set $S$. We use components of a graph to refer to its maximal connected components. 
A cut vertex in a graph refers to a vertex in the graph whose deletion results in increasing the number of components of the graph. 
A $k$-vertex cut in a graph for $k\geq 1$ refers to a set of $k$ vertices such that deleting them results in increasing the number of components of the graph. 
A graph that does not contain any cut vertex is called 2-vertex connected (2VC). 
Similarly, a bridge in a graph refers to an edge in the graph whose deletion results in increasing the number of components in the graph. A connected graph without bridges is called 2-edge connected (in short, 2EC). 
We say a subgraph $H$ of a graph $G$ is spanning if $V(H)=V(G)$.

A spanning subgraph $H$ of a graph $G$ is called a 2-edge cover of $G$ if the degree of each vertex in $H$ is at least 2. A cycle refers to a simple cycle. 
We use $\mathcal{C}_i$ to represent a cycle on $i$ vertices, and $\cThree$ to refer to a triangle. 
We say a 2-edge cover is triangle-free if none of its components is a triangle.

Given a graph $G$ and a vertex set $S\subseteq V(G)$, $G|S$ denotes the contracted graph obtained from contracting the vertex set $S$ into a single vertex. Graph contraction may give rise to parallel edges and self-loops. The edges of $G$ and $G|S$ are in one-to-one correspondence. For a graph $G$ and a subgraph $H$, we use $G|H$ to denote $G|V(H)$.

For a graph $G$, we use $|G|$ to denote $|E(G)|$.
For a 2EC graph $G$ we use $\OPT(G)$ to denote a 2-ECSS of $G$ with the minimum number of edges. We define $\opt(G)\coloneq|\OPT(G)|$ (when $G$ is clear from the context we may just say $\opt$ and $\OPT$ instead). For a graph $G$, a subgraph $H$ is called a 2-ECSS of $G$ if it is 2EC and spanning.

\section{Reduction to Structured Graphs}\label{sec:structured-graphs}
In our approximation algorithm for the 2-ECSS problem, 
we will assume that the input is 2EC. Otherwise, it has no feasible solutions, which is easy to check: delete each edge and check if the graph is connected. 
Note that the input graph may contain parallel edges or self loops. 
In~\cite{GargGA23improved}, structured graphs were defined, and it was shown that, for each small enough $\varepsilon>0$, if the 2-ECSS problem can be approximated in polynomial time on structured graphs with an approximation ratio $\alpha\geq\nicefrac 6 5$, then it can be approximated in polynomial time on general graphs with only a $O(\varepsilon)$ loss in approximation, i.e., with an approximation ratio of $\alpha + O(\varepsilon)$.
In this work, we define structured graphs to be even more restricted than in~\cite{GargGA23improved}, and yet prove a similar approximation-preserving reduction that maintains  an approximation ratio of $\alpha + O(\varepsilon)$.

Informally, a structured graph is a simple graph that is 2VC with sufficiently many vertices and does not contain any contractible subgraphs, irrelevant edges, non-isolating 2-vertex cuts, or large 3-vertex cuts. Forbidding large 3-vertex cuts in the definition is the only difference compared to the definition in~\cite{GargGA23improved}. We now  define the various terms appearing in the definition of structured graphs before defining structured graphs formally.

\begin{definition}[$\alpha$-contractible subgraph~\cite{GargGA23improved}]
    Let $\alpha \geq 1$ be a fixed constant. A 2EC subgraph~$C$ of a 2EC graph $G$ is $\alpha$-\emph{contractible} if every 2-ECSS of $G$ contains at least $\frac{1}{\alpha}\abs{E(C)}$ edges with both endpoints in $V(C)$.
\end{definition}

The intuition behind defining $\alpha$-contractible subgraphs is the following: While designing an $\alpha$-approximation algorithm for the 2-ECSS problem, if we can find an $\alpha$-contractible subgraph $C$ in polynomial time, we can contract $C$ into a single vertex and find an $\alpha$-approximate solution on $G|C$, and add the edges of $C$ to get an $\alpha$-approximate solution for $G$. The approximation guarantee holds since 
$\alpha \cdot \opt \geq \alpha \cdot \opt(G|C) + |E(C)|$, which is true since $\opt(G) \geq \opt(G|C) + |\OPT[V(C)]|$ and $\OPT[V(C)]$ needs to contain at least $\nicefrac 1 \alpha \cdot \abs{E(C)}$ edges of $G[V(C)]$. Thus, in our approximation-preserving reduction all constant-sized $\alpha$-contractible subgraphs will be identified and handled.

\begin{definition}[irrelevant edge~\cite{GargGA23improved}]
    Given a graph $G$ and an edge $e = uv \in E(G)$, we say that $e$ is \emph{irrelevant} if $\{ u, v\}$ is a 2-vertex cut of $G$.
\end{definition}

\begin{definition}[non-isolating 2-vertex cut~\cite{GargGA23improved}]
    Given a graph $G$, and a 2-vertex cut $\{u, v\}$ of~$G$, we say that $\{u, v\}$ is \emph{isolating} if $G \setminus \{u, v\}$ has exactly two connected components, one of which is of size 1.
    Otherwise, it is \emph{non-isolating}.
\end{definition}

Note that getting rid of non-isolating cuts means that $G$ can be assumed to be almost 3-vertex-connected. 
In this work, we introduce the definitions of small and large 3-vertex cuts.

\begin{definition}[small and large 3-vertex cuts]
Given a graph $G$ and a 3-vertex cut $\{u, v, w\} \subseteq V(G)$ of $G$, we say that $\{u, v, w\}$ is \emph{small} if $G \setminus \{u, v, w\}$ has exactly two connected components such that one of them has at most 6 vertices. Otherwise, we call it \emph{large}.
\end{definition}

We are now ready to define structured graphs formally.

\begin{definition}
   [($\alpha, \varepsilon)$-structured graph]\label{def:structured}
    Given $\alpha \geq 1$ and $\varepsilon > 0$, a graph $G$ is $(\alpha, \varepsilon)$-\emph{structured} if it is simple, 2VC, it contains at least $\frac 4 \varepsilon$ vertices, and it does not contain
    \begin{enumerate}
        \item $\alpha$-contractible subgraphs of size at most $\frac 4 \varepsilon$,
        \item irrelevant edges,
        \item non-isolating 2-vertex cuts, and
        \item large 3-vertex cuts.
    \end{enumerate}
\end{definition}

In~\cite{GargGA23improved}, it was shown how to get rid of parallel edges, self loops, irrelevant edges, cut-vertices, small $\alpha$-contractible subgraphs, and non-isolating 2-vertex cuts via an approximation-preserving reduction. We show that we can also get rid of large 3-vertex cuts by generalizing their proof for non-isolating 2-vertex cuts in various ways. This is a significant part of this work. 
In \cref{app:prep}, we prove the following approximation-preserving reduction.
This proof is a bit long, but it admits a systematic case analysis. After understanding the main ideas through the first few cases, the remaining cases can be left as a simple exercise (though for completeness, we cover all the cases, including the ones proved in~\cite{GargGA23improved} for Properties 1-3).
\begin{restatable}{lemma}{lemmaReduction}\label{lem:reduction-to-structured}
    For all $\alpha \geq \frac{5}{4}$ and $\varepsilon \in (0, \frac{1}{24}]$, if there exists a deterministic polynomial-time $\alpha$-approximation algorithm for $2$-ECSS on ($\alpha, \varepsilon)$-structured graphs, then there exists a deterministic polynomial-time $(\alpha + 4 \varepsilon)$-approximation algorithm for $2$-ECSS.
\end{restatable}

When $\alpha$ and $\varepsilon$ are clear from the context, we simply write contractible and structured instead of $\alpha$-contractible and $(\alpha, \varepsilon)$-structured, respectively.

\smallskip\noindent\textbf{Benefits of structured graphs.}
Given the above reduction, we need to design approximation algorithms only for structured graphs. How do we exploit the fact that the input graph is structured? We will make generous use of the fact that there are no small contractible subgraphs. Furthermore, it was shown in~\cite{GargGA23improved} that certain
3-matchings exist in a structured graph, which will be very useful.

\begin{lemma}[$3$-Matching Lemma~\cite{GargGA23improved}]\label{lem:3-matching}
Let $G = (V, E)$ be a $2$VC simple graph without irrelevant edges and without non-isolating 2-vertex cuts. 
Consider any partition $(V_1, V_2)$ of $V$ such that for each $i \in \{1, 2\}$, $|V_i| \geq 3$, and if $|V_i| = 3$, then $G[V_i]$ is a triangle.
Then, there exists a matching of size $3$ between $V_1$ and $V_2$.
\end{lemma}

We further make use of the non-existence of large $3$-vertex cuts in structured graphs to establish a $4$-matching-lemma, which is going to be very useful in the gluing step. In fact, with similar techniques, we can also exclude large $k$-vertex cuts (for any constant $k$), where both sides of the partition have $\Omega(\frac{k}{1- \alpha})$ many vertices, leading to a similar $k$-Matching Lemma. However, we do not see an easy way to exploit this fact in improving the approximation ratio with the current techniques.

\begin{lemma}[4-Matching Lemma]\label{lem:4-matching}
Let $G$ be a  graph without large 3-vertex cuts.
Consider any partition $(V_1, V_2)$ of $V(G)$ such that for each $i \in \{1, 2\}$, $\abs{V_i} \geq 10$. 
Then, there exists a matching of size $4$ between $V_1$ and $V_2$ in $G$.
\end{lemma}

\begin{proof}
Consider the bipartite graph $F$ induced by the edges with exactly one endpoint in $V_1$ (and the other in $V_2$). 
For contradiction, we assume that a maximum matching $M$ in $F$ has size $\abs{M} \leq 3$. 
Then, by the K{\"o}nig-Egerv\'{a}ry Theorem (see, e.g.,~\cite{schrijver2003combinatorial}), there exists a vertex cover $C$ of $F$ of size $\abs{M}$.
Since $C$ is a vertex cover of $F$, there are no edges in $F$ (hence in $G$) between $U_1 \coloneqq V_1 \setminus C$ and $U_2 \coloneqq V_2 \setminus C$. 
Therefore, $C$ is a 3-vertex cut in $G$, which is large because $\abs{V_1}, \abs{V_2} \geq 10$ and $\abs{C} \leq 3$ imply $\abs{U_1}, \abs{U_2} \geq 7$; this is a contradiction. \end{proof}

\section{Canonical 2-Edge Cover}\label{sec:canonical}
It was shown in~\cite{kobayashi2023approximation} that a minimum triangle-free 2-edge cover can be computed in polynomial time using a polynomial-time algorithm for computing a maximum triangle-free 2-matching, e.g.~\cite{hartvigsen2024finding, P23}.

\begin{lemma}[Proposition 7 in~\cite{kobayashi2023approximation}]
\label{lem:triangle-free-2-edge-cover}
For a graph $G = (V, E)$, we can compute a triangle-free 2-edge cover of $G$ with
minimum cardinality in polynomial time if one exists.
\end{lemma}

Notice that for a structured graph $G$, a 2-ECSS of $G$ is also a triangle-free 2-edge cover because $G$ contains at least $\nicefrac{4}{\varepsilon}$ vertices, and thus, a 2-ECSS of $G$ cannot be a triangle.
Thus, if $H$ is a minimum triangle-free 2-edge cover of our input $G$, then $\abs{H}\leq \opt(G)$. 
Roughly speaking, we will compute $H$ using the above lemma, and then transform $H$ into a 2-ECSS of $G$ by adding at most $\nicefrac 1 4\cdot\abs{H}$ edges to it. Thus, we would get a 2-ECSS with at most $\nicefrac 5 4 \cdot\opt(G)$, proving our main result.
In reality, we might add some more edges and delete some edges, so that the net total of additional edges will be at most $\nicefrac 1 4\cdot\abs H$.

We will first convert $H$ into a \emph{canonical} triangle-free 2-edge cover without changing the number of edges in it. 
To define canonical, we first need the following definitions pertaining to a 2-edge cover.
\begin{definition}\label{def:complex}
Let $H$ be a $2$-edge cover of a structured graph $G$.
We call
$H$ \emph{bridgeless} if it contains no bridges, i.e., all the components of $H$ are 2EC. 
A component of $H$ is \emph{complex} if it contains a bridge.  
Given a complex component $C$, a maximal\footnote{w.r.t.\ the subgraph relationship} 2EC subgraph of $H$ containing at least 3 (as $G$ has no parallel edges) vertices are \emph{blocks}. 
Any vertex of $C$ not contained in a block is called \emph{lonely}.
A block $B$ of some complex component $C$ of $H$ is called \emph{pendant} if $C \setminus V(B)$ is connected. 
Otherwise, it is \emph{non-pendant}.
\end{definition}

Note that by contracting each block of a complex component to a distinct node, we obtain a tree whose leaves correspond to pendant blocks. We are now ready to define canonical 2-edge covers.

\begin{definition}[canonical]\label{def:canonicalD2}
A $2$-edge cover $H$ is \emph{canonical} if the following properties are satisfied:
\begin{itemize}
    \item Each non-complex component of $H$ is either a $\mathcal C_i$, for $4 \leq i \leq 7$, or contains at least $8$ edges.
    \item For every complex component, each of its pendant blocks contains at least $6$ edges and at least $6$ vertices and each of its non-pendant blocks contains at least $4$ edges and at least $4$ vertices.
\end{itemize}
\end{definition}

Note that a canonical 2-edge cover is also a triangle-free 2-edge cover.
We establish the following lemma in~\cref{app:canonical} that shows how to convert a minimum triangle-free 2-edge cover of a structured graph into a minimum canonical 2-edge cover in polynomial time.

\begin{restatable}{lemma}{lemmaCanonicalMain}
\label{lem:canonical-cover:main}
    Given a structured graph $G$, in polynomial time we can compute a canonical $2$-edge cover $H$ of $G$, where the number of edges in $H$ is at most (and hence equal to) the number of edges in a minimum triangle-free $2$-edge cover.
\end{restatable}
This definition of canonical and the proof of the above lemma are similar to the ones in~\cite{GargGA23improved}, except that now components with $7$ edges have to be a $\cSeven$.

\section{Bridge Covering and Credits}
\label{sec:bridge-covering}
We now have a structured graph $G$ and a minimum canonical 2-edge cover $H$ of $G$ such that $|H|\leq \opt(G)$. We will show how to convert $H$ into a 2-ECSS of $G$ by adding a net total of at most $\nicefrac 1 4\cdot\abs H$ edges to it. There are two reasons why $H$ is not already a 2-ECSS of $G$: First, $H$ might have multiple components. Second, some of these components might have bridges. We will first modify $H$ so that the resulting components do not have any bridges. Then, we will glue the various resulting 2EC components together to get a desired 2-ECSS, which we defer to the next section. 
Throughout this transformation, we keep the property that the intermediate solution is a canonical 2-edge cover.

To keep track of how many edges we are adding, we will define a credit invariant that is satisfied throughout the transformation and a cost function. 
We have a total budget of $\nicefrac 1 4 \cdot\abs H$, which we will distribute all over $H$ as credits of different kinds.
The cost of $H$ is defined keeping in mind that the cost of a single edge is 1. 
Then, $\cost(H)$ denotes the total number of edges in $H$ plus the total credits in $H$. 
Thus, initially the $\cost(H)$ will be at most $\nicefrac 5 4 \cdot\opt$, which will follow from Lemmas~\ref{lem:canonical-cover:main} and~\ref{credit-invariant-start}. 
At each step of our transformation, we aim to ensure that the cost never increases, so that the final solution will have cost at most $\nicefrac 5 4\cdot\opt$, and we will be done.
We now define our credit invariant and the cost function, which is also used in the next section.

\begin{definition}[credits]
    \label{def:credit}
    Let $H$ be a $2$-edge cover of $G$. We keep the following credit invariant for blocks, components and bridges:
    \begin{itemize}
        \item Each 2EC component $C$ of $H$ that is a $\mathcal C_i$, for $4 \leq i \leq 7$, receives credit $ \credit(C) = \frac{1}{4}|E(C)|$.
        \item Each 2EC component $C$ of $H$ that contains 8 or more edges receives credit $ \credit(C) = 2$.
        \item Each block $B$ of a complex component $C$ of $H$ receives a credit $\credit(B) =1$.
\item Each bridge $e$ of a complex component $C$ of $H$ receives credit $\credit(e) = \frac{1}{4}$.
        \item Each complex component $C$ of $H$ receives a component credit $\credit(C)=1$.
    \end{itemize}
\end{definition}

\begin{definition}[cost]
    \label{def:cost}
    The \emph{cost} of a $2$-edge cover $H$ of $G$ is defined as $\cost(H) = |H| + \credit(H)$. Here $\credit(H)$ denotes the sum total of all credits in $H$.
\end{definition}

We first show that we can distribute our budget of $\nicefrac 1 4\cdot \abs{H}$  over our minimum canonical triangle-free 2-edge cover to satisfy the credit invariant. In other words, we assume the credits are there as required by the credit invariant, and show that $\cost(H)\leq \nicefrac 5 4\cdot\abs H$. 

\begin{restatable}{lemma}{lemmaInitialCredit}\label{credit-invariant-start}
    Let $H$ be a canonical $2$-edge cover of a graph $G$. Then, we have $\cost(H) \leq \frac{5}{4} |H|$. 
\end{restatable}
\begin{proof}
    We show that $\cost(C) \leq \frac{5}{4} |C|$ for each component $C$ in $H$. 
    This then implies that $\cost(H) \leq \frac{5}{4} |H|$.
    
    First, consider a 2EC component $C$ that is a $\mathcal C_i$, for $4 \leq i \leq 7$.
    Since $C$ has exactly $i$ edges, $\cost(C) = |C| + \frac{1}{4}|E(C)| = \frac{5}{4} |C|$.
    Second, consider a 2EC component $C$ with at least 8 edges. 
    Note that $\frac{1}{4}|E(C)| \geq 2$. Hence, we have that $\cost(C) = |C| + 2 \leq |C| + \frac{1}{4}|E(C)| = \frac{5}{4} |C|$.

    Third, consider a complex component $C$. 
Each bridge $e$ of the complex component receives credit $\credit(e) = \frac{1}{4}$.
Each block receives a credit of $1$. Finally, the component credit of $C$ is 1. 
    Note that since $H$ is canonical, each pendant block has at least 6 edges, and non-pendant blocks have at least 4 edges. Moreover, there are at least 2 pendant blocks in $C$. 
    Let $c_p$ be the number of pendant blocks of $C$, $c_n$ be the number of non-pendant blocks of $C$, and $c_b$ be the number of bridges in $C$.

    Note that $\abs C\geq c_b+6c_p+4c_n$.
    Therefore, $\frac 1 4\abs C \geq \frac 1 4 c_b+\frac 3 2 c_p + c_n $. Since $c_p\geq 2$ implies $\frac 3 2 c_p \geq c_p + 1$, we have $\frac 1 4\abs C \geq\frac 1 4 c_b+ c_p + c_n+1$. Now, we plug this fact in the cost calculation:
    \[
    \cost(C) = |C| + \credit(C) = |C| + \frac{1}{4} c_b + c_p + c_n + 1 \leq |C| + \frac{1}{4} |C| \leq \frac{5}{4} |C|.
    \]
\end{proof}

We are now ready to state our main lemma for bridge covering, which essentially is already proved in~\cite{GargGA23improved}. 
Its proof is deferred to \cref{app:bridge-covering}.

\begin{restatable}{lemma}{lemmaBridgeCoveringMain}
    \label{lem:bridge-covering-start2}
    There is a polynomial-time algorithm that takes as input a canonical  2-edge cover $H$ of a structured graph $G$ and outputs a bridgeless canonical 2-edge cover $H'$ of $G$ such that $\cost(H') \leq \cost(H)$.
\end{restatable}

We will feed our minimum canonical 2-edge cover $H$ to the algorithm in the above lemma to obtain a bridgeless canonical 2-edge cover $H'$, and $\cost(H')\leq\cost(H)\leq \nicefrac 5 4\cdot\opt(G)$. Each component of $H'$ now is 2EC and either has $8$ or more edges, in which case it has a credit of $2$, or it is a cycle with $i$ edges ($\mathcal C_i$) for $4\leq i\leq 7$ with $\nicefrac 1 4\cdot i$ credits. We will use these credits  to glue the 2EC components of $H'$ together to obtain a desired 2-ECSS of $G$, as explained in the next section.

\section{Gluing}
\label{sec:gluing}
We have a structured graph $G$ as input and a bridgeless canonical 2-edge cover $H$ of $G$ with $\cost(H)\leq \nicefrac 5 4 \cdot \opt(G)$ such that $H$ satisfies the credit invariant:  $2$ in components with at least $8$ edges, and $\nicefrac i 4$ in components that are cycles with $i$ edges, for $4\leq i\leq 7$.
Now, in the final step of our algorithm, we glue the components of $H$ together (while maintaining the credit-invariant) using the credits available and a marginal additional cost of $3$. Specifically, we show the following.

\begin{lemma}
    \label{lem:gluing:main}  
    Given a structured graph $G$ and a bridgeless canonical 2-edge cover $H$ of $G$, we can transform $H$ in polynomial time into a 2-ECSS $H'$ of $G$ such that $\cost(H') \leq \cost(H) + 3$.
\end{lemma}

Using the above lemma, the proof of~\cref{thm:main} follows easily: 
\begin{proof}[Proof of \Cref{thm:main}]
For a structured graph $G$, we first compute a triangle-free 2-edge cover using \Cref{lem:triangle-free-2-edge-cover}, which can be turned into a canonical 2-edge cover $H$ without increasing the number of edges (\Cref{lem:canonical-cover:main}).
As observed before, we have $|H| \leq \opt$.
By \Cref{credit-invariant-start}, $\cost(H) \leq \frac{5}{4}|H|$ and therefore we have that $\cost(H) \leq \frac{5}{4} \opt$.
By \Cref{lem:bridge-covering-start2}, we obtain a bridgeless canonical 2-edge cover $H'$ such that $\cost(H') \leq \cost(H) \leq \frac{5}{4}\opt$.
Finally, we apply \Cref{lem:gluing:main} to turn $H'$ into a feasible solution $H''$ such that $\cost(H'') \leq \frac{5}{4}\opt + 3$.
Now $\credit(H'')=2$ as it will be a single 2EC component with more than $8$ edges. 
Thus, $\abs{H''} \leq \cost(H') + 3 -2 \leq \frac 5 4 \cdot \opt(G) + 1 \leq (\frac 5 4 + \frac {\varepsilon} {4})\cdot\opt(G)$, where the last inequality follows from $\opt(G)\geq\abs{V(G)}\geq \frac 4 \varepsilon$, which holds for structured graphs $G$. 
Now, using the approximation preserving reduction (\cref{lem:reduction-to-structured}) for general graphs $G$, we incur an additional loss of $4\varepsilon$. Thus, for any $\varepsilon' > 0$, we get an approximation ratio of 
$\frac 5 4 + \frac{\varepsilon}{4} + 4\varepsilon =  \frac 5 4 + \frac {17} 4\varepsilon = \frac 5 4  + \varepsilon'$ by setting $\varepsilon = \frac{4}{17} \varepsilon'$, which proves~\cref{thm:main}.
\end{proof}

It remains to prove \Cref{lem:gluing:main}. 
To this end, we first introduce some definitions. In what follows, unless otherwise mentioned, $G$ is a structured graph and $H$ is a bridgeless canonical 2-edge cover of~$G$.

\begin{definition}[large and huge]
We say that a component $C$ of $H$ is \emph{large} if $|V(C)| \geq 8$ and we say that it is \emph{huge} if $|V(C)| \geq 10$.
\end{definition}

Notice that huge components are also large and must have 2 credits. The benefit of having a huge component is that the 4-matching lemma can be applied to it to conclude that each huge component has a matching of size $4$ going out of it (whenever there are at least 10 vertices outside the huge component).

\begin{definition}[component graph]
The \emph{component graph} $\hG_H$ 
    w.r.t.\ $G$ and $H$ is the multigraph obtained from  $G$ by contracting the vertices of each 2EC component of $H$ into a single node. 
    We call the vertices of $\hG_H$ \emph{nodes} to distinguish them from the vertices in $G$. There is a one-to-one correspondence between the edges of $\hG_H$ and $G$. Furthermore, there is a one-to-one correspondence between the nodes of $\hG_H$ and the components of $H$. We will denote the component corresponding to a node $A$ by $C_A$. \end{definition}

Although not strictly required, to remove any ambiguities in our presentation later while dealing with cycles, we get rid of all the self-loops that appear in the component graph, and call the resulting graph the component graph. Now,
observe that the component graph $\hat G_H$ is 2EC as $G$ is 2EC (as contracting a set of vertices in a 2EC graph results in a 2EC graph; also deleting self loops from a 2EC graph results in a 2EC graph).

\begin{definition}[non-trivial and trivial segments]
A \emph{non-trivial segment}\footnote{In the graph theory literature, blocks often refer to maximal 2VC subgraphs, but in this work we use the term blocks to refer to maximal 2EC subgraphs in line with previous works on 2-edge connectivity. Thus, 2VC blocks consisting of at least 3 vertices and non-trivial segments defined above are precisely the same objects.} $S$ of $\hG_H$ is a maximal 2-node-connected subgraph of $\hG_H$ consisting of at least 3 nodes. A \emph{trivial segment} is a single node of $\hG_H$ that is not part of any non-trivial segment.
For a segment $S$, $V(S)$ refers to the set of all the vertices in the various components corresponding to the nodes contained in $S$.
\end{definition}

Notice that we can break the component graph into segments such that two segments can potentially overlap only at a cut node (follows from Algorithm 4.1.23 in~\cite{West2001}; alternatively see~\cite{KorteV2002}).

Since a trivial segment $A$ is not part of any cycles of length at least 3 in $\hG_H$ (otherwise $A$ would be part of a non-trivial segment), $A$ has to be either a cut node of $\hG_H$ or a pendant node (neighbor of exactly one other node) of $\hG_H$. 
In either case, we can apply the 3-matching lemma (\cref{lem:3-matching}) to argue that there is a matching of size 3 between $V(C_A)$ and $V(C_B)$, where node $B$ is a neighbor of node $A$ in $\hG_H$. The condition of the lemma is satisfied by partitioning $V(G)$ into two parts such that one part has vertices from $V(C_A)$ and the other part has vertices from $V(C_B)$ and the only edges going across the parts are between $V(C_A)$ and $V(C_B)$. We have the following propositions.
\begin{proposition}\label{prop:gluing:cut-pendent}
    Let $A$ be a trivial segment of $\hG_H$. Then, $A$ is either a cut node of $\hG_H$ or it is a pendant node of $\hG_H$.
\end{proposition}
\begin{proposition}\label{prop:gluing:trivial-3-matching}
    Let $A$ and $B$ be adjacent nodes in $\hG_H$ such that $A$ is a trivial segment. Then, there exists a matching of size $3$ going between $V(C_A)$ and $V(C_B)$ in $G$.
\end{proposition}

Having set up the requisite definitions, we now outline the steps involved in the gluing process. First, if $H$ does not contain a huge component, then we first glue a few components together, so that we obtain a canonical bridgeless 2-edge cover $H'$ that has a huge component, while maintaining the credit invariant. 
We will show that such an $H'$ can be obtained with $\cost(H')\leq \cost(H)+3$. 
After this step, we assume $H$ has a huge component, an invariant that will be maintained at each iteration of our algorithm. 
Next, let $L$ be a huge component in $H$. 
We will distinguish two cases: either $L$ is a trivial segment of $\hG_H$, or it is part of a non-trivial segment of $\hG_H$. In either case, we show that we can make progress, i.e., we can find in polynomial time a canonical bridgeless 2-edge cover $H'$ that has fewer components than $H$ such that $H'$ contains a huge component and $\cost(H') \leq \cost(H)$. 
Repeatedly
applying the above step will result in a single bridgeless 2EC component, a 2-ECSS of $G$. Thus, if we can establish the existence of these steps,~\cref{lem:gluing:main} follows immediately. These steps are encapsulated in the following three lemmas.

\begin{lemma}
    \label{lem:gluing:obtaining-huge-component}
    If H does not contain a huge component, then we can compute in polynomial time a canonical bridgeless 2-edge cover $H'$ such that $H'$ contains a huge component and $\cost(H') \leq \cost(H) + 3$.
\end{lemma}

\begin{lemma}
    \label{lem:gluing:trivial-segment}
    Let $C_L$ be a huge component such that $L$ is a trivial segment of $\hG_H$ and $H\neq C_L$. Then, in polynomial time we can compute a canonical bridgeless 2-edge cover $H'$ of $G$ such that $H'$ has fewer components than $H$, $H'$ contains a huge component, and $\cost(H') \leq \cost(H)$.
   \end{lemma}

\begin{lemma}
    \label{lem:gluing:non-trivial-segment}
    Let $C_L$ be a huge component such that $L$ is part of a non-trivial segment of $\hG_H$. Then, in polynomial time we can compute a canonical bridgeless 2-edge cover $H'$ of $H$ such that $H'$ has fewer components than $H$, $H'$ contains a huge component, and $\cost(H') \leq \cost(H)$.
\end{lemma}

As explained before, the above three lemmas immediately imply~\cref{lem:gluing:main}. All that remains is to prove the above three lemmas, which we do in the next three subsections.

\subsection{Proof of~\cref{lem:gluing:obtaining-huge-component}}
\begin{proof} 
Without loss of generality, $H$ has at least two components; if $H$ consists of a single component, it has to be huge as $\abs{V(H)}=\abs{V(G)}\geq \frac 4 \varepsilon \geq 10$. 
Let $A$ be a node in $\hG_H$.
    Let $K$ be some cycle in $\hG_H$ containing $A$, which must exist (and can be easily found) as $\hG_H$ is 2EC and has at least 2 nodes.

    Let $H_1\coloneq H\cup K$ (i.e., add the edges of the cycle $K$ to $H$). Observe that $H_1$ is a canonical bridgeless 2-edge cover with fewer components than $H$, and the newly created component, call it $C_B$, containing the vertices of $C_A$ is large ($\geq 8$ vertices: as each component of $H$ contains at least $4$ vertices and $|K|\geq 2$). Furthermore, $|H_1| = |H| + |K|$ and $\credit(H_1) \leq \credit(H) - |K| + 2$, because every component in $K$ has a credit of at least $1$ in $H$ and the newly created large component receives a credit of $2$. 
    If $H_1$ consists of only a single component, namely $C_B$, then $C_B$ must be huge and we are immediately done.

    Otherwise, arguing as before, we can find another cycle $K'$ through $B$ in $\hG_{H_1}$ and add $K'$ to $H_1$ to obtain $H_2$, i.e., $H_2 \coloneq H \cup K \cup K'$.
    Observe that $H_2$ is again a canonical bridgeless 2-edge cover and the newly formed component (containing the vertices of $B$) has at least 12 vertices, and is huge.
    We have $|H_2| = |H| + |K| + |K'|$ and $\credit(H_2) \leq 
    \credit(H_1) -(|K'|+1) + 2$ (since every component in $H_1$ has a credit of at least 1,  $\credit(B)=2$, and the newly formed component receives 2 credits). 
    Plugging the upper bound on $\credit(H_1)$ from above, we have $\credit(H_2) \leq \credit(H) -|K|+2 -(|K'|+1) + 2=\credit(H) -|K|-|K'|+3$. Thus, $\cost(H_2)=|H_2|+\credit(H_2)\leq|H| + |K| + |K'|+\credit(H) -|K|-|K'|+3=|H|+\credit(H)+3=\cost(H)+3$, which completes the proof.
\end{proof}

\subsection{Proof of~\cref{lem:gluing:trivial-segment}}

\begin{figure}[tb]
    \centering
    \begin{subfigure}[t]{0.31\textwidth}
        \centering
        \begin{tikzpicture}[scale=0.5]
            \node[comp] (l) at (0,0) {$L$};
            \node[vx] (v1) at (5,1) {};
            \node[vx] (v2) at (4,1) {};
            \node[vx] (v3) at (4,0) {};
            \node[vx] (v4) at (4,-1) {};
            \node[vx] (v5) at (5,-1) {};
            \node[vx] (v6) at (6,-1) {};
            \node[vx] (v7) at (6,1) {};

            \node at (5,0) {$A$};

            \draw[edge] (v1) -- (v2) -- (v3) -- (v4) -- (v5) -- (v6) -- (v7) -- (v1);

            \draw[ledge,buy] (l) edge[bend right=50] (v6);
            \draw[ledge,buy] (l) edge[bend left=50] (v7);
            \draw[ledge] (l) edge[bend right=10] (v4);
            \draw[ledge] (l) edge[bend left=10] (v2);
            
            \draw[edge,sell] (v6) -- (v7);

            \node[comp] (b) at (6,3) {$B$};
            \draw[ledge] (v7) edge[bend left=0] (b);

            \node[comp] (c1) at (-2,1.5) {};
            \node[comp] (c2) at (-2,-1.5) {};

             \draw[ledge] (l) edge[bend left=20] (c1);
             \draw[ledge] (l) -- (c1);
             \draw[ledge] (l) edge[bend right=20] (c1);

             \draw[ledge] (l) edge[bend left=20] (c2);
             \draw[ledge] (l) -- (c2);
             \draw[ledge] (l) edge[bend right=20] (c2);
            
        \end{tikzpicture}
        \caption{$C_A$ is a $\cSeven$ and there is a 4-matching between vertices of $C_L$ and $C_A$.}
        \label{fig:trivial-case-1}
    \end{subfigure}
    \hfill
    \begin{subfigure}[t]{0.31\textwidth}
        \centering
        \begin{tikzpicture}[scale=0.5]
            \node[comp] (l) at (0,0) {$L$};
            \node[vx] (v1) at (5,1) {};
            \node[vx] (v2) at (4,1) {};
            \node[vx] (v3) at (4,0) {};
            \node[vx] (v4) at (4,-1) {};
            \node[vx] (v5) at (5,-1) {};
            \node[vx] (v6) at (6,-1) {};
            \node[vx] (v7) at (6,1) {};

            \node at (5,0) {$A$};

            \draw[edge] (v1) -- (v2) -- (v3) -- (v4) -- (v5) -- (v6) -- (v7) -- (v1);

            \draw[ledge,buy] (l) edge[bend left=50] (v7);
            \draw[ledge] (l) edge[bend right=10] (v4);
            \draw[ledge,buy] (l) edge[bend left=10] (v2);
            \draw[ledge,buy] (v1) -- (v3);
            
            \draw[edge,sell] (v2) -- (v3);
            \draw[edge,sell] (v1) -- (v7);

            \node[comp] (c1) at (-2,1.5) {};
            \node[comp] (c2) at (-2,-1.5) {};

             \draw[ledge] (l) edge[bend left=20] (c1);
             \draw[ledge] (l) -- (c1);
             \draw[ledge] (l) edge[bend right=20] (c1);

             \draw[ledge] (l) edge[bend left=20] (c2);
             \draw[ledge] (l) -- (c2);
             \draw[ledge] (l) edge[bend right=20] (c2);
            
        \end{tikzpicture}
        \caption{$C_A$ is a $\cSeven$ and there is a nice diagonal in $C_A$.}
        \label{fig:trivial-case-2}
    \end{subfigure}
    \hfill
    \begin{subfigure}[t]{0.31\textwidth}
        \centering
        \begin{tikzpicture}[scale=0.5]
            \node[comp] (l) at (0,0) {$L$};
            \node[vx] (v1) at (5,1) {};
            \node[vx] (v2) at (4,1) {};
            \node[vx] (v3) at (4,0) {};
            \node[vx] (v4) at (4,-1) {};
            \node[vx] (v5) at (5,-1) {};
            \node[vx] (v6) at (6,-1) {};
            \node[vx] (v7) at (6,1) {};

            \node at (6.5,0) {$A$};

            \draw[edge] (v1) -- (v2) -- (v3) -- (v4) -- (v5) -- (v6) -- (v7) -- (v1);

            \draw[ledge] (l) edge[bend left=50] (v7);
            \draw[ledge] (l) edge[bend right=10] (v4);
            \draw[ledge] (l) edge[bend left=10] (v2);
           
            \draw[ledge,lightblue] (v6) -- (v1)  (v6) -- (v3);
           \draw[edge,lightblue] (v2) -- (v3)
            (v1) -- (v7)
            (v4) -- (v5) -- (v6)
            ;

            \node[comp] (c1) at (-2,1.5) {};
            \node[comp] (c2) at (-2,-1.5) {};

             \draw[ledge] (l) edge[bend left=20] (c1);
             \draw[ledge] (l) -- (c1);
             \draw[ledge] (l) edge[bend right=20] (c1);

             \draw[ledge] (l) edge[bend left=20] (c2);
             \draw[ledge] (l) -- (c2);
             \draw[ledge] (l) edge[bend right=20] (c2);
            
        \end{tikzpicture}
         \caption{$C_A$ is a $\cSeven$ and $C_A$ is contractible as every solution must buy at least $6$ edges of $G[V(C_A)]$ (e.g., the edges in blue).}
         \label{fig:trivial-case-3}
    \end{subfigure}
    \caption{Illustrations of different cases considered in the proof of \Cref{lem:gluing:trivial-segment}. 
    Nodes of $\hG_H$ are depicted by a circle with black border, and vertices of $G$ are depicted by a black, filled circle.
    Edges of $H$ are black and solid, edges of $G \setminus H$ are gray and dashed, edges of $G \setminus H$ that are added to $H$ to obtain $H'$ are red, and edges that are removed from $H$ to obtain $H'$ are green. }
    \label{fig:trivial-segment}
\end{figure}

\begin{proof}
Let $C_L$ be a huge component such that the node $L$ is a trivial segment of $\hG_H$ and $H\neq C_L$. Our goal is to construct a canonical bridgeless 2-edge-cover $H'$ of $G$ with fewer components compared to $H$ such that $H'$ contains a huge component and $\cost(H')\leq\cost(H)$.

Since $H\neq C_L$, $\hG_H$ has at least 2 nodes, and as $\hG_H$ is 2EC, $L$ has a neighbor, say $A$.
By \Cref{prop:gluing:trivial-3-matching}, we know that there must be a 3-matching $M_3$ between $C_L$ and $C_A$.
First, note that if $C_A$ is a large component, we simply add any two edges of $M_3$ to $H$ to obtain $H'$, and it can be easily checked that $H'$ meets our goal as $\cost(H') - \cost(H) = (|H'| - |H|) + (\credit(H') - \credit(H))  \leq 2 + (2-2-2)\leq 0$ (we add 2 edges, $\credit(C_L)=\credit(C_A)=2$, and, in $H'$, $C_L$ and $C_A$ combine to form a single component of credit $2$). Thus, from now on, we assume $C_A$ is $\mathcal C_i$ for some $i\in\{4,5,6,7\}$.

\textbf{Condition ($\ast$):}
If there are two distinct edges between $L$ and $A$, say $e_1$ and $e_2$, that are incident to two distinct vertices in $C_A$, say $u$ and $v$, such that there is a Hamiltonian Path $P$ of $G[V(C_A)]$ from $u$ to $v$, then we set $H' \coloneq (H \setminus E(C_A)) \cup \{e_1, e_2\} \cup P$, which is clearly a canonical bridgeless 2-edge-cover of $G$ with fewer components than $H$. Furthermore, $\cost(H') - \cost(H) = (|H'| - |H|) + (\credit(H') - \credit(H))  = (2 - 1) + (2-2-\credit(C_A))\leq 0$, since we added two edges $e_1$ and $e_2$, effectively removed 1 edge as $|C_A| - |P|=1$, $\credit(C_L)=2$ and $\credit(C_A)\geq 1$ in $H$, whereas in $H'$ those combine to form a single component with credit 2. Thus, $\cost(H') \leq \cost(H)$, meeting our goal. 

We now show that Condition ($\ast$) always holds. If $C_A$ is a $\cFour$ or a $\cFive$, then in either case, there must be two distinct edges in the matching $M_3$, say $e_1$ and $e_2$, whose endpoints in $C_A$ are adjacent via an edge $e$ in $C_A$. Thus, Condition ($\ast$) holds by setting $P=E(C_A)\setminus \{e\}$. Otherwise,  $C_A$ is a $\cSix$ or $\cSeven$. If $A$ is not a pendant node in $\hG_H$ (i.e., it has other neighbors other than $A$), then, \Cref{lem:4-matching} is applicable, implying a 4-matching $M_4$ between $C_A$ and $C_L$.
This is because we can partition $V(G)$ in two parts to have $V(C_L)$ in one part and $V(C_A)\cup V(C_B)$ in another part (see~\cref{prop:gluing:cut-pendent}) such that the only edges crossing the parts are between $V(C_L)$ and $V(C_A)$, and $|V(C_A)| + |V(C_B)| \geq 6+4=10$ and $|V(C_L)| \geq 10$, because $L$ is huge. Arguing as before, there must exist edges $e_1$ and $e_2$ whose endpoints in $C_A$ are adjacent via an edge $e$ in $C_A$. Again, Condition ($\ast$) holds by setting $P= E(C_A)\setminus \{e\}$. An illustration of this case is given in \Cref{fig:trivial-case-1}.

Finally, we consider the case that $C_A$ is a $\cSix$ or $\cSeven$ and $A$ is pendant node of $\hG_H$, i.e., the only neighbor it has is $L$. Assuming that Condition ($\ast$) does not hold, we prove that $C_A$ is contractible, which contradicts the fact that $G$ is structured. Let the vertices of the cycle $C_A$ be labeled as  $v_1 - v_2 - v_3 - v_4 - v_5 - v_6 - v_1$ if $C_A$ is a $\cSix$ or $v_1 - v_2 - v_3 - v_4 - v_5 - v_6 - v_7 - v_1$ if $C_A$ is a $C_7$.
    Since Condition ($\ast$) does not hold and there is the 3-matching $M_3$ between $C_L$ and $C_A$, $C_L$ is adjacent to exactly three mutually non-adjacent vertices of $C_A$ via $M_3$, say w.l.o.g.\ $\{v_1, v_3, v_5\}$ (up to symmetry this is the only possibility in both the $\cSix$ and $\cSeven$ cases).
We say an edge in $G[V(C_A)]$ is a nice diagonal if both its endpoints are in $\{v_2,v_4,v_6\}$. 
    Therefore, the set of nice diagonals are $\{ v_2 v_4, v_4 v_6, v_6 v_1 \}$.
If there exists a nice diagonal $e$ in $G[V(C_A)]$,
    then observe that there are two distinct edges in $M_3$ that are incident to say $u$ and $v$ in $C_A$, and there is a Hamiltonian Path $P\subseteq \{e\}\cup E(C_A)$ from $u$ to $v$ in $G[V(C_A)]$. For example, if the diagonal $v_2 v_4$ exists in $G$ and $C_A$ is $\cSix$, then the desired Hamiltonian path is $v_3 -v_2-v_4 -v_5 -v_6- v_1$ (for $\cSix$, up to symmetry, this is the only case; the $\cSeven$ case is left to the reader to verify, a simple exercise). Hence, Condition ($\ast$) holds, a contradiction. An illustration of this case is given in \Cref{fig:trivial-case-2}.
    Otherwise, if none of these nice diagonals are in $G$, we can conclude that $C_A$ is contractible: The three sets of edges of $G$, namely, edges incident on $v_2$, edges incident on $v_4$, and edges incident on $v_6$ are all inside $G[V(C_A)]$ and mutually disjoint. 
    As any 2-ECSS of $G$ must contain at least $2$ edges incident on each vertex, it has to contain at least $6$ edges from $G[V(C_A)]$, and hence $C_A$ is contractible. An illustration of this case is given in \Cref{fig:trivial-case-3}.
\end{proof}

\subsection{Proof of~\cref{lem:gluing:non-trivial-segment}}

We will make use of the following lemma, which we will repeatedly apply to glue a huge component with all the $\cFour$ and $\cFive$ components 
(and potentially also larger components)
in a non-trivial segment.

\begin{lemma}
    \label{lem:gluing:cycle-through-large}
    Let $S$ be a non-trivial segment of $\hG_H$ containing nodes $L$ and $A$ such that $C_L$ is huge and $C_A$ is a $\cFour$ or $\cFive$. 
    We can compute in polynomial time a cycle $K$ in $\hG_H$ that passes through the nodes $L$ and $A$, where the two edges of the cycle $K$ entering the node $A$ are incident on two distinct vertices $u$ and $v$ of $C_A$ in $H$, and at least one of the following two conditions holds.
    \begin{enumerate}[(a)]
        \item We can find a Hamiltonian Path from $u$ to $v$ in $G[V(C_A)]$. 
        \item We can find a node $D$ of $\hG_H$ that is not in the cycle $K$ and a set of edges $F \subseteq E(G)$ disjoint to the edges of the cycle $K$ such that $|F| = |E(C_A)| + |E(C_D)|$ and $F \cup \{uv\}$ forms a 2-ECSS of $G[V(C_A)\cup V(C_D)]$.
    \end{enumerate}
\end{lemma}

Using the above lemma, we prove~\cref{lem:gluing:non-trivial-segment}, and then prove~\cref{lem:gluing:cycle-through-large}.

\begin{proof}[Proof of \Cref{lem:gluing:non-trivial-segment}]
    We are given a non-trivial segment $S$ of $\hG_H$ which has a node $L$ such that $C_L$ is a huge component. 
    Our goal is to construct a canonical bridgeless 2-edge-cover $H'$ of $G$ with fewer components compared to $H$ such that $H'$ contains a huge component and $\cost(H')\leq\cost(H)$.

    If $S$ also contains a node $A$ corresponding to a $\cFour$ or a $\cFive$, then we apply \Cref{lem:gluing:cycle-through-large} to obtain a cycle $K$ through $L$ and $A$ such that $K$ is incident to the vertices $u$ and $v$ of $C_A$. Now either Case (a) or Case (b) in the lemma statement holds.
    
    If Case~(a) holds, then there is a Hamiltonian Path $P$ from $u$ to $v$ in $G[V(C_A)]$, and we set $H' \coloneq (H \setminus E(C_A)) \cup K \cup P$.
    Notice that $H'$ is a canonical bridgeless 2-edge cover and has fewer components compared to $H$; in $H'$ we have created a cycle using $K$ and $P$ that passes through some 2EC components (including $L$) and all the vertices of $C_A$ using the Hamiltonian path $P$.
    Now, observe that $|H'| \leq |H| + |K|-1$, and as each component of $H$ has a credit of at least 1, the large component corresponding to $L$ has credit $2$, and the newly created huge component in $H'$ receives a credit of 2, we have $\credit(H') \leq \credit(H) - (|K| - 1 + 2) + 2=\credit(H)-|K|+1$. Therefore, $\cost(H')=|H'| + \credit(H') \leq  |H| +  |K|-1 +\credit(H)-|K|+1=|H|+\credit(H)=\cost(H)$.
    
    Otherwise, Case~(b) holds, and there exist a node $D$ which is not in the cycle $K$, and a set of edges $F$ disjoint to $K$, such that $|F| = |E(C_A)| + |E(C_D)|$ and $F \cup \{uv\}$ forms a 2-ECSS of $G[V(C_A) \cup V(C_D)]$. Here, we set $H' \coloneq (H \setminus (E(C_A) \cup E(C_D))) \cup K \cup F$. Like before, $H'$ is bridgeless and has fewer components compared to $H$ for the following reason. From the condition of Case~(b), the vertices of $C_A$ and $C_D$ remain 2EC in $H'$ with the edges of $F$ and an additional edge $uv$; in $H'$ the work of the edge $uv$ is instead done by a $u-v$ path induced by the cycle $K$. Also, the cycle $K$ 2-edge connects the various 2EC components (including $L$) with the vertices of $C_A$ and $C_D$ with the help of $F$. Thus, no bridges are introduced.   
    Furthermore, observe that $|H'| \leq |H| + |K|$ as $|F| =  |E(C_A)| + |E(C_D)|$, and since each component in $H$ has at least a credit of 1, the huge component corresponding to $L$ has credit $2$, $C_D$ is not a component of $H'$ and was not on $K$, and the newly created huge component receives credit 2,
    we have $\credit(H') \leq \credit(H) - (|K| - 1 + 2 + 1)  + 2=\credit(H) - |K|$. Therefore, $\cost(H')=|H'| + \credit(H') \leq  |H| + |K| +\credit(H) - |K|= |H| + \credit(H) = \cost(H)$.

    Hence, we can assume that $S$ does not contain a node corresponding to a $\cFour$ or a $\cFive$.
    Thus, all nodes of $S$ correspond to components that are either large, a $\cSix$ or a $\cSeven$. Temporarily, we drop the parallel edges from the segment $S$ in the component graph $\hG_H$ keeping precisely one edge between adjacent nodes in $S$. Observe that $S$ continues to be 2-node connected even after dropping the parallel edges. Now, the smallest cycle through $L$ in $S$ in the component graph has length at least $3$ as there are no parallel edges.
    Compute any cycle $K$ of length at least $3$ through $L$ in the segment $S$ and set $H' \coloneq H \cup K$. $H'$ is clearly a canonical bridgeless 2-edge cover of $G$ with fewer components.  We have $|H'| = |H| + |K|$ and since
    all components incident on the cycle $K$ have a credit of at least $\frac 3 2$ and $L$ has a credit of $2$, and they merge into a single huge component in $H'$ which receives credit 2, we have $\credit(H') \leq \credit(H) - (\frac{3}{2}(|K|-1) + 2) + 2=\credit(H) - \frac{3}{2}|K|+\frac 3 2 \leq \credit(H) - |K|$, as $|K| \geq 3$ implies $\frac {|K|}2 \geq \frac3 2$. Combining the above facts, we have $\cost(H')=|H'| + \credit(H') 
    \leq|H| + |K| + \credit(H) -|K| \leq |H| + \credit(H)=\cost(H)$.
\end{proof}

\subsubsection{Proof of Lemma~\ref{lem:gluing:cycle-through-large}}
Only the proof of \Cref{lem:gluing:cycle-through-large} is left.
We will make use of the following results in our proof.

\begin{proposition}[Theorem 1 in~\cite{kawarabayashi2008improved}]
\label{prop:cycle-through-specified-edges}
Given a graph $G$ and an edge set $F \subseteq E(G)$ of constant size, we can find a cycle in $G$ that contains all edges of $F$ in polynomial
time if one exists.
\end{proposition}

\begin{proposition}[Lemma D.8 in~\cite{GargGA23improved}; reformulated]\label{lem:gluing:c4-c5}
    Let $G$ be a graph and $H$ be a subgraph of $G$.
    Let $C$ and $C'$ be two 2EC components of $H$ such that $C$ is a $\cFive$ and $C'$ is a $\cFour$, and there is a $3$-matching between $C$ and $C'$ in $G$. 
    Then, for any two distinct vertices $x,y \in V(C)$ there exists a set $F$ of $9$ edges in $G$ such that $F \cup \{xy\}$ induces a 2-ECSS on $V(C) \cup V(C')$ in $G$.
\end{proposition}

Furthermore, the following simple lemma showing the existence of certain $3$-matchings within a non-trivial segment will also be used in the proof.

\begin{lemma}
    \label{lem:gluing:3-matching-segment}
    Let $S$ be a non-trivial segment of $\hG_H$ and $A$ be a node in $S$. Then, there exists a matching of size 3 in $G$ between $V(C_A)$ and $V(S) \setminus V(C_A)$.
\end{lemma}

\begin{proof}
If $A$ is not a cut node in $\hG_H$, then, by the definition of a segment, there are no edges between $V(C_A)$ and the vertices of components which are not contained in $S$. Thus, in this case the claim follows from \Cref{lem:3-matching} using $(V(C_A),V(G)\setminus V(C_A))$ as the partition of $V(G)$.
Otherwise, $A$ is a cut node in $\hG_H$, which implies $C_A$ is a separator of $G$ with the vertices of $V(S)\setminus V(C_A)$  contained on the same side of the separator. Consequently, in this case, our claim follows by applying~\Cref{lem:3-matching} on a partition of $V(G)$ into two parts: one side of the separator that contains the vertices of $V(S)\setminus V(C_A)$  and the other part contains the vertices from $V(C_A)$ along with the other side of the separator such that the only edges crossing the two parts are between $V(S)\setminus V(C_A)$ and $V(C_A)$.
\end{proof}

We are now ready to prove \Cref{lem:gluing:cycle-through-large}.

\begin{proof}[Proof of \Cref{lem:gluing:cycle-through-large}]
We are given a non-trivial segment $S$ of $G_H$ that contains a huge node $L$ and a node $A$ such that $C_A$ is either a $\cFour$ or a $\cFive$. 
We need to construct in polynomial time a cycle $K$ in $G_H$ through $L$ and $A$, such that the two cycle edges incident on $C_A$ are incident on two distinct vertices $u$ and $v$ in $C_A$ such that either (a) we can find a $u-v$ Hamiltonian path in $G[V(C_A)]$ or (b) we can find a node $D$ outside the cycle $K$ and a set of edges $F$ disjoint to $E(K)$ such that $|F| = |E(C_A)| + |E(C_D)|$ and $F \cup \{uv\}$ forms a 2-ECSS of $G[V(C_A)\cup V(C_D)]$.

In the proof, we will show the existence of the cycle $K$ and a desired Hamiltonian path or a desired $D$ and $F$ with $|F|=9$. The polynomial-time constructions for these objects then easily follow; $K$ from the algorithm in~\cref{prop:cycle-through-specified-edges} by feeding it all possible $4$ edges, $2$ incident on $L$ and $2$ incident on $A$, and then the Hamiltonian path or $F$ can be found by brute-force enumeration over a constant domain.

The cycle $C_A$ will be labeled as $v_1 - v_2 - v_3 - v_4 - v_1$ if it is a $\cFour$ or $v_1 - v_2 - v_3 - v_4 - v_5 - v_1$ if it is a $\cFive$.
Since $S$ is a segment containing $L$ and $A$, it is 2-node-connected, and thus, there exist two internally node-disjoint paths in $S$ between $L$ and $A$.
Moreover, \Cref{lem:gluing:3-matching-segment} implies that 
at least three distinct vertices of $V(C_A)$ have incident edges in $G$ from three distinct vertices in $V(S)\setminus V(C_A)$. We consider a general condition first that allows us to find certain cycles in $S$ between the nodes $L$ and $A$. This condition will show up multiple times in our analysis.

  \textbf{Condition ($\star$)} Let $P_1$ and $P_2$ be two internally node-disjoint paths from $L$ to $A$ that are incident on vertices $x$ and $y$ of $C_A$, respectively. Additionally, let $z$ be a vertex in $C_A$ that is adjacent to some node $A'\neq A$ in the segment $S$. We can now find the following cycles in $S$ that pass through the nodes $L$ and $A$. Note that $x$, $y$, and $z$ may not be all distinct.
  \begin{itemize}
\item Since $A'$ is part of the segment $S$, there are two internally node-disjoint paths from $A'$ to $L$. At least one of these paths will not include the node $A$, say $P'$. 
      \item \textbf{$(\star 1$)} If $P'$ is node-disjoint from $P_1$, except at $L$, then there is a cycle that intersects $C_A$ at $x$ and $z$: $z\myrightarrow{}A'\myrightarrow{P'}L\myrightarrow{P_1}x$. Similarly, if $P'$ is node-disjoint from $P_2$, except at $L$, then there is a cycle that intersects $C_A$ at $y$ and $z$.
      \item Otherwise, if $P'$ is not node-disjoint from $P_1$, let $B$ be the first node when traversing $P'$ from $A'$ to $L$ that intersects $P_1$ or $P_2$. Note that $B$ cannot be on both $P_1$ and $P_2$ as they are internally node-disjoint.  
      \item \textbf{$(\star 2)$} If $B$ is on $P_1$, then we have the cycle that intersects $C_A$ at $y$ and $z$: $z\myrightarrow{}A'\myrightarrow{P'}B\myrightarrow{P_1}L\myrightarrow{P_2}y$. If $B$ is on $P_2$, similarly we get a cycle that intersects $C_A$ at $x$ and $z$.
      \item No matter which case we are in, there is always a cycle intersecting $C_A$ at $z$ and ($x$ or $y$).
  \end{itemize}

As mentioned above, \Cref{lem:gluing:3-matching-segment} implies that 
at least three distinct vertices of $V(C_A)$ have incident edges in $G$ from three distinct vertices in $V(S)\setminus V(C_A)$.
We consider two exhaustive cases based on which set  of vertices of $C_A$ has such incident edges: Either (\textbf{Case 1}) each of $v_1,v_2,v_3$  or (\textbf{Case 2}) each of $v_1,v_3,v_4$ has an incident edge from the vertices of $V(S)\setminus V(C_A)$. 
Observe up to symmetry these are the only cases,  and for $C_A$ being a $\cFour$ considering only Case~1 is sufficient. Thus, in Case~2, we will assume $C_A$ is a $\cFive$.

\textbf{Case 1:} $v_1$, $v_2$, and $v_3$ have incident edges from vertices of $V(S)\setminus V(C_A)$.

\textbf{Case 1.1:}
There exists two internally node-disjoint paths $P_1$ and $P_2$ from $A$ to $L$ that are incident to distinct vertices $u$ and $v$ of $C_A$ such that there exists a Hamiltonian Path between $u$ and $v$ in $G[V(C_A)]$. Then, $K=P_1\cup P_2$ forms a cycle in $\hG_H$ through $A$ and $L$ as $P_1$ and $P_2$ are internally node-disjoint, which satisfies the requirements of (a) in the lemma statement. \newcommand{\dist}{\mathsf{dist_{C_A}}}

\textbf{Case 1.2:}
Otherwise, if we are not in Case 1.1, there must exist two internally node-disjoint paths $P_1$ and $P_2$ from $A$ to $L$ such that $P_1, P_2$ are incident to $V(C_A)$ at vertices $u_1, u_2 \in V(C_A)$, respectively ($u_1 = u_2$ is possible, but $u_1$ is not a neighbor of $u_2$). 
Since, $C_A$ is a cycle consisting of at most 5 vertices, we can always find $w_i\in \{v_1,v_2,v_3\}$ adjacent to $u_i$ in $C_A$ for each $i\in \{1,2\}$. If possible, set $w_1=w_2$.
In particular, if $u_1=u_2$, we can always set $w_1=w_2\in\{v_1,v_2,v_3\}$. Furthermore, if $C_A=\cFour$, $u_1\neq u_2$, and $u_1$ not adjacent to $u_2$ (otherwise we would be in Case 1.1) implies they are diagonally opposite, and we can again set $w_1=w_2\in \{v_1,v_2,v_3\}$.
Consider two subcases:

\textbf{Case 1.2.1: $u_1 = u_2$ or $w_1 = w_2$.}
This implies that both $u_1$ and $u_2$ are adjacent to $w_1$. 
Since $w_1\in\{v_1,v_2,v_3\}$ it is adjacent to a node $A'\neq A$ in $S$. 
We now apply Condition ($\star$) by setting $x= u_1$, $y=v_1$, and $z=w_1$ to find a cycle through $L$ and $A$ that intersects $C_A$ at $w_1$ and ($u_1$ or $u_2$), and since $w_1$ is adjacent to both $u_1$ and $u_2$, we can find a desired Hamiltonian path in $G[V(C_A)]$ satisfying condition (a) of the lemma statement.

\textbf{Case 1.2.2: $u_1 \neq u_2$ and $w_1 \neq w_2$.}
Since we are not in the previous case, we only need to consider $C_A=C_5$ and $\dist(u_1=u_2)=2$. Observe that in this case we can always set $w_1$ and $w_2$ to be adjacent vertices.
Recall that $w_i\in \{v_1,v_2,v_3\}$ for each $i\in \{1,2\}$ and since we are in Case 1, both $w_1$ and $w_2$ have incident edges from vertices of $V(S) \setminus V(C_A)$.
Let $A_{i} \neq A$ be the node in $S$ that is adjacent to $w_i$, for each $i \in \{1,2\}$.
Since $S$ is 2-node connected, there must be a path $P'_i$ in $S$ from $A_i$ to $L$ which does not contain $A$, for $i \in \{1,2\}$.

If $P'_1$ is node-disjoint from $P_1$ or if $P'_{2}$ is node-disjoint from $P_2$, then we can find a desired cycle using Condition ($\star 1$) that intersects $C_A$ at adjacent vertices $u_i$ and $w_i$ for some $i\in\{1,2\}$ which have a Hamiltonian path between them in $C_A$. 
Otherwise, for each $i\in\{1,2\}$, let $B_i$ be the first node on $P'_i$ when traversing $P'_i$ from $A_i$ to $L$ that is also on either $P_1$ or on $P_2$.
Now if $B_1$ is on $P_2$ or $B_2$ is on $P_1$, we reach the same conclusion using Condition ($\star 2$).

Otherwise, the paths $P'_1$ and $P'_2$ are such that $B_1$ is on $P_1$ and $B_2$ is on $P_2$, and, thus,
the subpaths $A_1 \myrightarrow{P'_1} B_1$ and $A_2 \myrightarrow{P'_2} B_2$ 
must be node-disjoint; if not, let $t$ be the first node on $A_1 \myrightarrow{P'_1} B_1$ that is also on $A_2 \myrightarrow{P'_2} B_2$. We will modify $P'_1$ to be $A_1 \myrightarrow{P'_1} t \myrightarrow{P'_2} B_2 \myrightarrow {P'_2}L$, which now has its first intersecting node with $P_1$ or $P_2$ to be $B_2$, which is on $P_2$, and we are back to the case considered above.
But then there exist two-node disjoint paths from $A$ to $L$ in $\hG_H$ that are incident to $w_1$ and $w_2$: $L \myrightarrow{P_i} B_i \myrightarrow{P'_i} A_i - w_i$ for $i \in \{1,2\}$, and there is a Hamiltonian Path from $w_1$ to $w_2$ in $G[V(C_A)]$ as $w_1$ and $w_2$ are adjacent, meaning that we are actually in Case~1.1.

\textbf{Case 2:}
$v_1$, $v_3$, and $v_4$ have incident edges from vertices of $V(S)\setminus V(C_A)$.
As discussed above, this case can only apply if $A$ corresponds to a $\cFive$. Further, note that neither $v_2$ nor $v_5$ can have an incident edge to a node in $S$ other than $A$, as otherwise we are also in Case~1 by simply relabeling the $v_i$'s cyclically so that $v_1,v_2,v_3$ have such incident edges.
We have two subcases based on whether $v_2$ or $v_5$ has an incident edge to a component outside of segment $S$. 

\textbf{Case 2.1:}
Neither $v_2$ nor $v_5$ has incident edges to nodes of $\hG_H$ other than $A$.
Then, there must be an edge between $v_2$ and $v_5$ in $G$, as otherwise $C_A$ is contractible:
This is true since then the set of edges incident to $v_2$ and the set of edges incident to $v_5$ are all inside $G[V(C_A)]$ and mutually disjoint. Since any 2-ECSS of $G$ must include 2 edges incident on each vertex, it has to contain at least 4 edges from $G[V(C_A)]$, and hence $C_A$ is contractible.
Like in Case 1.1, if there exists internally node-disjoint paths $P_1$ and $P_2$  that are incident to vertices $u$ and $v$ in $C_A$ such that there exists a Hamiltonian Path between $u$ and $v$ in $G[V(C_A)]$, then $K =P_1 \cup P_2$ forms the desired cycle in $\hG_H$ through $A$ and $L$ satisfying the requirements of Condition (a) in the lemma statement. As a result, observe that both $P_1$ and $P_2$ must be incident to the same vertex $v \in \{v_1, v_3, v_4\}$.
This is true since neither $P_1$  nor $P_2$ can be incident on $v_2$ or $v_5$. 
Furthermore, the edge $v_2v_5$ exists helping in forming the required Hamiltonian paths; a straightforward check. For example, $v_1 - v_2 - v_5 - v_4 - v_3$ is a Hamiltonian $v_1$-$v_3$-path.
Now, If $v = v_1$ (a similar argument works if $v = v_3$ or $v=v_4$), let $A'\neq A$ be the node adjacent to $v_3$ ($v_1$ if $v = v_3$ or $v=v_4$). Now from Condition ($\star$), by setting $x=y=v$, and $z=v_3$, we get the desired cycle passing through $A$ and $L$ that intersects $C_A$ at $v_1$ and $v_3$ having a Hamiltonian path between them, satisfying Condition (a) of the lemma statement.

\textbf{Case 2.2:}
There is an edge between $v_2$ (or $v_5$ by symmetry) and a node $D$ of $\hG_H$ that is a trivial segment $S'$ or lies in a non-trivial segment $S'$, and $S' \neq S$. Note that $A$ must be a cut node in $\hG_H$ in this case. 
Observe that
$|V(S)| - |V(C_A)| \geq 10$, because $S$ contains the node $L$ ($C_L$ is huge with 10 vertices) other than $A$.
If $|V(S') \cup V(C_A)| \geq 10$, \Cref{lem:4-matching} guarantees the existence of a $4$-matching $M_4$ between $V(S) \setminus V(C_A)$ and $V(C_A)$ (as $C_A$ is a separator, $V(G)$ can be partitioned in 2 parts with $V(C_A)\cup V(S')$ on one side and $V(S)\setminus V(C_A)$ on the other, so that the only edges crossing are between $V(C_A)$ and $V(S)\setminus V(C_A)$).
But this means that we are actually in Case~1 as some three consecutive vertices of $C_A$ must be incident to the other nodes of $S$. 

Thus, the only possibility left is $|V(S') \cup V(C_A)| \leq 9$. Since $A$ corresponds to a $\cFive$ as we are in Case~2, we conclude that $S'$ is a trivial segment that corresponds to a single node $D$ where $C_D$ is a $\cFour$. 
Note that $|E(C_A)| + |E(C_D)| = 9$.
Since $A$ is a cut node and $D$ is a trivial segment, \Cref{prop:gluing:trivial-3-matching} implies that there exists a $3$-matching $M_3$ between $V(C_A)$ and $V(C_D)$ in $G$. 
Therefore, \Cref{lem:gluing:c4-c5} is applicable to  $A$ and $D$ with any two distinct vertices $u$ and $v$ in $C_A$,
which implies that there exists a set $F$ of edges of $G$ with $|F| = 9$ such that $F \cup \{uv\}$ induces a 2-ECSS of $G[V(C_A) \cup V(C_D)]$. If $P_1$ and $P_2$, the two internally node-disjoint paths between $L$ and $A$, are incident on two distinct vertices in $C_A$, then we set those two distinct vertices to $u$ and $v$, set $K:=P_1 \cup P_2$, which together with $F$ satisfies the Condition (b) in the lemma statement.
Otherwise, $P_1$ and $P_2$ are incident to the same vertex of $V(C_A)$, which we call $u$, and let $A' \neq A$ be a node adjacent to $v$, where $v \in\{v_1,v_3,v_4\}\setminus\{u\}$. 
Now from Condition ($\star$), there exists a cycle in $S$ passing through $A$ and $L$ that intersects $C_A$ at $u$ and $v$. 
This cycle together with $F$ satisfies Condition (b) in the lemma statement.
\end{proof}

\section{Conclusion}
\label{sec:conclusion}

In this work, we presented a new approach to the 2-edge-connected spanning subgraph (2-ECSS) problem, achieving an improved polynomial-time approximation ratio of 
$\nicefrac54 + \varepsilon$. 
This result follows a long line of research efforts aimed at improving the approximation ratio, building on prior methods while simplifying the complexity of the algorithm. Our method significantly streamlined the gluing step, which has traditionally been one of the most complex parts of such algorithms. By introducing a preprocessing phase to eliminate large 3-vertex cuts and leveraging a novel 4-matching lemma, we modularized the solution into manageable components. This led to a method that is not only simpler but also more intuitive than previous approaches.

The modularity and clarity of our approach opens new avenues for future research, providing a solid foundation for adaptations and further advancements. The combination of the 4-matching lemma and the Pac-Man-inspired gluing strategy we proposed could also prove useful in tackling other network design problems that are hindered by complex structural challenges.

\printbibliography

\newpage

\appendix

\section{Reduction to Structured Graphs}\label{app:prep}

This section is dedicated to the proof of \Cref{lem:reduction-to-structured}, which we first restate.

\lemmaReduction*

We assume throughout this section that $\alpha \geq \frac{5}{4}$ and $\varepsilon \in (0, \frac{1}{24}]$.
Let $\ALG$ be a polynomial-time algorithm that, given an $(\alpha,\varepsilon)$-structured graph $G' = (V',E')$, returns an $\alpha$-approximate 2EC spanning subgraph $\ALG(G') \subseteq E'$ of $G'$.
In the following, we define a reduction algorithm which recursively modifies and partitions the input graph $G$ until we obtain $(\alpha,\varepsilon)$-structured graphs, then solves these with $\ALG$, and finally puts the resulting solutions together to an overall solution for $G$. 
In particular, our reduction checks for structures that are forbidden in $(\alpha,\varepsilon)$-structured graphs. These are graphs of constant size (w.r.t.\ $\varepsilon$), 
parallel edges and self loops,
1-vertex cuts, small $\alpha$-contractible subgraphs (w.r.t.\ $\varepsilon$), irrelevant edges, non-isolating 2-vertex cuts and large 3-vertex cuts (cf.\ \Cref{def:structured}). 
For all structures except large 3-vertex cuts, we can use the reduction presented in~\cite{GargGA23improved}. The main contribution of this section is to show how to reduce to graphs without large 3-vertex cuts, given that all other substructures have already been removed.

This section is structured as follows. In \cref{sec:reduce-defs} we first introduce definitions to classify solutions regarding 2-vertex cuts and 3-vertex cuts, and then give the formal description of the reduction algorithm, which is split into three algorithm environments, where \cref{alg:reduce} is the main algorithm, which calls \cref{alg:2-vertex-cuts,alg:3-vertex-cuts} as subroutines.
In \cref{sec:reduce-auxiliary}, we give auxiliary results, mainly regarding the removal of 2-vertex cuts and 3-vertex cuts. Finally, in \cref{sec:reduce-proofs} we combine these auxiliary results together and give the main proof of \cref{lem:reduction-to-structured}.

\subsection{Definitions and the Algorithm}\label{sec:reduce-defs}

We will heavily use the following two facts on 2EC graphs with contracted components. We will not refer to them explicitly.

\begin{fact}\label{fact:contract-2ec}
    Let $G$ be a 2EC graph and $W \subseteq V(G)$. Then $G | W$ is 2EC.
\end{fact}

\begin{fact}\label{fact:decontract}
    Let $H$ be a 2EC subgraph of a 2EC graph $G$, and $S$ and $S'$ be 2EC spanning subgraphs of $H$ and $G|H$, respectively. Then $S \cup S'$ is a 2EC spanning subgraph of $G$, where we interpret $S'$ as edges of $G$ after decontracting $V(H)$.
\end{fact}

\subsubsection{Definitions for Non-Isolating 2-Vertex Cuts}

Let $\{u,v\}$ be a $2$-vertex cut in a graph $G$. 
Note that that we can find a partition $(V_1,V_2)$ of $V \setminus \{u,v\}$ such that $V_1 \neq \emptyset \neq V_2$ and there are no edges between $V_1$ and $V_2$. Once we delete $\{u,v\}$ from $G$, the graph breaks into various components such that each is either part of $V_1$ or $V_2$. Similarly, $\OPT(G)$ breaks into various components once we remove $\{u,v\}$ such that each component is either part of $V_1$ or $V_2$. In the following definition, we consider different types of graphs that match $H \coloneq \OPT(G)[V_1 \cup \{u,v\}]$ (or analogously $\OPT(G)[V_2 \cup \{u,v\}]$).

\begin{definition}[2-vertex cut solution types]\label{def:2vc-types}
Let $H$ be a graph and $u,v \in V(H)$. Let $H'$ be obtained by contracting each 2EC component $C$ of $H$ into a single super-node $C$, and let $C(x)$ be the corresponding super-node of the 2EC component of $H$ that contains $x \in V(H)$. We define the following types of $H'$ with respect to $\{u,v\}$:
\begin{description}
\item[Type $\typA$:] $H'$ is composed of a single super-node $H' = C(u) = C(v)$.
\item[Type $\typB$:] $H'$ is a $C(u) - C(v)$ path of length at least $1$.
\item[Type $\typC$:] $H'$ consists of two isolated super-nodes $C(u)$ and $C(v)$.
\end{description}
Let $\typesTwoVC \coloneq \{\typA,\typB,\typC \}$ be the set of 2-vertex cut optimal solution types.
\end{definition}

\begin{definition}[Order on types]\label{def:2vc-ties}
$\typA \succ \typB \succ \typC$.
\end{definition}

\subsubsection{Definitions for Large 3-Vertex Cuts}

Similarly to 2-vertex cuts and \Cref{def:2vc-types}, we give in the following definition types that match the subgraph of an optimal solution that is induced by the partition corresponding to a large $3$-vertex cut $\{u,v,w\}$ in $G$. Naturally, here are more patterns that can occur compared to 2-vertex cuts.

\begin{definition}[3-vertex cut solution types]
Let $H$ be a graph and $u,v,w \in V(H)$. Let $H'$ be obtained by contracting each 2EC component $C$ of $H$ into a single super-node $C$, and let $C(x)$ be the corresponding super-node of the 2EC component of $H$ that contains $x \in V(H)$. We define the following types of $H'$ with respect to $\{u,v,w\}$:
\begin{description}
\item[Type $\typA$:] $H'$ is composed of a single super-node $H' = C(u) = C(v) = C(w)$.
\item[Type $\typBi$:] $H'$ is a $C_1 - C_k$ path of length at least $1$ such that $u,v,w \in C_1 \cup C_k$ and $\abs{\{u,v,w\} \cap C_i} \leq 2$ for $i \in \{1,k\}$. We assume w.l.o.g.\ that $C_1 = C(u) = C(v)$.
\item[Type $\typBii$:] $H'$ is composed of two isolated super-nodes $C_1$ and $C_2$, where $\abs{\{u,v,w\} \cap C_i} \leq 2$ for $i \in \{1,2\}$. We assume w.l.o.g.\ that $C_1 = C(u) = C(v)$.
\item[Type $\typCi$:] $H'$ is a tree and $u$, $v$, and $w$ are in distinct super-nodes. \item[Type $\typCii$:] $H'$ is composed of a $C_1-C_k$ path of length at least $1$ and an isolated super-node $C_\ell$ such that $\abs{C_i \cap \{u,v,w\}} = 1$ for $i \in \{1,k,\ell\}$.
\item[Type $\typCiii$:] $H'$ is composed of the three isolated super-nodes $C(u), C(v)$ and $C(w)$.
\end{description}
Let $\typesThreeVC \coloneq \{\typA,\typBi,\typBii,\typCi,\typCii,\typCiii\}$ be the set of 3-vertex cut optimal solution types.
\end{definition}

\begin{definition}[Order on types]\label{def:3vc-ties}
$\typA \succ \typBi \succ \typBii \succ \typCi \succ \typCii \succ \typCiii$.
\end{definition}

The following proposition lists compatible solution types of both sides of a 3-vertex cut $\{u,v,w\}$. The proof is straightforward and, thus, omitted. 
For example, if $\OPT_1$ (using the notation introduced below) is of type $\typCiii$, it is composed of three distinct 2EC components $C(u)$, $C(v)$, and $C(w)$. Therefore, $u$, $v$, and $w$ must be in the same 2EC connected component in $\OPT_2$ as otherwise $\OPT$ cannot be a 2-ECSS, and thus, $\OPT_2$ must be of type $\typA$.

\begin{proposition}[feasible type combinations]\label{lemma:3vc-type-combinations}
Let $\{u,v,w\}$ be a large 3-vertex cut in a graph $G$, and $(V_1,V_2)$ be a partition of $V \setminus \{u,v,w\}$ such that $7 \leq \abs{V_1} \leq \abs{V_2}$ and there are no edges between $V_1$ and $V_2$ in $G$. 
Let $G_1 = G[V_1 \cup \{u,v,w\}]$, $G_2 = G[V_2 \cup \{u,v,w\}] \setminus \{ uv, vw, uw \}$ and, for a fixed optimal solution $\OPT(G)$, let $\OPT_i = \OPT(G) \cap E(G_i)$ for $i \in \{1,2\}$. Then, with respect to $\{u,v,w\}$, if
\begin{enumerate}[(a)]
    \item $\OPT_1$ is of type $\typA$, $\OPT_2$ must be of type $t \in \{\typA,\typBi,\typBii,\typCi,\typCii,\typCiii\} = \typesThreeVC$.
    \item $\OPT_1$ is of type $\typBi$, $\OPT_2$ must be of type $t \in \{\typA,\typBi,\typBii,\typCi,\typCii\}$.
    \item $\OPT_1$ is of type $\typBii$, $\OPT_2$ must be of type $t \in \{\typA,\typBi,\typBii\}$.
    \item $\OPT_1$ is of type $\typCi$, $\OPT_2$ must be of type $t \in \{\typA,\typBi,\typCi\}$.
    \item $\OPT_1$ is of type $\typCii$, $\OPT_2$ must be of type $t \in \{\typA,\typBi\}$.
    \item $\OPT_1$ is of type $\typCiii$, $\OPT_2$ must be of type $\typA$.
\end{enumerate}
\end{proposition}

Observe that according to the above definition we have $\abs{\OPT(G)} = \abs{\OPT_1} + \abs{\OPT_2}$.

\subsubsection{The Full Algorithm}

We are now ready to state our main algorithm $\Reduce$ (\Cref{alg:reduce}). 
Let $\reduce(G) \coloneq |\Reduce(G)|$ denote the size of the solution computed by $\Reduce$ for a given 2-ECSS instance $G$.
The reductions applied in \Cref{reduce:1vc,reduce:loop,reduce:contractible,reduce:irrelevant,reduce:2vc} (here and in the following, we do not distinguish between the line where the actual reduction is computed and the line in which the condition for the reduction is checked) are identical to the reductions given in \cite{GargGA23improved}. 
Specifically, \Cref{alg:2-vertex-cuts}, which is called in \cref{reduce:2vc}, reduces  non-isolating 2-vertex cuts and was first given in \cite{GargGA23improved}. We emphasize that we present it slightly different, but without changing its actual behavior.
When \Cref{alg:reduce} reaches \Cref{reduce:3vc} it calls \Cref{alg:3-vertex-cuts} to handle a large 3-vertex cut $\{u,v,w\}$. 

Whenever the reduction returns edges $K$ of contracted graphs $G|H$, i.e., $K=\Reduce(G|H)$, we interpret these as the corresponding edges after decontracting $H$, e.g., in \Cref{alg:reduce} in \Cref{reduce:contractible} or in \Cref{alg:3-vertex-cuts} in \Cref{reduce-3vc:both-large:recurse}.

Hence, from now on let us focus on \Cref{alg:3-vertex-cuts}.
In \Cref{app:prep:helper-3vc} we prove that the claimed sets of edges $F$ in \Cref{alg:3-vertex-cuts} exist and are of constant size.
The computations done in \Cref{reduce-3vc:compute-opt-types} (and in \cref{reduce-2vc:compute-opt-types} in \cref{alg:2-vertex-cuts}) can be done as explained in \cite{GargGA23improved}. 
That is, we first compute via enumeration minimum-size subgraphs $\overline{\OPT}_1^t$ of $G_1$ of every type $t \in \typesThreeVC$ if such a solution exists, because $G_1$ is of constant size in this case. 
If $\overline{\OPT}_1^t$ does not exist, $\OPT_1^t$ also cannot exist. 
Otherwise, that is, $\overline{\OPT}_1^t$ exists, we check whether $G_2$ is a graph that admits a solution that is compatible with $\overline{\OPT}_1^t$ according to \Cref{lemma:3vc-type-combinations}, and, if that is the case, we set $\OPT_1^t \coloneq \overline{\OPT}_1^t$.

Finally, note that if the algorithm reaches \Cref{reduce:call-alg} in \Cref{alg:reduce}, then $G$ must be $(\alpha,\varepsilon)$-structured.

\begin{algorithm}[tb]
\small
\DontPrintSemicolon
\caption{ $\Reduce(G)$: reduction to $(\alpha,\varepsilon)$-structured graphs}\label{alg:reduce}
\KwIn{2EC graph $G=(V,E)$.}
\If{$|V(G)| \leq \frac{4}{\varepsilon}$}{
compute $\OPT(G)$ via enumeration and \Return $\OPT(G)$\label{reduce:bruteforce}
}
\If{$G$ has a 1-vertex cut $\{v\}$\label{reduce:1vc}}{
let $(V_1,V_2)$, be a partition of $V \setminus \{v\}$ such that $V_1 \neq \emptyset \neq V_2$ and there are no edges between $V_1$ and~$V_2$. \;
\Return $\Reduce(G[V_1 \cup \{v\}]) \cup \Reduce(G[V_2 \cup \{v\}])$. \;
}
\If{$G$ contains a self loop or a parallel edge $e$\label{reduce:loop}}{
\Return $\Reduce(G \setminus \{e\})$. \;
}
\If{$G$ contains an $\alpha$-contractible subgraph $H$ with $|V(H)| \leq \frac{4}{\varepsilon}$\label{reduce:contractible}}{
\Return $H \cup \Reduce(G \mid H)$.
}
\If{$G$ contains an irrelevant edge $e$\label{reduce:irrelevant}}{
\Return $\Reduce(G \setminus \{e\})$. \;
}
\If{$G$ contains a non-isolating 2-vertex cut $\{u,v\}$\label{reduce:2vc}}{
Execute \Cref{alg:2-vertex-cuts} for $\{u,v\}$. \;
}
\If{$G$ contains a large 3-vertex cut $\{u,v,w\}$\label{reduce:3vc}}{
Execute \Cref{alg:3-vertex-cuts} for $\{u,v,w\}$.
}
\Return $\ALG(G)$.\label{reduce:call-alg}
\end{algorithm}

\begin{algorithm}
\small    
\DontPrintSemicolon
\caption{Remove a non-isolating 2-vertex cut \cite{GargGA23improved}}\label{alg:2-vertex-cuts}
\KwIn{A 2EC graph $G=(V,E)$ without cut vertices and without $\alpha$-contractible subgraphs with at most $\frac{4}{\varepsilon}$ vertices. A non-isolating 2-vertex cut $\{u,v\}$ in $G$.}
Let $(V_1,V_2)$ be a partition of $V \setminus \{u,v\}$ such that $2 \leq \abs{V_1} \leq \abs{V_2}$ and there are no edges between $V_1$ and $V_2$ in $G$. \;
Let $G_1 \coloneq G[V_1 \cup \{u,v\}]$ and $G_2 \coloneq G[V_2 \cup \{u,v\}] \setminus \{uv\}$.\;
Let $G'_i \coloneq G_i | \{u,v\}$ for $i \in \{1,2\}$.\;
\If{$\abs{V_1} > \frac{2}{\varepsilon} - 4$}{
\label{reduce-2vc:both-large}
Let $H'_i \coloneq \Reduce(G'_i)$ for $i \in \{1,2\}$.\label{reduce-2vc:both-large:recurse} \;
Let $F' \subseteq E$ be a minimum-size edge set such that $H' \coloneq H'_1 \cup H'_2 \cup F'$ is 2EC. 
\Return $H'$\;
}
\Else{\label{reduce-2vc:one-small}
Let $\OPT_1^t$ be the minimum-size subgraphs of $G_1$ for every type $t \in \typesTwoVC$ (w.r.t.\ $\{u,v\}$ and $G_1$) that belong to some 2EC spanning subgraph of $G$, if exists. Let $\OPT_1^{\min}$ be the existing subgraph of minimum size among all types in $\typesTwoVC$, where ties are broken according to \Cref{def:2vc-ties} (that is, $t' \in \typesTwoVC$ is preferred over $t \in \typesTwoVC$ if $t' \succ t$), and $t^{\min} \in \typesTwoVC$ the corresponding type. Let $\opt_1^t \coloneq |\OPT_1^t|$ for every $t \in \typesTwoVC$ and $\opt_1^{\min} \coloneq |\OPT_1^{\min}|$. \label{reduce-2vc:compute-opt-types} \;
\If{$t^{\min} = \typB$}{\label{reduce-2vc:b}
    Let $G_2^{\typB} \coloneq (V(G_2) \cup \{w\},E(G_2) \cup \{uw,vw\})$ and $H_2^{\typB} \coloneq \Reduce(G_2^{\typB}) \setminus \{uw,vw\}$, where $w$ is a dummy vertex and $uw, vw$ are dummy edges.  \;
$H^{\typB} \coloneq \OPT_1^{\typB} \cup H_2^{\typB}$
    \Return $H^{\typB}$. \;
}
\ElseIf{$t^{\min} = \typC$}{\label{reduce-2vc:c}
    Let $G_2^{\typC} \coloneq (V(G_2),E(G_2) \cup \{uv\})$ and $H_2^{\typC} \coloneq \Reduce(G_2^{\typC}) \setminus \{uv\}$.  \;
    Let $F^{\typC} \subseteq E
    $ be a min-size edge set s.t.\ $H^{\typC} \coloneq \OPT_1^{\typC} \cup H_2^{\typC} \cup F^{\typC}$ is 2EC. 
    \Return $H^{\typC}$. \;
}
}
\end{algorithm}

\begin{algorithm}
\small    \DontPrintSemicolon
\caption{Remove a large 3-vertex cut}\label{alg:3-vertex-cuts}
\KwIn{A 2EC graph $G=(V,E)$ without cut vertices, non-isolating 2-vertex cuts, or $\alpha$-contractible subgraphs with at most $\frac{4}{\varepsilon}$ vertices. A large 3-vertex cut $\{u,v,w\}$ in $G$.}
Let $(V_1,V_2)$ be a partition of $V \setminus \{u,v,w\}$ such that $7 \leq \abs{V_1} \leq \abs{V_2}$ and there are no edges between $V_1$ and $V_2$ in $G$. \;
Let $G_1 \coloneq G[V_1 \cup \{u,v,w\}]$ and $G_2 \coloneq G[V_2 \cup \{u,v,w\}] \setminus \{uv, vw, uw \}$.\;
Let $G'_i \coloneq G_i | \{u,v,w\}$ for $i \in \{1,2\}$.\;
\If{$\abs{V_1} > \frac{2}{\varepsilon} - 4$}{
\label{reduce-3vc:both-large}
Let $H'_i \coloneq \Reduce(G'_i)$ for $i \in \{1,2\}$. \label{reduce-3vc:both-large:recurse} \;
Let $F' \subseteq E$ be a minimum-size edge set such that $H' \coloneq H'_1 \cup H'_2 \cup F'$ is 2EC. 
\Return $H'$\;
}
\Else{
\label{reduce-3vc:one-small}
Let $\OPT_1^t$ be the minimum-size subgraphs of $G_1$ for every type $t \in \typesThreeVC$ (w.r.t.\ $\{u,v,w\}$ and $G_1$) that belong to some 2EC spanning subgraph of $G$, if exists. Let $\OPT_1^{\min}$ be the existing subgraph of minimum size among all types in $\typesThreeVC$, where ties are broken according to \Cref{def:3vc-ties} (that is, $t' \in \typesThreeVC$ is preferred over $t \in \typesThreeVC$ if $t' \succ t$), and $t^{\min} \in \typesThreeVC$ the corresponding type. Let $\opt_1^t \coloneq |\OPT_1^t|$ for every $t \in \typesThreeVC$ and $\opt_1^{\min} \coloneq |\OPT_1^{\min}|$. \label{reduce-3vc:compute-opt-types} \;
\If(\tcp*[h]{also if $t^{\min} = \typBi$}){$\OPT_1^{\typBi}$ exists and $\opt_1^{\typBi} \leq \opt_1^{\min} + 1$}{
\label{reduce-3vc:b1}
Let $G_2^{\typBi} \coloneq G_2'$ and $H_2^{\typBi} \coloneq \Reduce(G_2^{\typBi})$.  \;
Let $F^{\typBi} \subseteq E
$ be a min-size edge set s.t.\ $H^{\typBi} \coloneq \OPT_1^{\typBi} \cup H_2^{\typBi} \cup F^{\typBi}$ is 2EC. 
\Return $H^{\typBi}$. \;
}
\ElseIf{$t^{\min} = \typBii$}{
\label{reduce-3vc:b2}
Let $G_2^{\typBii} \coloneq G_2'$ and $H_2^{\typBii} \coloneq \Reduce(G_2^{\typBii})$. \;
Let $F^{\typBii} \subseteq E
$ be a min-size edge set s.t.\ $H^{\typBii} \coloneq \OPT_1^{\typBii} \cup H_2^{\typBii} \cup F^{\typBii}$ is 2EC.  
\Return $H^{\typBii}$. \;
}
\ElseIf{$t^{\min} = \typCi$}{
\label{reduce-3vc:c1}
Let $G_2^{\typCi} \coloneq G_2'$ and $H_2^{\typCi} \coloneq \Reduce(G_2^{\typCi})$. \;
Let $F^{\typCi} \subseteq E
$ be a min-size edge set s.t.\ $H^{\typCi} \coloneq \OPT_1^{\typCi} \cup H_2^{\typCi} \cup F^{\typCi}$ is 2EC.  
\Return $H^{\typCi}$. \;
}
\ElseIf(\tcp*[h]{from now on $y, z$ are dummy vertices with incident dummy edges}){$t^{\min} = \typCii$}{
\label{reduce-3vc:c2-general}
    \If{Every $\OPT_1^{\typCii}$ solution contains a $C(u)-C(v)$ path.\label{reduce-3vc:c2-i}}{
        Let $G_2^{\typCii} \coloneq (V(G_2) \cup \{y\}, E(G_2) \cup \{uy,vy,vw\})$. \;
        Let $H_2^{\typCii} \coloneq \Reduce(G_2^{\typCii}) \setminus \{uy,vy,vw\}$. \;
        Let $F^{\typCii} \subseteq E
        $ be a min-size edge set s.t.\ $H^{\typCii} \coloneq \OPT_1^{\typCii} \cup H_2^{\typCii} \cup F^{\typCii}$ is 2EC.
        \Return $H^{\typCii}$. \;
    }
    \ElseIf{Every $\OPT_1^{\typCii}$ solution contains either a $C(u)-C(v)$ or a $C(v)-C(w)$ path.\label{reduce-3vc:c2-ii}}{
        Let $G_2^{\typCii} \coloneq (V(G_2) \cup \{y,z\}, E(G_2) \cup \{uy,vz,zy,wy\})$. \;
        Let $H_2^{\typCii} \coloneq \Reduce(G_2^{\typCii}) \setminus \{uy,vz,zy,wy\}$. \;
        Let $F^{\typCii} \subseteq E
        $ be a min-size edge set and $\OPT_1^{\typCii}$ be s.t.\ $H^{\typCii} \coloneq \OPT_1^{\typCii} \cup H_2^{\typCii} \cup F^{\typCii}$ is 2EC. \label{reduce-3vc:c2-ii-sol-constr}
        \Return $H^{\typCii}$. \;
    }
    \Else(\tcp*[h]{all paths are possible between $C(u)$, $C(v)$, and $C(w)$}){\label{reduce-3vc:c2-iii}
        Let $G_2^{\typCii} \coloneq (V(G_2) \cup \{y\}, E(G_2) \cup \{uy,vy,wy\})$. \;
        Let $H_2^{\typCii} \coloneq \Reduce(G_2^{\typCii}) \setminus \{uy,vy,wy\}$. \;
        Let $F^{\typCii} \subseteq E
        $ be a min-size edge set and $\OPT_1^{\typCii}$ be s.t.\ $H^{\typCii} \coloneq \OPT_1^{\typCii} \cup H_2^{\typCii} \cup F^{\typCii}$ is 2EC. \label{reduce-3vc:c2-iii-sol-constr}
        \Return $H^{\typCii}$. \;
    }   
}
\ElseIf{$t^{\min} = \typCiii$}{
\label{reduce-3vc:c3}
Let $\OPT_1^{\typCiii}$ be a \typCiii solution such that there is an edge $e^\typCiii_{uv}$ in $G_1$ between $C(u)$ and $C(v)$ and there is an edge $e^\typCiii_{vw}$ in $G_1$ between $C(v)$ and $C(w)$. \label{reduce-3vc:c3-select-solution} \; 
Let $G_2^{\typCiii} \coloneq (V(G_2), E(G_2) \cup \{uv,uv,vw,vw\})$ and $H_2^{\typCiii} \coloneq \Reduce(G_2^{\typCiii}) \setminus \{uv,uv,vw,vw\}$. \;
Let $F^{\typCiii} \subseteq E
$ be a min-size edge set s.t.\ $H^{\typCiii} \coloneq \OPT_1^{\typCiii} \cup H_2^{\typCiii} \cup F^{\typCiii}$ is 2EC.  \label{reduce-3vc:c3-sol-constr}
\Return $H^{\typCiii}$. \;
}
}
\end{algorithm}

\subsection{Auxiliary Lemmas}\label{sec:reduce-auxiliary}

The following two lemmas are standard results, and can be found in, e.g., \cite{HVV19,GargGA23improved}.

\begin{lemma}\label{lemma:cut-vertex-egal}
    Let $G$ be a 2EC graph, $v$ be a cut vertex of $G$, and $(V_1,V_2), V_1 \neq \emptyset \neq V_2$, be a partition of $V \setminus \{v\}$ such that there are no edges between $V_1$ and $V_2$. Then, $\opt(G) = \opt(G[V_1 \cup \{v\}]) + \opt(G[V_2 \cup \{v\}])$.
\end{lemma}

\begin{lemma}\label{lemma:self-loops-parallel-edges}
    Let $G$ be a 2VC multigraph with $|V(G)| \geq 3$. Then, there exists a minimum size 2EC spanning subgraph of $G$ that is simple, that is, it contains no self loops or parallel edges.
\end{lemma}

Next, we replicate a lemma of \cite{GargGA23improved} that
shows that irrelevant edges are irrelevant for obtaining an optimal 2-ECSS.

\begin{lemma}[Lemma 2.1 in \cite{GargGA23improved}]
    \label{lemma:irrelevant-edge}
    Let $e = uv$ be an irrelevant edge of a 2VC simple graph $G = (V, E)$. Then there exists a minimum-size 2EC spanning subgraph of $G$ not containing $e$.
\end{lemma}

\begin{proof}
Assume by contradiction that all optimal 2EC spanning subgraphs of $G$ contain $e$, and let $\OPT$ be one such solution. Define $\OPT' \coloneq \OPT \setminus \{e\}$.
Clearly, $\OPT'$ cannot be 2EC as it would contradict the optimality of $\OPT$. Let $(V_1,V_2)$ be any partition of $V \setminus \{u,v\}$, $V_1 \neq \emptyset \neq V_2$, such that there are no edges between $V_1$ and $V_2$, which must exist since $\{u, v\}$ is a 2-vertex cut. 
Notice that at least one of $\OPT'_1 \coloneq \OPT'[V_1 \cup \{u,v\}]$ and $\OPT'_2 \coloneq \OPT'[V_2 \cup \{u,v\}]$, say $\OPT'_2$, needs to be connected, as otherwise $\OPT$ would not be 2EC. 
We also have that $\OPT'_1$ is disconnected, as otherwise $\OPT'$ would be 2EC. 
More precisely, $\OPT'_1$ consists of exactly two connected components $\OPT'_1(u)$ and $\OPT'_1(v)$ containing $u$ and $v$, respectively. Assume w.l.o.g.\ that $|V(\OPT'_1(u))| \geq 2$. 
Observe that there must exist an edge $f$ between $V(\OPT'_1(u))$ and $V(\OPT'_1(v))$. Otherwise, $u$ would be a 1-vertex cut separating $V(\OPT'_1(u)) \setminus \{u\}$ from $V \setminus (V(\OPT'_1(u)) \cup \{u\})$.
Thus, $\OPT'' \coloneq \OPT' \cup \{f\}$ is an optimal 2EC spanning subgraph of $G$ not containing $e$, a contradiction.
\end{proof}

The next lemma shows that for both, 3-vertex cuts and 2-vertex cuts, type $\typA$ solutions are more expensive than the cheapest solution.

\begin{lemma}\label{lemma:type-A-high-cost}
    If \Cref{alg:2-vertex-cuts} (\Cref{alg:3-vertex-cuts}) reaches \Cref{reduce-2vc:one-small} (\Cref{reduce-3vc:one-small}) and $\OPT_1^{\typA}$ exists, then $\opt_1^{\typA} \geq \alpha \cdot \opt_1^{\min} + 1$.
\end{lemma}
    
\begin{proof}
    If $\OPT_1^{\typA}$ exists and $\opt_1^{\typA} \leq \alpha \cdot \opt_1^{\min}$ then $\OPT_1^{\typA}$ is 2EC and $\alpha$-contractible. 
    However, since we assume that we reached \Cref{reduce-2vc:one-small} (\Cref{reduce-3vc:one-small}), we have that $|V_1| \leq \frac{2}{\varepsilon} - 4 \leq \frac{4}{\varepsilon}$. This is a contradiction, because we would have contracted $G_1$ in \Cref{reduce:contractible} of \Cref{alg:reduce} earlier. Thus, it must be that $\opt_1^{\typA} \geq \alpha \cdot \opt_1^{\min} + 1$. \end{proof}

\subsubsection{Auxiliary Lemmas for \Cref{alg:2-vertex-cuts}}\label{app:prep:helper-2vc}

\begin{proposition}\label{prop:properties-for-2vc-alg}
    If \Cref{alg:reduce} executes \Cref{alg:2-vertex-cuts}, then $G$ contains no self loops, parallel edges, irrelevant edges, or $\alpha$-contractible subgraphs with at most $\frac{4}{\varepsilon}$ vertices.
\end{proposition}

\begin{lemma}[Lemma A.4 in \cite{GargGA23improved}]
    \label{lemma:2vc-type-combinations}
    Let $G$ be the graph when \Cref{alg:reduce} executes \Cref{alg:2-vertex-cuts}, $\{u,v\}$ be a non-isolating 2-vertex cut of $G$, $(V_1,V_2)$ be a partition of $V \setminus \{u,v\}$ such that $V_1 \neq \emptyset \neq V_2$ and there are no edges between $V_1$ and $V_2$, and $H$ be a 2EC spanning subgraph of $G$. Let $G_i = G[V_i \cup \{u,v\}]$ and $H_i = E(G_i) \cap H$ for $i \in \{1,2\}$. The following statements are true.
    \begin{enumerate}[1.]
        \item Both $H_1$ and $H_2$ are of a type in $\typesTwoVC$ with respect to $\{u,v\}$. Furthermore, if one of the two is of type $\typC$, then the other must be of type $\typA$.
        \item If $H_i$, for an $i \in \{1,2\}$, is of type $\typC$, then there exists one edge $f \in E(G_i)$ such that $H_i' \coloneq H_i \cup \{f\}$ is of type $\typB$. As a consequence, there exists a 2EC spanning subgraph $H'$ where $H_i' \coloneq H' \cap E(G_i)$ is of type $\typA$ or $\typB$.
    \end{enumerate}
\end{lemma}

\begin{proof}
    First note that $uv \notin E(G)$ since there are no irrelevant edges by \Cref{prop:properties-for-2vc-alg}.
    We first prove the first claim. We prove it for $H_1$, the other case being symmetric. 
    First notice that if $H_1$ contains a connected component not containing $u$ nor $v$, $H$ would be disconnected. Hence, $H_1$ consists of one connected component or two connected components $H_1(u)$ and $H_1(v)$ containing $u$ and $v$, respectively.

    Suppose first that $H_1$ consists of one connected component. Let us contract the 2EC components of $H_1$, hence obtaining a tree $T$. If $T$ consists of a single vertex, $H_1$ is of type $\typA$. Otherwise, consider the path $P$ (possibly of length $0$)
    between the super nodes resulting from the contraction of the 2EC components $C(u)$ and $C(v)$ containing $u$ and $v$, respectively. Assume by contradiction that $T$ contains an edge $e$ not in $P$. Then $e$ does not belong to any cycle of $H$, contradicting the fact that $H$ is 2EC. Thus, $H_1$ is of type $\typB$ (in particular $P$ has length at least 1 since $H_1$ is not 2EC).

    Assume now that $H_1$ consists of 2 connected components $H_1(u)$ and $H_1(v)$. Let $T_u$ and $T_v$ be the two trees obtained by contracting the 2EC components of $H_1(u)$ and $H_1(v)$, respectively. By the same argument as before, the 2-edge connectivity of $H$ implies that these two trees contain no edge. Hence, $H_1$ is of type $\typC$.

    The second part of the first claim follows easily since if $H_1$ is of type $\typC$ and $H_2$ is not of type $\typA$, then $H$ would either be disconnected or it would contain at least one bridge edge, because $uv \notin E(G)$.

    We now move to the second claim of the lemma. There must exist an edge $f$ in $G_i$ between the two connected components of $H_i$ since otherwise at least one of $u$ or $v$ would be a cut vertex. 
    Clearly, $H'_i \coloneq H_i \cup \{f\}$ satisfies the claim. For the second part of the claim, consider any 2EC spanning 
    subgraph $H'$ and let $H'_i \coloneq H' \cap E(G_i)$.
    The claim holds if there exists one such $H'$ where $H'_i$ is of type $\typA$ or $\typB$, hence assume that this is not the case. Hence, $H'_i$ is of type $\typC$ by the first part of this lemma. The same argument as above implies the existence of an edge $f$ in $G_i$ such that $H''_i \coloneq H'_i \cup \{f\}$ is of type $\typB$, implying the claim for $H'' \coloneq H' \cup \{f\}$.
\end{proof}

\begin{lemma}\label{lemma:2vc-opt-1-large}
    In \Cref{alg:2-vertex-cuts} it holds that $\opt_1^{\min} \geq 3$. 
\end{lemma}

\begin{proof}
    Consider any feasible 2-ECSS solution $H_1$ for $G_1$, and let $H'_1 = H_1 \setminus \{u,v\}$ be the corresponding solution for $G_1 | \{u,v\}$.  
    Since $G_1 | \{u,v\}$ contains at least $3$ vertices and $H'_1$ is a feasible 2-ECSS solution, $H'_1$ must contain at least $3$ edges. 
    Thus, $|H_1| \geq |H'_1| \geq 3$.
\end{proof}

We have an immediate corollary from \Cref{lemma:type-A-high-cost} and $\alpha > 1$. 

\begin{corollary}\label{lemma:2vc-type-A-condition}
    If \Cref{alg:2-vertex-cuts} reaches \Cref{reduce-2vc:one-small} and $\OPT_1^{\typA}$ exists, then $\opt_1^{\typA} \geq \opt_1^{\min} + 2$. In particular, $t^{\min} \neq \typA$.
\end{corollary}

\begin{lemma}[Claim 1 in \cite{GargGA23improved}]\label{lemma:2vc-only-b-and-c}
    If \Cref{alg:2-vertex-cuts} is executed on instance $G$, there exists an optimal solution $\opt(G)$ such that $\OPT_1$ is of type $\typB$ or $\typC$.
\end{lemma}

\begin{proof}
For the sake of contradiction, consider any optimal solution $\OPT(G)$ and assume that $\OPT_1$ is of type $\typA$.
\cref{lemma:2vc-type-combinations} implies that every feasible solution must use at least $\opt_1^{\min}$ edges from $G_1$. Moreover, $\opt_1^{\typA} > \alpha \cdot \min\{\opt_1^\typB,\opt_1^\typC\}$ by \Cref{lemma:type-A-high-cost}.
If $\opt_1^\typA \geq \opt_1^\typB + 1$, then we obtain an alternative optimum solution with the desired property by taking $\OPT_1^\typB \cup \OPT_2$, and, in case $\OPT_2$ is of type $\typC$, by adding one edge $f$ between the connected components of $\OPT_2$ whose existence is guaranteed by \cref{lemma:2vc-type-combinations}. Otherwise, that is, $\opt_1^\typA \leq \opt_1^\typB$, we must have $\opt_1^\typA > \alpha \cdot \opt_1^\typC$. Specifically, $\opt_1^\typA \geq \opt_1^\typC + 1 \geq 4$ by \cref{lemma:2vc-opt-1-large}. \cref{lemma:2vc-type-combinations} also implies $\opt_1^\typB \leq \opt_1^\typC + 1$, which together gives $\opt_1^\typA = \opt_1^\typB = \opt_1^\typC + 1$. Then, $\opt_1^\typA > \alpha \cdot \opt_1^\typC = \alpha(\opt_1^\typA - 1)$ implies $\opt_1^\typA \leq 5$ since $\alpha \geq \frac{6}{5}$. We established that $\opt_1^\typA \in \{4,5\}$.

If $\opt_1^\typA = 4$, $\OPT_1^\typA$ has to be a $4$-cycle $C =u - a -v -b- u$. Notice that the edge $ab$ cannot exist since otherwise $\opt_1^\typB = 3$ due to the path $u - a - b - v$. Then every feasible solution must include the 4 edges of $C$ to guarantee degree at least 2 on $a$ and $b$. But this makes $C$ an $\alpha$-contractible subgraph on $4 \leq \frac{4}{\varepsilon}$ vertices, which is a contradiction to \cref{prop:properties-for-2vc-alg}.

If $\opt_1^\typA = 5$, the minimality of $\OPT_1^\typA$ implies that $\OPT_1^\typA$ is a $5$-cycle, say $C = u - a - v - b - c -u$. Observe that the edge $ab$ cannot exist, since otherwise $\{va,ab,bc,cu\}$ would be a type $\typB$ solution in $G_1$ of size $4$, contradicting $\opt_1^\typB = \opt_1^\typB = 5$. A symmetric construction shows that edge $ac$ cannot exist. Hence, every feasible solution restricted to $G_1$ must include the edges $\{au,av\}$ and furthermore at least $3$ more edge incident to $b$ and $c$, so at least $5$ edges altogether. This implies that $C$ is an $\alpha$-constractible subgraph of size $5 \leq \frac{4}{\varepsilon}$ vertices, which is a contradiction to \cref{prop:properties-for-2vc-alg}.
\end{proof}

\begin{lemma}[Part of the proof of Lemma A.6 in \cite{GargGA23improved}]\label{lemma:2vc-both-large-F}
    If the condition in \Cref{reduce-2vc:both-large} in \Cref{alg:2-vertex-cuts} applies and both, $H'_1$ and $H'_2$ are 2EC spanning subgraphs of $G'_1$ and $G'_2$, respectively, then there exist edges $F'$ with $|F'| \leq 2$ such that $H'$ is a 2EC spanning  subgraph of $G$.
\end{lemma}
\begin{proof}
First note that $H'_i$, $i \in \{1,2\}$, must be of type $\typA$, $\typB$, or $\typC$. If one of them is of type $\typA$ or they are both of type $\typB$, then $H'_1 \cup H'_2$ is 2EC. 
If one of them is of type $\typB$, say $H'_1$, and the other of type $\typC$, say $H'_2$, then let $H'_2(u)$ and $H'_2(v)$ be the two 2EC components of $H'_2$ containing $u$ and $v$, respectively. By \cref{lemma:2vc-type-combinations}, there must exist an edge $f \in E(G_2)$ between $H'_2(u)$ and $H'_2(v)$. Hence, $H'_1(u) \cup H'_2(v) \cup \{f\}$ is 2EC. 
The remaining case is that $H'_1$ and $H'_2$ are both of type $\typC$. In this case, $H'_1 \cup H'_2$ consists of precisely two 2EC components $H'(u)$ and $H'(v)$ containing $u$ and $v$, respectively. 
Since $G$ is 2EC, there must exist two edges $f$ and $g$ between $H'(u)$ and $H'(v)$ such that $H'_1 \cup H'_2 \cup \{f,g\}$ is 2EC. Thus, in all cases a desired set $F^\typC$ of size at most 2 exists such that $H'$ is a 2-ECSS of $G$.
\end{proof}

\begin{lemma}[Part of the proof of Lemma A.6 in \cite{GargGA23improved}]\label{lemma:2vc-b-F}
    If the condition in \Cref{reduce-2vc:b} in \Cref{alg:2-vertex-cuts} applies and $\Reduce(G_2^{\typB})$ is a 2EC spanning subgraph of $G_2^{\typB}$, then  $H^\typB$ is a 2EC spanning subgraph of~$G$.
\end{lemma}
\begin{proof}
Note that $\Reduce(G_2^{\typB})$ must contain the two dummy edges $wv$ and $wu$, which induce a type $\typB$ graph. Hence, by \cref{lemma:2vc-type-combinations}, $H^{\typB}_2$ is of type $\typA$ or $\typB$.
Thus, $H^{\typB} = \OPT_1^\typB \cup H^\typB$ is a 2-ECSS of $G$.
\end{proof}

\begin{lemma}[Part of the proof of Lemma A.6 in \cite{GargGA23improved}]\label{lemma:2vc-c-F}
    If the condition in \Cref{reduce-2vc:c} in \Cref{alg:2-vertex-cuts} applies and $\Reduce(G_2^{\typC})$ is a 2EC spanning subgraph of $G_2^{\typC}$, then there exists an edge set $F^{\typC}$ with $|F^{\typC}| \leq 1$ such that $H^\typC$ is a 2EC spanning subgraph of $G$.
\end{lemma}

\begin{proof}
Note that $H_2^\typC$ must be of type $\typA$ or $\typB$. 
\cref{lemma:2vc-type-combinations} guarantees the existence of an edge $f \in E(G_2)$ such that $\OPT_1^{\typC} \cup \{f\}$ is of type $\typB$ for $G_1$. Thus, $\OPT_1^{\typC} \cup H_2^\typC \cup \{f\}$ is a 2-ECSS of $G$.
\end{proof}

Finally, we prove that the reductions of \cref{alg:2-vertex-cuts} preserve the approximation factor.

\begin{lemma}[Part of the proof of Lemma A.7 in \cite{GargGA23improved}]\label{lemma:2vc-approx}
    Let $G$ be a 2-ECSS instance. If every recursive call to $\Reduce(G')$ in \Cref{alg:2-vertex-cuts} on input $G$ satisfies $\reduce(G') \leq \alpha \cdot \opt(G') + 4\varepsilon \cdot |V(G')| - 4$, then it holds that $\reduce(G) \leq \alpha \cdot \opt(G) + 4\varepsilon \cdot |V(G)| - 4$.
\end{lemma}

\begin{proof}
Let $\OPT_i \coloneq \OPT(G) \cap E(G_i)$ and $\opt_i \coloneq \abs{\OPT_i}$ for $i \in \{1,2\}$.
Note that $\opt(G) = \opt_1 + \opt_2$.
For the remainder of this proof, all line references (if not specified differently) refer to \Cref{alg:2-vertex-cuts}. We use in every case that $\opt_1 \geq 2$ (cf. \Cref{lemma:2vc-opt-1-large}).
We distinguish in the following which reduction the algorithm uses.

\begin{description}
    \item[(\Cref{reduce-2vc:both-large}: $|V_1| > \frac{2}{\varepsilon} - 4$)]
Note that for $i \in \{1,2\}$ we have that $\OPT_i$ is a feasible solution for $G_i'$ and therefore $\opt_i \geq \opt(G_i')$.
Using the assumption of the lemma, \Cref{lemma:2vc-both-large-F} and $\abs{V(G_1')} + \abs{V(G_2')} = \abs{V(G)}$, we conclude that
\begin{align*}
    \reduce(G) = \abs{H'_1} + \abs{H'_2} + \abs{F'} &\leq \alpha \cdot (\opt_1 + \opt_2) + 4\varepsilon(\abs{V(G_1')} + \abs{V(G_2')}) -8 + 2 \\ 
    &\leq \alpha \cdot \opt(G) + 4\varepsilon \cdot (\abs{V(G)}) - 4 \ .
\end{align*}

\item[(\Cref{reduce-2vc:b}: $t^{\min} = \typB$)]
Note that $\opt_2(G_2^{\typB}) \leq \opt_2 + 2$, because we can turn $\OPT_2$, which we can assume to be of type $\typA$ or $\typB$ by \Cref{lemma:2vc-only-b-and-c} and \Cref{lemma:2vc-type-combinations}, using both dummy edges $\{uw,vw\}$ into a solution for $G_2^{\typB}$. Thus,
\begin{align*} 
    \reduce(G) = \abs{H^{\typB}} &\leq \opt_1^{\typB} + \abs{H^{\typB}}
    \leq \opt_1 + (\alpha \cdot \opt_2(G_2^{\typB}) + 4 \varepsilon |V(G)| - 4) - 2 \\
    &\leq \opt_1 + (\alpha \cdot (\opt_2+2) + 4 \varepsilon |V(G)| - 4) - 2 \\
    &\leq \alpha \opt(G) + 4 \varepsilon |V(G)| - 4 +
    (\alpha - 1)(2 - \opt_1) \leq \alpha \opt(G) + 4 \varepsilon |V(G)| - 4 \ .
\end{align*}

\item[(\Cref{reduce-2vc:c}: $t^{\min} = \typC$)]
We distinguish the following cases depending on which type combination of \Cref{lemma:2vc-type-combinations} is present for $(\OPT_1,\OPT_2)$. We exclude the case that $\OPT_1$ is of type $\typA$ due to \Cref{lemma:2vc-only-b-and-c}. Moreover, \cref{lemma:2vc-c-F} gives that $|F^{\typC}| \leq 1$.

\begin{enumerate}[(a)]

    \item $(\typB, \{\typA,\typB\})$.
    In this case, $\opt_1^{\typB}$ exists but $t^{\min} = \typC$. Thus, the definition of \Cref{alg:2-vertex-cuts} and \Cref{def:2vc-ties} imply that 
    $\opt_1^{\typC} \leq \opt^{\typB}_1 - 1 = \opt_1 - 1$. Moreover, $\opt_2(G_2^{\typC}) \leq \opt_2 + 1$, because we can turn $\OPT_2$ with the dummy edge $\{uv\}$ into a solution for $G_2^{\typC}$. Thus,
    \begin{align*} 
    \reduce(G) = \abs{H^{\typC}} &\leq \opt_1^{\typC} + \abs{H^{\typC}} + \abs{F^\typC}\\
    &\leq (\opt_1 - 1) + (\alpha \cdot (\opt_2 + 1) + 4 \varepsilon |V(G)| - 4) - 1 + 1  \\
    &\leq \alpha \opt(G) + 4 \varepsilon |V(G)| - 4 +
    (\alpha - 1)(1 - \opt_1) \leq \alpha \opt(G) + 4 \varepsilon |V(G)| - 4 \ .
    \end{align*}

    \item $(\typC, \typA)$.
    In this case, $\opt_1 = \opt_1^{\typC}$ and $\OPT_2$ is a feasible solution for $G_2^{\typC}$. Thus,
    \begin{align*} 
    \reduce(G) = \abs{H^{\typC}} &\leq \opt_1^{\typC} + \abs{H^{\typC}} + \abs{F^\typC} \\
    &\leq \opt_1 + (\alpha \cdot \opt_2 + 4 \varepsilon |V(G)| - 4) - 1 + 1 \leq \alpha \opt(G) + 4 \varepsilon |V(G)| - 4 \ .
    \end{align*}
\end{enumerate}
\end{description}

This completes the proof of the lemma.
\end{proof}

\subsubsection{Auxiliary Lemmas for \Cref{alg:3-vertex-cuts}}\label{app:prep:helper-3vc}

\begin{proposition}\label{prop:properties-for-3vc-alg}
    If \Cref{alg:reduce} executes \Cref{alg:3-vertex-cuts}, then $G$ contains no non-isolating 2-vertex cuts, self loops, parallel edges, irrelevant edges, or $\alpha$-contractible subgraphs with at most $\frac{4}{\varepsilon}$ vertices.
\end{proposition}

\begin{lemma}
    Let $G$ be the graph when \Cref{alg:reduce} executes \Cref{alg:3-vertex-cuts}, $\{u,v,w\}$ be a large 3-vertex cut of $G$, $(V_1,V_2)$ be partition of $V \setminus \{u,v\}$ such that $V_1 \neq \emptyset \neq V_2$ and there are no edges between $V_1$ and $V_2$, and $H$ be a 2EC spanning subgraph of $G$. Let $G_i = G[V_i \cup \{u,v,w\}]$ and $H_i = E(G_i) \cap H$ for $i \in \{1,2\}$. Then, $H_1$ and $H_2$ are of a type in $\typesThreeVC$ with respect to $\{u,v,w\}$ and $G_1$ and $G_2$, respectively.
\end{lemma}

\begin{proof}
    Let $i \in \{1,2\}$. First, observe that $H_i$ consist of either one, two or three connected components, and each component contains at least one of $\{u,v,w\}$, as otherwise $H$ would be disconnected. Let $H_i'$ be obtained by contracting each 2EC component $C$ of $H_i$ into a single super-node $C$, and let $C_i(x)$ be the corresponding super-node of the 2EC component of $H$ that contains $x \in V(H_i)$.

    First, assume that $H_i$ is a single connected component. 
    In this case $H_i'$ is a tree. If $H_i'$ is composed of a single vertex, $H_i$ is of type $\typA$. Otherwise, let $F$ be the set of edges of $H_i'$ that are on any simple path between $C_i(u)$, $C_i(v)$, and $C_i(w)$. If $H_i'$ contains an edge $e \in E(H_i') \setminus F$, then $e$ cannot belong to any cycle of $H$, because it does not belong to a cycle in $H_i$ and is not on a simple path between any two vertices of the cut $\{u,v,w\}$, contradicting that $H$ is 2EC. Thus, $H_i$ is of type $\typCi$ if $C_i(u)$, $C_i(v)$, and $C_i(w)$  are distinct, and of type $\typBi$ otherwise.

     Second, assume that $H_i$ is composed of two connected components.
     Let $H_i'(v)$ and $H_i'(w)$ denote the two trees in $H_i'$ such that w.l.o.g.\ $u$ and $v$ are in $H_i'(v)$ and $w$ is in $H_i'(w)$. 
     First note that $H_i'(w)$ must be a single vertex. Otherwise, there is an edge $e$ in $H_i'(w)$ that cannot be part of any cycle in $H$, contradicting that $H$ is 2EC.
     If $H_i'(v)$ is also a single vertex, then $H_i$ is of type $\typBii$.
     Otherwise, let $F$ be the set of edges of $H_i'(v)$ on the simple path between $C_i(u)$ and $C_i(v)$. 
     If $H_i'(v)$ contains an edge $e \in E(H_i'(v)) \setminus F$, then $e$ cannot belong to any cycle of $H$, contradicting that $H$ is 2EC. Thus, $H_i$ is of type $\typCii$.
     
     Finally, assume that $H_i$ is composed of three connected components.
     Let $H_i'(u)$, $H_i'(v)$, and $H_i'(w)$ denote the trees of $H_i'$ that contain $u$, $v$, and $w$, respectively. 
     Using the same argument as before, we can show that each tree is composed of a single vertex, and, thus, $H_i$ is of type $\typCiii$.
\end{proof}

\begin{lemma}\label{lemma:3vc-opt-1-large}
   In \Cref{alg:3-vertex-cuts} it holds that $\opt_1^{\min} \geq 8$.
\end{lemma}

\begin{proof}
Consider any feasible 2-ECSS solution $H_1$ for $G_1$, and let $H'_1 = H_1 \setminus \{u,v,w\}$ be the corresponding solution for $G_1 | \{u,v,w\}$.  
Since $G_1 | \{u,v,w\}$ contains at least $8$ vertices and $H'_1$ is a feasible 2-ECSS solution, $H'_1$ must contain at least $8$ edges. 
Thus, $|H_1| \geq |H'_1| \geq 8$.
\end{proof}

We have an immediate corollary from the above lemma, \cref{lemma:type-A-high-cost}, and $\alpha \geq \frac{5}{4}$.

\begin{corollary}\label{lemma:3vc-type-A-condition}
If \Cref{alg:3-vertex-cuts} reaches \Cref{reduce-3vc:one-small} and $\OPT_1^{\typA}$ exists, then $\opt_1^{\typA}  \geq \opt_1^{\min} + 3$. In particular, $t^{\min} \neq \typA$.
\end{corollary}

\begin{lemma}\label{lemma:3vc-spanning-tree}
     Let $G = (V,E)$ be a (multi-)graph that does not contain cut vertices or non-isolating 2-vertex cuts, and let $\{u,v,w\}$ be a large 3-vertex cut of $G$. 
     Let $(V_1,V_2)$ be a partition of $V \setminus \{u,v,w\}$ such that $7 \leq \abs{V_1} \leq \abs{V_2}$ and there are no edges between $V_1$ and $V_2$ in $G$, and let $G_1 \coloneq G[V_i \cup \{u,v,w\}]$ and $G_2 \coloneq G[V_i \cup \{u,v,w\}] \setminus \{uv, vw, uw\}$.
     For $i \in \{1,2\}$, let $H_i \subseteq E(G_i)$, let $H_i(z)$ be the connected component of $H_i$ that contains $z \in V(H_i)$ and let $H_i$ be such that $V(H_i) = V(H_i(u)) \cup V(H_i(v)) \cup V(H_i(w))$, i.e., each vertex $y \in V(H_i)$ is connected to $u$, $v$, or $w$ in $H_i$. Then, the following holds.
     \begin{enumerate}[1)]
         \item For $x \in \{u,v,w\}$ and $i \in \{1,2\}$, if $H_i(x) \neq H_i$, then there exists an edge $e \in E(G_i) \setminus H_i$ between the vertices of $H_i(x)$ and the vertices of $H_i \setminus H_i(x)$.
         \item For $i \in \{1,2\}$, if $H_i(u)$, $H_i(v)$ and $H_i(w)$ are all pairwise vertex-disjoint, then there exist edges $F \subseteq E(G_i) \setminus H_i$ such that $|F| \geq 2$ and $F \subseteq \{e_{uv},e_{vw},e_{uw}\}$ where $e_{xy}$ is between the vertices of $H_i(x)$ and the vertices of $H_i(y)$.
     \end{enumerate}
\end{lemma}

\begin{proof}
    We prove the first claim, and note that the second claim follows from applying the first claim twice. 
    Let $i \in \{1,2\}$, $x \in \{u,v,w\}$ and $H_i(x) \neq H_i$.
    If $H_i(x)$ is composed of the single vertex $x$ such that there is no edge between $x$ and $H_i \setminus H_i(x)$, then $\{u,v,w\} \setminus \{x\}$ must be a non-isolating two-vertex cut in $G$, a contradiction. Otherwise, that is, $H_i(x)$ consists of at least two vertices, and there is no edge between $H_i(x)$ and $H_i \setminus H_i(x)$ then $x$ is a cut vertex in $G$ (separating $V(H_i(x)) \setminus \{x\}$ from $V \setminus V(H_i(x))$), a contradiction.
\end{proof}

\begin{lemma}\label{lemma:3vc-make-2ec}
    Let $H \subseteq E$ be such that $H$ is connected and contains a spanning subgraph of $G$, and let $H_i \coloneq H \cap E(G_i)$ for $i \in \{1,2\}$.
    If for all $i \in \{1,2\}$ every edge $e \in H_i$ is part of a 2EC component of $H_i$ or lies on an $x-y$ path in $H_i$ where $x,y \in \{u,v,w\}$ such that $x \neq y$ and there exists an $x-y$ path in $H_{j}$ where $j=1$ if $i=2$ and $j=2$ if $i=1$, then $H$ is a 2EC spanning subgraph of $G$.
\end{lemma}

\begin{proof}
    We show that every edge $e \in H$ is part of a 2EC component of $H$, which proves the lemma because $H$ is connected and spanning.
    If for all $i \in \{1,2\}$ every edge $e \in H_i$ is part of a 2EC component of $H_i$, it is also part of a 2EC component of $H$. Otherwise, for all $i \in \{1,2\}$, if $e$ lies on an $x-y$ path in $H_i$ where $x,y \in \{u,v,w\}$ such that $x \neq y$ and there exists an $x-y$ path in $H_{j}$ where $j=1$ if $i=2$ and $j=2$ if $i=1$, then $e$ is on a cycle of $H$ as both paths are edge-disjoint.
\end{proof}

\begin{lemma}\label{lemma:3vc-both-large-F}
    If the condition in \Cref{reduce-3vc:both-large} in \Cref{alg:3-vertex-cuts} applies and both, $H'_1$ and $H'_2$ are 2EC spanning subgraphs of $G'_1$ and $G'_2$, respectively, then there exist edges $F'$ with $|F'| \leq 4$ such that $H'$ is a 2EC spanning subgraph of $G$.
\end{lemma}

\begin{proof}
Let $i \in \{1,2\}$, and let $H'_i(u)$, $H'_i(v)$, and $H'_i(w)$ be the connected components of $H'_i$ after decontracting $\{u,v,w\}$ that contain $u$, $v$ and $w$, respectively.
    By \Cref{lemma:3vc-spanning-tree}, there exist edges $F_i \subseteq E(G_i)$ such that $|F_i| \leq 2$ and $H'_i \cup F_i$ is connected in $G_i$.
    Since $H'_i$ is spanning in $G_i$, setting $F \coloneq F_1 \cup F_2$ implies that $|F| \leq 4$ and that $H' = H'_1 \cup H'_2 \cup F$ is spanning and connected in $G$. 
    Moreover, observe that for $i \in \{1,2\}$ each edge $e \in H'_i \cup F_i$ is either in a 2EC component of $H'_i \cup F_i$ or lies on an $x-y$ path in $H'_i \cup F_i$, where $x, y \in \{u, v, w\}$, $x \neq y$. Thus, \Cref{lemma:3vc-make-2ec} implies that $H'$ is a 2EC spanning subgraph of $G$. 
\end{proof}

\begin{lemma}\label{lemma:3vc-b1-F}
    If the condition in \Cref{reduce-3vc:b1} in \Cref{alg:3-vertex-cuts} applies and $\Reduce(G_2^{\typBi})$ is a 2EC spanning subgraph of $G_2^{\typBi}$, then there exist edges $F^{\typBi}$ with $|F^{\typBi}| \leq 1$ such that $H^\typBi$ is a 2EC spanning subgraph of $G$.
\end{lemma}

\begin{proof}
Assume w.l.o.g.\ that $u$ and $v$ belong to the same 2EC component of $\OPT_1^{\typBi}$. 
    Let $H_1^{\typBi} \coloneq \OPT_1^{\typBi}$ and $C_i^{\typBi}(x)$ be the 2EC component of $H_i^{\typBi}$ that contains $x$, where $x \in \{u, v, w\}$ and $i \in \{1, 2\}$.  
    Since $C_1^{\typBi}(u) = C_1^{\typBi}(v)$, we have that in $\OPT_1^{\typBi} \cup H_2^{\typBi}$, the vertices of $C_1^{\typBi}(u)$, $C_2^{\typBi}(u)$, and $C_2^{\typBi}(v)$ are in the same 2EC component.
    Furthermore, in $\OPT_1^{\typBi}$, there is a connection between $C_1^{\typBi}(u)$ and $C_1^{\typBi}(w)$.
    If $w$ and $v$ are already connected in $H_2^{\typBi}$ then we can set $F^{\typBi} \coloneq \emptyset$ and see that this satisfies the lemma, since $H^{\typBi}$ contains a spanning tree and every edge in $H^{\typBi}$ is either in the 2EC component containing $u$ and $v$, in a 2EC component containing $w$ or on a $w-v$ path in $G_i$, $i \in \{1, 2\}$.
    Otherwise, if $w$ and $v$ are not in the same connected component in $H_2^{\typBi}$, by \Cref{lemma:3vc-spanning-tree}, there must be an edge $e$ from $C_2^{\typBi}(w)$ to either $C_2^{\typBi}(v)$ or $C_2^{\typBi}(u)$ in $E(G_2) \setminus H^{\typBi}$.
    Now it can be easily verified that $F^{\typBi} \coloneq \{ e \}$ satisfies the lemma, since $H^{\typBi}$ contains a spanning tree and every edge in $H^{\typBi}$ is either in the 2EC component containing $u$ and $v$, in a 2EC component containing $w$ or on a $w-v$ path in $G_i$, $i \in \{1, 2\}$.
\end{proof}

\begin{lemma}\label{lemma:3vc-b2-F}
    If the condition in \Cref{reduce-3vc:b2} in \Cref{alg:3-vertex-cuts} applies and $\Reduce(G_2^{\typBii})$ is a 2EC spanning subgraph of $G_2^{\typBii}$, then there exist edges $F^{\typBii}$ with $|F^{\typBii}| \leq 2$ such that $H^\typBii$ is a 2EC spanning subgraph of $G$.
\end{lemma}

\begin{proof}
    The proof is similar to the previous one.
Assume w.l.o.g.\ that $u$ and $v$ belong to the same 2EC component of $\OPT_1^{\typBii}$. 
    Let $H_1^{\typBii} \coloneq \OPT_1^{\typBii}$ and $C_i^{\typBii}(x)$ be the 2EC component of $H_i^{\typBii}$ that contains $x$, where $x \in \{u, v, w\}$ and $i \in \{1, 2\}$.  
    Since $C_1^{\typBii}(u) = C_1^{\typBii}(v)$, we have that in $\OPT_1^{\typBii} \cup H_2^{\typBii}$, the vertices of $C_1^{\typBii}(u)$, $C_2^{\typBii}(u)$ and $C_2^{\typBii}(v)$ are in the same 2EC component.
    Furthermore, in $\OPT_1^{\typBii}$, there is no connection between $C_1^{\typBi}(u)$ and $C_1^{\typBii}(w)$, since it is of type~$\typBii$.
    Hence, by \Cref{lemma:3vc-spanning-tree}, there must be an edge $e_1 \in E(G_1) \setminus \OPT_1^{\typBii}$.
    Note that $\OPT_1^{\typBii} \cup \{e_1\}$ is now a solution of type $\typBi$ for $G_1$, and hence we can do the same as in the proof of the previous lemma and observe that $\abs{F^{\typBii}} \leq 2$.
\end{proof}

\begin{lemma}\label{lemma:3vc-c1-F}
    If the condition in \Cref{reduce-3vc:c1} in \Cref{alg:3-vertex-cuts} applies and $\Reduce(G_2^{\typCi})$ is a 2EC spanning subgraph of $G_2^{\typCi}$, then there exist edges $F^{\typCi}$ with $|F^{\typCi}| \leq 2$ such that $H^\typCi$ is a 2EC spanning subgraph of $G$.
\end{lemma}

\begin{proof}
    Let $C_2^{\typCi}(x)$ be the 2EC component of $H_2^{\typCi}$ that contains $x$, for $x \in \{u,v,w\}$. 
Note that $\OPT_1^{\typCi}$ is connected and spanning for $G_1$ and that each edge $e \in \OPT_1^{\typCi}$ is either in some 2EC component or it is on a $y-z$-path in $\OPT_1^{\typCi}$ for $y, z \in \{u,v,w\}$, $y \neq z$.
We consider three cases. 
    If $H_2^{\typCi}$ is connected, then we claim that $F^{\typCi} \coloneq \emptyset$ satisfies the lemma.
    In this case every edge of $H_2^{\typCi}$ is part of a 2EC component of $H_2^{\typCi}$, or it is on a $y-z$-path in $H_2^{\typCi}$ for $y, z \in \{u,v,w\}$, $y \neq z$.
    Hence, \Cref{lemma:3vc-make-2ec} implies that $H_2^{\typCi} \cup \OPT_1^{\typCi}$ is a 2EC spanning subgraph of $G$.
    
    Next, assume that $H_2^{\typCi}$ is composed of two connected components, say w.l.o.g.\ one containing $u$ and $v$ and the other containing $w$, and observe that each component is 2EC. Then \Cref{lemma:3vc-spanning-tree} guarantees the existence of an edge $e$ in $G_2$ between $C_2^{\typCi}(w)$ and $C_2^{\typCi}(u)$ such that $H_2^{\typCi} \cup \{e\}$ is connected.
    Thus, we can analogously to the previous case show that 
    $\OPT_1^{\typCi} \cup H_2^{\typCi} \cup F^{\typCi}$ with $F^{\typCi} \coloneq \{e\}$ is a 2EC spanning subgraph of $G$, since in $ H_2^{\typCi} \cup F^{\typCi}$ each edge is either in some 2EC component or it is on a $y-z$-path in $ H_2^{\typCi} \cup F^{\typCi}$ for $y, z \in \{u,v,w\}$, $y \neq z$ (hence \Cref{lemma:3vc-make-2ec} holds).
    In the remaining case $H_2^{\typCi}$ is composed of three connected components, each containing exactly one vertex of $\{u, v, w\}$.
    Observe that each component is 2EC.
    Hence, \Cref{lemma:3vc-spanning-tree} guarantees the existence of at least two edges $e_1$ and $e_2$ in $G_2$ between two pairs of connected components of $H_2^{\typCi}$ such that $H_2^{\typCi} \cup \{e_1,e_2\}$ is connected.
    Thus, we can analogously to the previous cases show that $\OPT_1^{\typCi} \cup H_2^{\typCi} \cup F^{\typCi}$ with $F^{\typCi} \coloneq \{e_1,e_2\}$ is a 2EC spanning subgraph of $G$, since in $H_2^{\typCi} \cup F^{\typCi}$ each edge is either in some 2EC component or it is on a $y-z$-path in $ H_2^{\typCi} \cup F^{\typCi}$ for $y, z \in \{u,v,w\}$, $y \neq z$ (hence \Cref{lemma:3vc-make-2ec} holds). 
\end{proof}

\begin{lemma}\label{lemma:3vc-c2-i-F}
    If the condition in \Cref{reduce-3vc:c2-i} in \Cref{alg:3-vertex-cuts} applies and $\Reduce(G_2^{\typCii})$ is a 2EC spanning subgraph of $G_2^{\typCii}$, then there exist edges $F^{\typCii}$ with $|F^{\typCii}| \leq 1$ such that $H^\typCii$ is a 2EC spanning subgraph of $G$ and $|F^{\typCii}| - |\Reduce(G_2^{\typCii})| + |H_2^{\typCii}| \leq - 2$.
\end{lemma}

\begin{proof}
    Fix an optimal \typCii solution $\OPT_1^{\typCii}$ for $G_1$. Note that both dummy edges $uy$ and $vy$ must be part of $\Reduce(G_2^{\typCii})$ as otherwise $\Reduce(G_2^{\typCii})$ cannot be feasible for $G_2^{\typCii}$. By the same reason, $H_2^{\typCii}$ must be connected and spanning for $G_2$.
    
    We first consider the case where $vw \in \Reduce(G_2^{\typCii})$. 
    Let $C_1^{\typCii}(w)$ be the 2EC component of $\OPT_1^{\typCii}$ that contains $w$. Since $C_1^{\typCii}(w)$ is isolated in $\OPT_1^{\typCii}$, \Cref{lemma:3vc-spanning-tree} implies that there exists an edge $e$ in $G_1$ between $C_1^{\typCii}(w)$ and $V(G_1) \setminus C_1^{\typCii}(w)$. 
    Then, $\OPT_1^{\typCii} \cup \{e\}$ is connected and spanning for $G_1$, and every edge in $\OPT_1^{\typCii} \cup \{e\}$ is part of a 2EC component of $\OPT_1^{\typCii} \cup \{e\}$ or lies on a path in $\OPT_1^{\typCii} \cup \{e\}$ between any two distinct vertices of $\{u,v,w\}$ between which is a path in $H_2^{\typCii}$. 
    Similarly, every edge in $H_2^{\typCii}$ is part of a 2EC component of $H_2^{\typCii}$ or lies on a path in $H_2^{\typCii}$ between any two distinct vertices of $\{u,v,w\}$ between which is a path in $\OPT_1^{\typCii} \cup \{e\}$.
    Thus, $\OPT_1^{\typCii} \cup H_2^{\typCii} \cup F^{\typCii}$ with $F^{\typCii} \coloneq \{e\}$ is a 2EC spanning subgraph of $G$ by \Cref{lemma:3vc-make-2ec}.  Moreover, $|F^{\typCii}| - |\Reduce(G_2^{\typCii})| + |H_2^{\typCii}| = 1 - 3 = -2$.
    
    If $vw \notin \Reduce(G_2^{\typCii})$, then every edge in $\OPT_1^{\typCii}$ is part of a 2EC component of $\OPT_1^{\typCii}$ or lies on a path in $\OPT_1^{\typCii}$ between $u$ and $v$ between which is also a path in $H_2^{\typCii}$. 
    Further, every edge in $H_2^{\typCii}$ is part of a 2EC component of $H_2^{\typCii}$ or lies on a path in $H_2^{\typCii}$ between $u$ and $v$, and there is a path between $u$ and $v$ in $\OPT_1^{\typCii}$. 
    Note that there cannot be an edge in $H_2^{\typCii}$ that is on a path between $w$ and $u$ or $v$ but not in a 2EC component, because then $\Reduce(G_2^{\typCii})$ cannot be feasible for $G_2^{\typCii}$ as $vw \notin \Reduce(G_2^{\typCii})$.
    Thus, $\OPT_1^{\typCii} \cup H_2^{\typCii}$ is a 2EC spanning subgraph of $G$ by \Cref{lemma:3vc-make-2ec}. 
    We have that $|F^{\typCii}| - |\Reduce(G_2^{\typCii})| + |H_2^{\typCii}| = 0 - 2 = -2$.
\end{proof}

\begin{lemma}\label{lemma:3vc-c2-ii-F}
    If the condition in \Cref{reduce-3vc:c2-ii} in \Cref{alg:3-vertex-cuts} applies and $\Reduce(G_2^{\typCii})$ is a 2EC spanning subgraph of $G_2^{\typCii}$, then there exist edges $F^{\typCii}$ with $|F^{\typCii}| \leq 1$ and $\OPT_1^{\typCii}$ such that $H^\typCii$ is a 2EC spanning subgraph of $G$ and $|F^{\typCii}| - |\Reduce(G_2^{\typCii})| + |H_2^{\typCii}| \leq - 3$. 
\end{lemma}

\begin{proof}
    Let $D = \{uy,vz,zy,wy\}$. 
    First observe that $H_2^{\typCii}$ must be connected and spanning for $G_2$. 
    Note that $vz$ and $zy$ must be part of $\Reduce(G_2^{\typCii})$ to make $z$ incident to 2 edges, and additionally $vy$ or $wy$ must be part of $\Reduce(G_2^{\typCii})$ to make $y$ incident to at least 2 edges. 
    Thus, we can distinguish the following cases.

    If $\Reduce(G_2^{\typCii}) \cap D = \{uy,vz,zy\}$ (or by symmetry  $\Reduce(G_2^{\typCii}) \cap D = \{wy,vz,zy\}$), let $\OPT_1^{\typCii}$ be an optimal $\typCii$ solution for $G_1$ that contains a $C_1^{\typCii}(u)-C_1^{\typCii}(v)$ path, where $C_1^{\typCii}(x)$ denotes the 2EC component of $\OPT_1^{\typCii}$ that contains $x$, for $x \in \{u,v,w\}$.
    Note that every edge of $\OPT_1^{\typCii}$ is part of a 2EC component of $\OPT_1^{\typCii}$ or lies on a path in $\OPT_1^{\typCii}$ between $u$ and $v$ between which is also a path in $H_2^{\typCii}$.
    Further, every edge in $H_2^{\typCii}$ is part of a 2EC component of $H_2^{\typCii}$ or lies on a path in $H_2^{\typCii}$ between $u$ and $v$, and there is a path between $u$ and $v$ in $\OPT_1^{\typCii}$. Note that there cannot be an edge in $H_2^{\typCii}$ that is on a path between $w$ and $u$ or $v$ but not in a 2EC component, because then $\Reduce(G_2^{\typCii})$ cannot be feasible for $G_2^{\typCii}$ as $wy \notin \Reduce(G_2^{\typCii})$.
    Thus, $\OPT_1^{\typCii} \cup H_2^{\typCii}$ is a 2EC spanning subgraph of $G$ by \Cref{lemma:3vc-make-2ec}. We have that $|F^{\typCii}| - |\Reduce(G_2^{\typCii})| + |H_2^{\typCii}| = 0-3 = -3$.
    
    Otherwise, that is, $\Reduce(G_2^{\typCii}) \cap D = \{vy,vz,zy,wy\}$,
    let $\OPT_1^{\typCii}$ be an optimal \typCii solution for $G_1$ that contains a $C_1^{\typCii}(u)-C_1^{\typCii}(v)$ path of $\OPT_1^{\typCii}$, where $C_1^{\typCii}(x)$ denotes the 2EC component of $\OPT_1^{\typCii}$ that contains $x$, for $x \in \{u,v,w\}$.
    \Cref{lemma:3vc-spanning-tree} implies that there exists an edge $e$ in $G_1$ between $C_1^{\typCii}(w)$ and the connected component containing $V(G_1) \setminus V(C_1^{\typCii}(w))$. Therefore, $\OPT_1^{\typCii} \cup \{e\}$ is connected. 
    Thus, every edge in $\OPT_1^{\typCii} \cup \{e\}$ (resp.\ $H_2^{\typCii}$) is part of a 2EC component of $\OPT_1^{\typCii} \cup \{e\}$ ($H_2^{\typCii}$) or lies on a path in $\OPT_1^{\typCii} \cup \{e\}$  ($H_2^{\typCii}$) between any two distinct vertices of $\{u,v,w\}$ between which is a path in $H_2^{\typCii}$ ($\OPT_1^{\typCii} \cup \{e\}$).
    We conclude using \Cref{lemma:3vc-make-2ec} that $H_2^{\typCii} \cup \OPT_1^{\typCii} \cup F^{\typCii}$ with $F^{\typCii} \coloneq \{e\}$ is a 2EC spanning subgraph of $G$. We have that $|F^{\typCii}| - |\Reduce(G_2^{\typCii})| + |H_2^{\typCii}| = 1-4=-3$.
\end{proof}

\begin{lemma}\label{lemma:3vc-c2-iii-F}
    If the condition in \Cref{reduce-3vc:c2-iii} in \Cref{alg:3-vertex-cuts} applies  and $\Reduce(G_2^{\typCii})$ is a 2EC spanning subgraph of $G_2^{\typCii}$, then there exist edges $F^{\typCii}$ with $|F^{\typCii}| \leq 1$ and $\OPT_1^{\typCii}$ such that $H^\typCii$ is a 2EC spanning subgraph of $G$ and $|F^{\typCii}| - |\Reduce(G_2^{\typCii})| + |H_2^{\typCii}| \leq - 2$. 
\end{lemma}

\begin{proof}
     Let $D = \{uy,vy,wy\}$. 
     First observe that $H_2^{\typCii}$ must be connected and spanning for $G_2$. 
     Note that $|D \cap \Reduce(G_2^{\typCii})| \geq 2$ as $y$ needs to be 2EC in $\Reduce(G_2^{\typCii})$.
    
     If $|D \cap \Reduce(G_2^{\typCii})| = 2$, let w.l.o.g.\ by symmetry $D \cap \Reduce(G_2^{\typCii}) = \{uy,vy\}$ and $\OPT_1^{\typCii}$ be an optimal $\typCii$ solution for $G_1$ that contains a $C_1^{\typCii}(u)-C_1^{\typCii}(v)$ path, where $C_i^{\typCii}(x)$ denotes the 2EC component of $\OPT_1^{\typCii}$ that contains $x$, for $x \in \{u,v,w\}$.
     Note that every edge of $\OPT_1^{\typCii}$ is part of a 2EC component of $\OPT_1^{\typCii}$ or lies on a path in $\OPT_1^{\typCii}$ between $u$ and $v$ between which is also a path in $H_2^{\typCii}$.
     Further, every edge in $H_2^{\typCii}$ is part of a 2EC component of $H_2^{\typCii}$ or lies on a path in $H_2^{\typCii}$ between $u$ and $v$, and there is a path between $u$ and $v$ in $\OPT_1^{\typCii}$. Note that there cannot be an edge in $H_2^{\typCii}$ that is on a path between $w$ and $u$ or $v$ but not in a 2EC component, because then $\Reduce(G_2^{\typCii})$ cannot be feasible for $G_2^{\typCii}$ as $wy \notin \Reduce(G_2^{\typCii})$.
     Thus, $\OPT_1^{\typCii} \cup H_2^{\typCii}$ is a 2EC spanning subgraph of $G$ by \Cref{lemma:3vc-make-2ec}. 
     We have that $|F^{\typCii}| - |\Reduce(G_2^{\typCii})| + |H_2^{\typCii}| = 0-2=-2$.

     Otherwise, that is, $\Reduce(G_2^{\typCii}) \cap D = \{uy,vy,wy\}$,
     let $\OPT_1^{\typCii}$ be an optimal \typCii solution for $G_1$ that contains a $C_1^{\typCii}(u)-C_1^{\typCii}(v)$ path, where $C_i^{\typCii}(x)$ denotes the 2EC component of $\OPT_1^{\typCii}$ that contains $x$, for $x \in \{u,v,w\}$.
     \Cref{lemma:3vc-spanning-tree} implies that there exists an edge $e$ in $G_1$ between $C_1^{\typCii}(w)$ and the connected component containing $V(G_1) \setminus V(C_1^{\typCii}(w))$.
     Therefore, $\OPT_1^{\typCii} \cup \{e\}$ is connected. 
     Thus, every edge in $\OPT_1^{\typCii} \cup \{e\}$ (resp.\ $H_2^{\typCii}$) is part of a 2EC component of $\OPT_1^{\typCii} \cup \{e\}$ ($H_2^{\typCii}$) or lies on a path in $\OPT_1^{\typCii} \cup \{e\}$  ($H_2^{\typCii}$) between any two distinct vertices of $\{u,v,w\}$ between which is a path in $H_2^{\typCii}$ ($\OPT_1^{\typCii} \cup \{e\}$).
     We conclude using \Cref{lemma:3vc-make-2ec} that $H_2^{\typCii} \cup \OPT_1^{\typCii} \cup F^{\typCii}$ with $F^{\typCii} \coloneq \{e\}$ is a 2EC spanning subgraph of $G$. We have that $|F^{\typCii}| - |\Reduce(G_2^{\typCii})| + |H_2^{\typCii}| = 1-3=-2$.
\end{proof}

\begin{lemma}\label{lemma:3vc-c3-F}
    If the condition in \Cref{reduce-3vc:c3} in \Cref{alg:3-vertex-cuts} applies and $\Reduce(G_2^{\typCiii})$ is a 2EC spanning subgraph of $G_2^{\typCiii}$, then there exist edges $F^{\typCiii}$ with $|F^{\typCiii}| \leq 4$ such that $H^\typCiii$ is a 2EC spanning subgraph of $G$ and $|F^{\typCiii}| - |\Reduce(G_2^{\typCiii})| + |H_2^{\typCiii}| \leq 0$. 
\end{lemma}

\begin{proof}
    First note that the claimed $\typCiii$ solution in \Cref{reduce-3vc:c3-select-solution} is guaranteed by \Cref{lemma:3vc-spanning-tree} (among renaming $u$, $v$ and $w$).
    Let $C_2^{\typCiii}(x)$ be the 2EC component of $H_2^{\typCiii}$ that contains $x$, for $x \in \{u,v,w\}$.
We have 4 main cases.

    \textbf{Case 1:}
    $|\Reduce(G_2^{\typCiii}) \cap \{uv,vw\}| = \ell \leq 2$. First observe that $H_2^{\typCiii}$ is connected as otherwise $\Reduce(G_2^{\typCiii})$ cannot be feasible for $G_2^{\typCiii}$. Moreover, we can select $\ell$ edges $F^{\typCiii} \subseteq \{e^\typCiii_{uv},e^\typCiii_{vw}\}$ such that 
    $\OPT_1^{\typCiii} \cup F^{\typCiii}$ is connected.
    Further, note that every edge in $\OPT_1^{\typCiii} \cup F^{\typCiii}$ (resp.\ $H_2^{\typCiii}$) is part of a 2EC component of $\OPT_1^{\typCiii} \cup F^{\typCiii}$ ($H_2^{\typCiii}$) or lies on a path in $\OPT_1^{\typCiii} \cup F^{\typCiii}$  ($H_2^{\typCiii}$) between any two distinct vertices of $\{u,v,w\}$ between which is a path in $H_2^{\typCiii}$ ($\OPT_1^{\typCiii} \cup F^{\typCiii}$).
    Thus, \Cref{lemma:3vc-make-2ec} implies that $H_2^{\typCiii} \cup \OPT_1^{\typCiii} \cup F^{\typCiii}$ is a 2EC spanning subgraph of $G$.
    We have that $|F^{\typCiii}| - |\Reduce(G_2^{\typCiii})| + |H_2^{\typCiii}| = \ell - \ell = 0$.

    \textbf{Case 2:}
    $|\Reduce(G_2^{\typCiii}) \cap \{uv,uv\}| = 2$ (the case $|\Reduce(G_2^{\typCiii}) \cap \{vw,vw\}| = 2$ is symmetric).
    We have that $C_2^{\typCiii}(v) = C_2^{\typCiii}(w)$, as otherwise $\Reduce(G_2^{\typCiii})$ cannot be feasible for $G_2^{\typCiii}$. 
    If $H_2^{\typCiii}$ is 2EC, then clearly $H_2^{\typCiii} \cup \OPT_2^\typCiii$ is a 2EC spanning subgraph of $G$, and $|F^{\typCiii}| - |\Reduce(G_2^{\typCiii})| + |H_2^{\typCiii}| = 0 - 2 \leq 0$.

    If $C_2^{\typCiii}(u)$ is connected to $C_2^{\typCiii}(v)$ in $H_2^{\typCiii}$, then $H_2^\typCiii$ is connected. 
    Further, $\OPT_1^\typCiii \cup F^\typCiii$ with $F^\typCiii \coloneq \{e^\typCiii_{uv}\}$ is connected.
    Every edge in $\OPT_1^{\typCiii} \cup F^{\typCiii}$ is part of a 2EC component of $\OPT_1^{\typCiii} \cup F^{\typCiii}$ or lies on a path in $\OPT_1^{\typCiii} \cup F^{\typCiii}$ between $u$ and $v$, which are also connected in $H_2^{\typCiii}$. Note that there cannot be an edge in $\OPT_1^{\typCiii} \cup F^{\typCiii}$ on a path between $w$ and $u$ or $v$ that is not in a 2EC component by the choice of the $\typCiii$ solution $\OPT_1^{\typCiii}$.
    Moreover, every edge in $H_2^{\typCiii}$ is part of a 2EC component of $H_2^{\typCiii}$ or lies on a path in $H_2^{\typCiii}$ between $u$ and $v$ or $w$, which are also connected in $\OPT_1^{\typCiii} \cup F^{\typCiii}$. Note that there cannot be an edge in $H_2^{\typCiii}$ on a path between $v$ and $w$ that is not in a 2EC component, because $C_2^{\typCiii}(v) = C_2^{\typCiii}(w)$ is 2EC.
    Thus, \Cref{lemma:3vc-make-2ec} implies that $H_2^{\typCiii} \cup \OPT_1^{\typCiii} \cup F^{\typCiii}$ is a 2EC spanning subgraph of $G$, and $|F^{\typCiii}| - |\Reduce(G_2^{\typCiii})| + |H_2^{\typCiii}| = 1 - 2 \leq 0$.

    Otherwise, that is, $C_2^{\typCiii}(u)$ is not connected to $C_2^{\typCiii}(v)$ in $H_2^{\typCiii}$,
    by \Cref{lemma:3vc-spanning-tree} there must exist an edge $e$ in $G_2$ between $C_2^{\typCiii}(u)$ and $C_2^{\typCiii}(v)$ such that $H_2^{\typCiii} \cup \{e\}$ is connected. 
    Thus, we can analogously to the previous case show that 
    $H_2^{\typCiii} \cup \OPT_1^{\typCiii} \cup F^{\typCiii}$ with $F^{\typCiii} \coloneq \{e, e^\typCiii_{uv}\}$ is a 2EC spanning subgraph of $G$.
    We have that $|F^{\typCiii}| - |\Reduce(G_2^{\typCiii})| + |H_2^{\typCiii}| = 2 - 2 = 0$.

    \textbf{Case 3:}
    $|\Reduce(G_2^{\typCiii}) \cap \{uv,uv,vw\}| = 3$ (the case $|\Reduce(G_2^{\typCiii}) \cap \{uv,vw,vw\}| = 3$ is symmetric).
    Note that in this case $C_2^{\typCiii}(u)$ is connected to $C_2^{\typCiii}(v)$ in $H_2^{\typCiii}$ as otherwise $u$ and $w$ cannot be 2EC in $\Reduce(G_2^{\typCiii})$. 
    We have two cases.

    If $C_2^{\typCiii}(v)$ is connected to $C_2^{\typCiii}(w)$ in $H_2^{\typCiii}$, then $H_2^{\typCiii}$ is connected. 
    Further, $\OPT_1^{\typCiii} \cup F^{\typCiii}$ with $F^{\typCiii} \coloneq \{e^{\typCiii}_{uv},e^{\typCiii}_{vw}\}$ is connected and spanning for $G_1$.
    Thus, every edge in $\OPT_1^{\typCiii} \cup F^{\typCiii}$ (resp.\ $H_2^{\typCiii}$) is part of a 2EC component of $\OPT_1^{\typCiii} \cup F^{\typCiii}$ ($H_2^{\typCiii}$) or lies on a path in $\OPT_1^{\typCiii} \cup F^{\typCiii}$  ($H_2^{\typCiii}$) between any two distinct vertices of $\{u,v,w\}$ between which is a path in $H_2^{\typCiii}$ ($\OPT_1^{\typCiii} \cup F^{\typCiii}$).
    We conclude via \Cref{lemma:3vc-make-2ec} that $H_2^{\typCiii} \cup \OPT_1^{\typCiii} \cup F^{\typCiii}$ is a 2EC spanning subgraph of $G$, and  $|F^{\typCiii}| - |\Reduce(G_2^{\typCiii})| + |H_2^{\typCiii}| = 2-3\leq 0$.
   
    Otherwise, that is, $C_2^{\typCiii}(v)$ is not connected to $C_2^{\typCiii}(w)$ in $H_2^{\typCiii}$, \Cref{lemma:3vc-spanning-tree} guarantees the existence of an edge $e$ in $G_2$ between the vertices of $V(C_2^{\typCiii}(w))$ and the vertices of $V(H_2^{\typCiii}) \setminus V(C_2^{\typCiii}(w))$. 
    Therefore, $H_2^{\typCiii} \cup \{e\}$ is connected.
    Hence, we can analogously to the previous case show that $H_2^{\typCiii} \cup \OPT_1^{\typCiii} \cup F^{\typCiii}$ with $F^{\typCiii} \coloneq \{e,e^{\typCiii}_{uv},e^{\typCiii}_{vw}\}$ is a 2EC spanning subgraph of $G$. We have that $|F^{\typCiii}| - |\Reduce(G_2^{\typCiii})| + |H_2^{\typCiii}| = 3-3 = 0$.

    \textbf{Case 4:}
    $|\Reduce(G_2^{\typCiii}) \cap \{uv,uv,vw,vw\}| = 4$. We distinguish three cases.
    
    First, assume that $H_2^{\typCiii}$ is connected. We have that $\OPT_1^{\typCiii} \cup F^{\typCiii}$ with $F^{\typCiii} \coloneq \{e^{\typCiii}_{uv},e^{\typCiii}_{vw}\}$ is connected and spanning for $G_1$. 
    Thus, every edge in $\OPT_1^{\typCiii} \cup F^{\typCiii}$ (resp.\ $H_2^{\typCiii}$) is part of a 2EC component of $\OPT_1^{\typCiii} \cup F^{\typCiii}$ ($H_2^{\typCiii}$) or lies on a path in $\OPT_1^{\typCiii} \cup F^{\typCiii}$  ($H_2^{\typCiii}$) between any two distinct vertices of $\{u,v,w\}$ between which is a path in $H_2^{\typCiii}$ ($\OPT_1^{\typCiii} \cup F^{\typCiii}$).
    We conclude via \Cref{lemma:3vc-make-2ec} that $H_2^{\typCiii} \cup \OPT_1^{\typCiii} \cup F^{\typCiii}$ is a 2EC spanning subgraph of $G$, and  $|F^{\typCiii}| - |\Reduce(G_2^{\typCiii})| + |H_2^{\typCiii}| = 2-4\leq 0$.
    
    If $H_2^{\typCiii}$ is composed of two connected components, assume w.l.o.g.\ that $C_2^{\typCiii}(u)$ is connected to $C_2^{\typCiii}(v)$ in $H_2^{\typCiii}$ and the other connected component in $H_2^{\typCiii}$ is $C_2^{\typCiii}(w)$. 
    In this case, \Cref{lemma:3vc-spanning-tree} guarantees that there is an edge $e$ in $E(G_2) \setminus C_2^{\typCiii}(w)$ between the vertices of $C_2^{\typCiii}(w)$ and the vertices of $V(H_2^{\typCiii}) \setminus V(C_2^{\typCiii}(w))$. 
    Therefore, $H_2^{\typCiii} \cup \{e\}$ is connected.
    Hence, we can analogously to the previous case show that $H_2^{\typCiii} \cup \OPT_1^{\typCiii} \cup F^{\typCiii}$ with $F^{\typCiii} \coloneq \{e,e^{\typCiii}_{uv},e^{\typCiii}_{vw}\}$ is a 2EC spanning subgraph of $G$, and $|F^{\typCiii}| - |\Reduce(G_2^{\typCiii})| + |H_2^{\typCiii}| = 3-4\leq 0$.
    
    Otherwise, that is, $C_2^{\typCiii}(u)$, $C_2^{\typCiii}(v)$ and $C_2^{\typCiii}(w)$ are all pairwise disjoint, \Cref{lemma:3vc-spanning-tree} guarantees the existence of two edges $F \subseteq E(G_2)$ such that $H_2^{\typCiii} \cup F$ is connected.
    Hence, we can analogously to the previous case(s) show that $H_2^{\typCiii} \cup \OPT_1^{\typCiii} \cup F^{\typCiii}$
    with $F^{\typCiii} \coloneq F \cup \{e^{\typCiii}_{uv},e^{\typCiii}_{vw}\}$ is a 2EC spanning subgraph of $G$. We have that $|F^{\typCiii}| - |\Reduce(G_2^{\typCiii})| + |H_2^{\typCiii}| = 4-4 = 0$.
\end{proof}

Finally, we prove that the reductions of \cref{alg:3-vertex-cuts} preserve the approximation factor.
While the proof is quite long, it is of the same spirit of the analogue for \cref{alg:2-vertex-cuts} (cf.\ \cref{lemma:2vc-approx}) and easy to verify: We consider different cases depending on the reduction the algorithm uses and, if \Cref{alg:3-vertex-cuts} is executed, we distinguish for some cases further which solution types the actual optimal solution uses for both sides of the cut. For each case, we can easily derive bounds for the different parts of the solution that our algorithm produces, and put them together to obtain the overall bound, which is stated below.

\begin{lemma}\label{lemma:3vc-approx}
    Let $G$ be a 2-ECSS instance. If every recursive call to $\Reduce(G')$ in \Cref{alg:3-vertex-cuts} on input $G$ satisfies $\reduce(G') \leq \alpha \cdot \opt(G') + 4\varepsilon \cdot |V(G')| - 4$, then it holds that $\reduce(G) \leq \alpha \cdot \opt(G) + 4\varepsilon \cdot |V(G)| - 4$.
\end{lemma}

\begin{proof}
    Let $\OPT_i \coloneq \OPT(G) \cap E(G_i)$ and $\opt_i \coloneq \abs{\OPT_i}$ for $i \in \{1,2\}$.
    Note that $\opt(G) = \opt_1 + \opt_2$.
    For the remainder of this proof, all line references (if not specified differently) refer to \Cref{alg:3-vertex-cuts}. Note that \Cref{lemma:3vc-opt-1-large} gives $\opt_1 \geq \opt_1^{\min} \geq 8$. Furthermore, throughout the proof we require $\alpha \geq \frac 54$. We distinguish in the following which reduction the algorithm uses.

    \begin{description}
        \item[(\Cref{reduce-3vc:both-large}: $|V_1| > \frac{2}{\varepsilon} - 4$)] 
    Note that for $i \in \{1,2\}$, we have that $\OPT_i$ is a feasible solution for $G_i'$ and therefore $\opt_i \geq \opt(G_i')$.
    Using the induction hypothesis, \Cref{lemma:3vc-both-large-F} and $\abs{V(G_1)} + \abs{V(G_2)} = \abs{V_1} + \abs{V_2} + 2 = \abs{V(G)} -1$, we conclude that
    \[
    \reduce(G) = \abs{H'_1} + \abs{H'_2} + \abs{F'} \leq \alpha \cdot (\opt_1 + \opt_2) + 4\varepsilon(\abs{V_1} + \abs{V_2} + 2) -8 + 4 \leq \alpha \cdot \opt(G) + 4\varepsilon \cdot (\abs{V(G)}) - 4 \ .
    \]
    
        \item[(\Cref{reduce-3vc:b1}: $\OPT_1^{\typBi}$ exists and $\opt_1^{\typBi} \leq \opt_1^{\min} + 1$)]
        Note that $\opt(G_2^{\typBi}) \leq \opt_2$, because $u$, $v$, and $w$ are contracted into a single vertex in $G_2^{\typBi}$. 
        Thus, $\opt(G) = \opt_2 + \opt_1 \geq \opt(G_2^{\typBi}) + \opt_1$. 
        Moreover, $\opt_1^{\typBi} \leq \opt_1^{\min} + 1 \leq \opt_1 + 1$, and $\abs{F^{\typBi}} \leq 1$ by \Cref{lemma:3vc-b1-F}. Therefore, we conclude
        \begin{align*}
            \reduce(G) = \abs{H^{\typBi}} &\leq \opt_1^{\typBi} + \abs{H_2^{\typBi}} + \abs{F^{\typBi}}
            \leq (\opt_1 + 1) + (\alpha \cdot \opt_2(G_2^{\typBi})  + 4\varepsilon\abs{V(G_2^{\typBi})} - 4) + 1 \\
            &\leq \alpha \cdot \opt(G) + 4\varepsilon\abs{V(G)} - 4 - (\alpha - 1)\opt_1 + 2 \leq \alpha \cdot \opt(G) + 4\varepsilon\abs{V(G)} - 4 \ .
        \end{align*}
    
        \item[(\Cref{reduce-3vc:b2,reduce-3vc:c1}: $t^{\min} = \typBii$ or $t^{\min} = \typCi$)]
        Let $\tau \coloneq t^{\min}$.
        Note that $\opt(G_2^{\tau}) \leq \opt_2$, because $u$, $v$, and $w$ are contracted into a single vertex in $G_2^{\tau}$. Thus, $\opt(G) \geq \opt_2 + \opt_1(G_1^{\tau}) \geq \opt(G_2^{\tau}) + \opt_1(G_1^{\tau})$. Moreover, $\abs{F^{\tau}} \leq 2$ by \Cref{lemma:3vc-b2-F,lemma:3vc-c1-F}. Therefore, we conclude
        \begin{align*}
            \reduce(G) = \abs{H^{\tau}} &\leq \opt_1^{\tau} + \abs{H_2^{\tau}} + \abs{F^{\tau}}
            \leq \opt^{\tau}_1 + (\alpha \cdot \opt_2(G_2^{\tau})  + 4\varepsilon\abs{V(G_2^{\tau})} - 4) + 2 \\
            &\leq \alpha \cdot \opt(G) + 4\varepsilon\abs{V(G)} - 4 - (\alpha - 1)\opt^{\tau}_1 + 2 \leq \alpha \cdot \opt(G) + 4\varepsilon\abs{V(G)} - 4 \ .
        \end{align*}

\item[(\Cref{reduce-3vc:c2-general}: $t^{\min} = \typCii$)]
        In \Cref{alg:3-vertex-cuts} we have the following $\typCii$ subtypes, which use different reductions~$G_2^{\typCii}$ for the recursion (\Cref{reduce-3vc:c2-i,reduce-3vc:c2-ii,reduce-3vc:c2-iii}).
        \begin{description}
            \item[Subtype $\typCiisubba$ (\Cref{reduce-3vc:c2-i}):] \emph{Every $\OPT_1^{\typCii}$ solution contains a $C(u)-C(v)$ path.} In this case 
            the definition of \Cref{alg:3-vertex-cuts}, \Cref{lemma:3vc-c2-i-F}, and the induction hypothesis gives
            \begin{align}
            \reduce(G) = |H^{\typCii}| &\leq |\OPT_1^{\typCii}| + \abs{H_2^{\typCii}} + \abs{F^{\typCii}} \notag \leq \opt_1^{\typCii} + \reduce(G_2^{\typCii}) - 2 \\
            &\leq \opt_1^{\typCii} + \alpha \cdot \opt(G_2^{\typCii}) + 4\varepsilon\abs{V(G_2^{\typCii})} - 4 - 2 \ . \label{eq:3vc-c2-i}
            \end{align}
            \item[Subtype $\typCiisubb$ (\Cref{reduce-3vc:c2-ii}):] \emph{Every $\OPT_1^{\typCii}$ solution contains either a $C(u)-C(v)$ path or a $C(v)-C(w)$ path.} In this case 
            the definition of \Cref{alg:3-vertex-cuts}, \Cref{lemma:3vc-c2-ii-F}, and the induction hypothesis gives
            \begin{align}
            \reduce(G) = |H^{\typCii}| &\leq |\OPT_1^{\typCii}| + \abs{H_2^{\typCii}} + \abs{F^{\typCii}} \notag \leq \opt_1^{\typCii} + \reduce(G_2^{\typCii}) - 3 \\
            &\leq \opt_1^{\typCii} + \alpha \cdot \opt(G_2^{\typCii}) + 4\varepsilon\abs{V(G_2^{\typCii})} - 4 - 3 \ . \label{eq:3vc-c2-ii}
            \end{align} 
            \item[Subtype $\typCiisubc$ (\Cref{reduce-3vc:c2-iii}):] \emph{Every $\OPT_1^{\typCii}$ solution contains some path between $C(u)$, $C(v)$, and $C(w)$.} In this case 
            the definition of \Cref{alg:3-vertex-cuts}, \Cref{lemma:3vc-c2-iii-F}, and the induction hypothesis gives
            \begin{align}
            \reduce(G) = |H^{\typCii}| &\leq |\OPT_1^{\typCii}| + \abs{H_2^{\typCii}} + \abs{F^{\typCii}} \notag \leq \opt_1^{\typCii} + \reduce(G_2^{\typCii}) - 2 \\
            &\leq \opt_1^{\typCii} + \alpha \cdot \opt(G_2^{\typCii}) + 4\varepsilon\abs{V(G_2^{\typCii})} - 4 - 2 \ . \label{eq:3vc-c2-iii}
            \end{align}
        \end{description}
        We distinguish the following cases depending on which type combination of \Cref{lemma:3vc-type-combinations} is present for $(\OPT_1,\OPT_2)$.
        \begin{enumerate}[(a)]
            \item $(\typA, \{\typA,\typBi,\typBii,\typCi,\typCii,\typCiii\})$. Since $\OPT_1^{\typA}$ exists, \Cref{lemma:3vc-type-A-condition} gives $\opt_1 = \opt_1^{\typA} \geq \opt_1^{\typCii} + 3$.
            \begin{description}
            \item[Subtype \typCiisubba:] We have $\opt(G_2^{\typCii}) \leq \opt_2 + 5$ because we can turn $\OPT_2$ with the dummy edges $\{uy,vy,vw\}$ and at most two more edges, which are guaranteed by \Cref{lemma:3vc-spanning-tree} if required, into a solution for $G_2^{\typCii}$. Thus, \eqref{eq:3vc-c2-i} gives
            \begin{align*}
            \reduce(G) &\leq (\opt_1 - 3) + (\alpha \cdot (\opt_2 + 5) + 4\varepsilon\abs{V(G_2^{\typCii})} - 4 - 2 \\
                &\leq \alpha \cdot \opt(G) + 4\varepsilon\abs{V(G)} - 4 + (\alpha - 1)(5 - \opt_1) \leq \alpha \cdot \opt(G) + 4\varepsilon\abs{V(G)} - 4  \ .
            \end{align*} 
            \item[Subtype \typCiisubb:] We have $\opt(G_2^{\typCii}) \leq \opt_2 + 6$ because we can turn $\OPT_2$ with the dummy edges $\{uy,vz,zy,wy\}$ and at most two more edges, which are guaranteed by \Cref{lemma:3vc-spanning-tree} if required, into a solution for $G_2^{\typCii}$. Thus, \eqref{eq:3vc-c2-ii} gives
            \begin{align*}
            \reduce(G) &\leq (\opt_1 - 3) + (\alpha \cdot (\opt_2 + 6) + 4\varepsilon\abs{V(G_2^{\typCii})} - 4 - 3 \\
                &\leq \alpha \cdot \opt(G) + 4\varepsilon\abs{V(G)} - 4 + (\alpha - 1)(6 - \opt_1) \leq \alpha \cdot \opt(G) + 4\varepsilon\abs{V(G)} - 4  \ .
            \end{align*} 
            \item[Subtype \typCiisubc:] We have $\opt(G_2^{\typCii}) \leq \opt_2 + 5$ because we can turn $\OPT_2$ with the dummy edges $\{uy,vy,wy\}$ and at most two more edges, which are guaranteed by \Cref{lemma:3vc-spanning-tree} if required, into a solution for $G_2^{\typCii}$. Thus, \eqref{eq:3vc-c2-iii} gives
            \begin{align*}
            \reduce(G) &\leq (\opt_1 - 3) + (\alpha \cdot (\opt_2 + 5) + 4\varepsilon\abs{V(G_2^{\typCii})} - 4 - 2 \\
                &\leq \alpha \cdot \opt(G) + 4\varepsilon\abs{V(G)} - 4 + (\alpha - 1)(5 - \opt_1) \leq  \alpha \cdot \opt(G) + 4\varepsilon\abs{V(G)} - 4  \ .
            \end{align*}  
            \end{description}
    
            \item $(\typBi, \{\typA,\typBi,\typBii,\typCi,\typCii\})$.
            Since $\OPT_1^{\typBi}$ exists but the condition in \Cref{reduce-3vc:b1} did not apply, we conclude that $\opt_1^{\typBi} \geq \opt_1^{\min} + 2 = \opt_1^{\typCii} + 2$.
            \begin{description}
            \item[Subtype \typCiisubba:] 
            We have $\opt(G_2^{\typCii}) \leq \opt_2 + 4$ because we can turn $\OPT_2$ with the dummy edges $\{uy,vy, vw\}$ and at most one more edge, which is guaranteed by \Cref{lemma:3vc-spanning-tree} if required, into a solution for $G_2^{\typCii}$.
            Thus, continuing \eqref{eq:3vc-c2-i} gives
            \begin{align*}
            \reduce(G) &\leq \opt_1 - 2 + \alpha \cdot (\opt_2 + 4) + 4\varepsilon\abs{V(G_2^{\typCii})} - 4 - 2  \\
            &\leq \alpha \cdot \opt(G) + 4\varepsilon\abs{V(G)} - 4 + (\alpha - 1)(4 - \opt_1) \leq \alpha \cdot \opt(G) + 4\varepsilon\abs{V(G)} - 4 \ .
            \end{align*} 
            \item[Subtype \typCiisubb:] We have $\opt(G_2^{\typCii}) \leq \opt_2 + 5$ because we can turn $\OPT_2$ with the dummy edges $\{uy,vz,zy,wy\}$ and at most one more edge, which is guaranteed by \Cref{lemma:3vc-spanning-tree} if required, into a solution for $G_2^{\typCii}$.
            Thus,  \eqref{eq:3vc-c2-ii} gives
            \begin{align*}
            \reduce(G) &\leq \opt_1 - 2 + \alpha \cdot (\opt_2 + 5) + 4\varepsilon\abs{V(G_2^{\typCii})} - 4 - 3  \\
            &\leq \alpha \cdot \opt(G) + 4\varepsilon\abs{V(G)} - 4 + (\alpha - 1)(5 - \opt_1) \leq \alpha \cdot \opt(G) + 4\varepsilon\abs{V(G)} - 4 \ .
            \end{align*} 
            \item[Subtype \typCiisubc:] We have $\opt(G_2^{\typCii}) \leq \opt_2 + 4$ because we can turn $\OPT_2$ with the dummy edges $\{uy,vy,wy\}$ and at most one more edge, which is guaranteed by \Cref{lemma:3vc-spanning-tree} if required, into a solution for $G_2^{\typCii}$.
            Thus, \eqref{eq:3vc-c2-iii} gives
            \begin{align*}
            \reduce(G) &\leq \opt_1 - 2 + \alpha \cdot (\opt_2 + 4) + 4\varepsilon\abs{V(G_2^{\typCii})} - 4 - 2  \\
            &\leq \alpha \cdot \opt(G) + 4\varepsilon\abs{V(G)} - 4 + (\alpha - 1)(4 - \opt_1) \leq \alpha \cdot \opt(G) + 4\varepsilon\abs{V(G)} - 4 \ .
            \end{align*} 
            \end{description}
    
            \item $(\typBii, \{\typA,\typBi,\typBii\})$. In this case $\OPT_1^{\typBii}$ exists but $t^{\min} = \typCii$. Thus, the definition of \Cref{alg:3-vertex-cuts} and \Cref{def:3vc-ties} imply that $\opt_1^{\typCii} \leq \opt^{\typBii}_1 - 1 = \opt_1 - 1$. 
            Furthermore, in this case at least two vertices of $\{u, v, w\}$ are in a 2EC component in $\OPT_2$.
            \begin{description}
            \item[Subtype \typCiisubba:] We have $\opt(G_2^{\typCii}) \leq \opt_2 + 3$ because we can turn $\OPT_2$ with two of the dummy edges $\{uy,vy,vw\}$ and at most one more edge, which is guaranteed by \Cref{lemma:3vc-spanning-tree} if required, into a solution for $G_2^{\typCii}$.
            Thus, \eqref{eq:3vc-c2-i} gives
            \begin{align*}
            \reduce(G) &\leq \opt_1 - 1 + \alpha \cdot (\opt_2 + 3) + 4\varepsilon\abs{V(G_2^{\typCii})} - 4 - 2  \\
            &\leq \alpha \cdot \opt(G) + 4\varepsilon\abs{V(G)} - 4 + (\alpha - 1)(3 - \opt_1) \leq \alpha \cdot \opt(G) + 4\varepsilon\abs{V(G)} - 4 \ .
            \end{align*} 
            \item[Subtype \typCiisubb:] We have $\opt(G_2^{\typCii}) \leq \opt_2 + 4$ because we we can turn $\OPT_2$ with the dummy edges $\{vz,zy \}$, one of the dummy edges $\{uy, wy\}$, and at most one more edge, which is guaranteed by \Cref{lemma:3vc-spanning-tree} if required, into a solution for $G_2^{\typCii}$.
            Thus, \eqref{eq:3vc-c2-ii} gives
            \begin{align*}
            \reduce(G) &\leq \opt_1 - 1 + \alpha \cdot (\opt_2 + 4) + 4\varepsilon\abs{V(G_2^{\typCii})} - 4 - 3  \\
            &\leq \alpha \cdot \opt(G) + 4\varepsilon\abs{V(G)} - 4 + (\alpha - 1)(4 - \opt_1) \leq \alpha \cdot \opt(G) + 4\varepsilon\abs{V(G)} - 4 \ .
            \end{align*} 
            \item[Subtype \typCiisubc:] We have $\opt(G_2^{\typCii}) \leq \opt_2 + 3$ because we can turn $\OPT_2$ with two of the dummy edges $\{uy,vy,wy\}$ and at most one more edge, which is guaranteed by \Cref{lemma:3vc-spanning-tree} if required, into a solution for $G_2^{\typCii}$.
            Thus, \eqref{eq:3vc-c2-iii} gives
            \begin{align*}
            \reduce(G) &\leq \opt_1 - 1 + \alpha \cdot (\opt_2 + 3) + 4\varepsilon\abs{V(G_2^{\typCii})} - 4 - 2  \\
            &\leq \alpha \cdot \opt(G) + 4\varepsilon\abs{V(G)} - 4 + (\alpha - 1)(3 - \opt_1) \leq \alpha \cdot \opt(G) + 4\varepsilon\abs{V(G)} - 4 \ .
            \end{align*} 
            \end{description}

            \item $(\typCi, \{\typA,\typBi,\typCi\})$. 
            In this case $\opt_1^{\typCi}$ exists but $t^{\min} = \typCii$. Thus, the definition of \Cref{alg:3-vertex-cuts} and \Cref{def:3vc-ties} imply that $\opt_1^{\typCii} \leq \opt^{\typCi}_1 - 1 = \opt_1 - 1$.
            \begin{description}
            \item[Subtype \typCiisubba:] We have $\opt(G_2^{\typCii}) \leq \opt_2 + 3$ because we can turn $\OPT_2$ with the dummy edges $\{uy,vy,vw\}$ into a solution for $G_2^{\typCii}$.
            Thus, continuing \eqref{eq:3vc-c2-i} gives
            \begin{align*}
            \reduce(G) &\leq \opt_1 - 1 + \alpha \cdot (\opt_2 + 3) + 4\varepsilon\abs{V(G_2^{\typCii})} - 4 - 2  \\
            &\leq \alpha \cdot \opt(G) + 4\varepsilon\abs{V(G)} - 4 + (\alpha - 1)(3 - \opt_1) \leq \alpha \cdot \opt(G) + 4\varepsilon\abs{V(G)} - 4 \ .
            \end{align*} 
            \item[Subtype \typCiisubb:] We have $\opt(G_2^{\typCii}) \leq \opt_2 + 4$ because we can turn $\OPT_2$ with the dummy edges $\{uy,vz,zy,wy\}$ into a solution for $G_2^{\typCii}$.
            Thus, continuing \eqref{eq:3vc-c2-ii} gives
            \begin{align*}
            \reduce(G) &\leq \opt_1 - 1 + \alpha \cdot (\opt_2 + 4) + 4\varepsilon\abs{V(G_2^{\typCii})} - 4 - 3  \\
            &\leq \alpha \cdot \opt(G) + 4\varepsilon\abs{V(G)} - 4 + (\alpha - 1)(4 - \opt_1) \leq \alpha \cdot \opt(G) + 4\varepsilon\abs{V(G)} - 4 \ .
            \end{align*} 
            \item[Subtype \typCiisubc:] We have $\opt(G_2^{\typCii}) \leq \opt_2 + 3$ because we can turn $\OPT_2$ with the dummy edges $\{uy,vy,wy\}$ into a solution for $G_2^{\typCii}$.
            Thus, continuing \eqref{eq:3vc-c2-iii} gives
            \begin{align*}
            \reduce(G) &\leq \opt_1 - 1 + \alpha \cdot (\opt_2 + 3) + 4\varepsilon\abs{V(G_2^{\typCii})} - 4 - 2  \\
            &\leq \alpha \cdot \opt(G) + 4\varepsilon\abs{V(G)} - 4 + (\alpha - 1)(3 - \opt_1) \leq \alpha \cdot \opt(G) + 4\varepsilon\abs{V(G)} - 4 \ .
            \end{align*} 
            \end{description}
            
          \item $(\typCii, \{\typA,\typBi\})$. We have $\opt_1^{\typCii} = \opt_1$.
          \begin{description}
            \item[Subtype \typCiisubba:] We have $\opt(G_2^{\typCii}) \leq \opt_2 + 3$ because we can turn $\OPT_2$ with the dummy  edges $\{uy,vy,vw\}$ into a solution for $G_2^{\typCii}$.
            Thus, continuing \eqref{eq:3vc-c2-i} gives
            \begin{align*}
            \reduce(G) &\leq \opt_1 + \alpha \cdot (\opt_2 + 3) + 4\varepsilon\abs{V(G_2^{\typCii})} - 4 - 2  \\
            &\leq \alpha \cdot \opt(G) + 4\varepsilon\abs{V(G)} - 4 + (\alpha - 1)(2 - \opt_1) + \alpha \\
            &\leq \alpha \cdot \opt(G) + 4\varepsilon\abs{V(G)} - 4 \ .
            \end{align*} 
            \item[Subtype \typCiisubb:] We have $\opt(G_2^{\typCii}) \leq \opt_2 + 4$ because we can turn $\OPT_2$ with the dummy edges $\{uy,vz,zy,wy\}$ into a solution for $G_2^{\typCii}$.
            Thus, continuing \eqref{eq:3vc-c2-ii} gives
            \begin{align*}
            \reduce(G) &\leq \opt_1 + \alpha \cdot (\opt_2 + 4) + 4\varepsilon\abs{V(G_2^{\typCii})} - 4 - 3  \\
            &\leq \alpha \cdot \opt(G) + 4\varepsilon\abs{V(G)} - 4 + (\alpha - 1)(3 - \opt_1) + \alpha \\
            &\leq \alpha \cdot \opt(G) + 4\varepsilon\abs{V(G)} - 4 \ .
            \end{align*} 
           \item[Subtype \typCiisubc:] We have $\opt(G_2^{\typCii}) \leq \opt_2 + 3$ because we can turn $\OPT_2$ with the dummy edges $\{uy,vy,wy\}$ into a solution for $G_2^{\typCii}$.
            Thus, continuing \eqref{eq:3vc-c2-iii} gives
            \begin{align*}
            \reduce(G) &\leq \opt_1 + \alpha \cdot (\opt_2 + 3) + 4\varepsilon\abs{V(G_2^{\typCii})} - 4 - 2  \\
            &\leq \alpha \cdot \opt(G) + 4\varepsilon\abs{V(G)} - 4 + (\alpha - 1)(2 - \opt_1) + \alpha \\ 
            &\leq \alpha \cdot \opt(G) + 4\varepsilon\abs{V(G)} - 4 \ .
            \end{align*} 
            \end{description}
    
            \item $(\typCiii, \typA)$. We have $\opt_1^{\typCii} = \opt_1^{\min} \leq \opt_1^{\typCiii} = \opt_1$.
            \begin{description}
            \item[Subtype \typCiisubba:] We have $\opt(G_2^{\typCii}) \leq \opt_2 + 2$ because we can turn $\OPT_2$ with the dummy edges $\{uy,vy\}$ into a solution for $G_2^{\typCii}$.
            Thus, continuing \eqref{eq:3vc-c2-i} gives
            \begin{align*}
            \reduce(G) &\leq \opt_1 + \alpha \cdot (\opt_2 + 2) + 4\varepsilon\abs{V(G_2^{\typCii})} - 4 - 2  \\
            &\leq \alpha \cdot \opt(G) + 4\varepsilon\abs{V(G)} - 4 + (\alpha - 1)(2 - \opt_1) \\ 
            &\leq \alpha \cdot \opt(G) + 4\varepsilon\abs{V(G)} - 4 \ .
            \end{align*} 
            \item[Subtype \typCiisubb:] We have $\opt(G_2^{\typCii}) \leq \opt_2 + 3$ because we can turn $\OPT_2$ with the dummy edges $\{uy,vz,zy\}$ into a solution for $G_2^{\typCii}$.
            Thus, continuing \eqref{eq:3vc-c2-ii} gives
            \begin{align*}
            \reduce(G) &\leq \opt_1 + \alpha \cdot (\opt_2 + 3) + 4\varepsilon\abs{V(G_2^{\typCii})} - 4 - 3  \\
            &\leq \alpha \cdot \opt(G) + 4\varepsilon\abs{V(G)} - 4 + (\alpha - 1)(3 - \opt_1) \\
            &\leq \alpha \cdot \opt(G) + 4\varepsilon\abs{V(G)} - 4 \ .
            \end{align*} 
            \item[Subtype \typCiisubc:] We have $\opt(G_2^{\typCii}) \leq \opt_2 + 2$ because we can turn $\OPT_2$ with the dummy edges $\{uy,vy\}$ into a solution for $G_2^{\typCii}$.
            Thus, continuing \eqref{eq:3vc-c2-iii} gives
            \begin{align*}
            \reduce(G) &\leq \opt_1 + \alpha \cdot (\opt_2 + 2) + 4\varepsilon\abs{V(G_2^{\typCii})} - 4 - 2  \\
            &\leq \alpha \cdot \opt(G) + 4\varepsilon\abs{V(G)} - 4 + (\alpha - 1)(2 - \opt_1) \\
            &\leq \alpha \cdot \opt(G) + 4\varepsilon\abs{V(G)} - 4 \ .
            \end{align*} 
            \end{description}
        \end{enumerate}

\item[(\Cref{reduce-3vc:c3}: $t^{\min} = \typCiii$)] 
        Let $D^{\typCiii} \coloneq \{uv,uv,vw,vw\}$ denote the set of dummy edges used for the construction of $G_2^\typCiii$.
        First observe that the definition of \Cref{alg:3-vertex-cuts}, \Cref{lemma:3vc-c3-F}, and the induction hypothesis gives
         \begin{align}
            \reduce(G) = |H^{\typCiii}| &\leq |\OPT_1^{\typCiii}| + \abs{H_2^{\typCiii}} + \abs{F^{\typCiii}} \leq \opt_1^{\typCiii} + \reduce(G_2^{\typCiii}) \notag \\ 
            &\leq \opt_1^{\typCiii} + \alpha \cdot \opt(G_2^{\typCiii}) + 4\varepsilon\abs{V(G_2^{\typCiii})} - 4 \ . \label{eq:3vc-c3}
         \end{align}
        
        We distinguish the following cases depending on which type combination of \Cref{lemma:3vc-type-combinations} is present for $(\OPT_1,\OPT_2)$.
        \begin{enumerate}[(a)]
            \item $(\typA, \{\typA,\typBi,\typBii,\typCi,\typCii,\typCiii\})$. Since $\OPT_1^{\typA}$ exists, \Cref{lemma:3vc-type-A-condition} gives $\opt_1 = \opt_1^{\typA} \geq \opt_1^{\typCiii} + 3$. Moreover, $\opt(G_2^{\typCiii}) \leq \opt_2 + 4$ because we can turn $\OPT_2$ with at most four dummy edges of $D^{\typCiii}$ into a solution for $G_2^{\typCiii}$. Thus, \eqref{eq:3vc-c3} gives
            \begin{align*}
            \reduce(G) &\leq \opt_1 - 3 + \alpha \cdot (\opt_2 + 4) + 4\varepsilon\abs{V(G_2^{\typCiii})} - 4 \\
            &\leq \alpha \cdot \opt(G) + 4\varepsilon\abs{V(G)} - 4 + (\alpha - 1)(3 - \opt_1) + \alpha \leq \alpha \cdot \opt(G) + 4\varepsilon\abs{V(G)} - 4 \ . 
            \end{align*} 
            
            \item $(\typBi, \{\typA,\typBi,\typBii,\typCi,\typCii\})$.
            Since $\OPT_1^{\typBi}$ exists but the condition in \Cref{reduce-3vc:b1} did not apply, we conclude that $\opt_1 = \opt_1^{\typBi} \geq \opt_1^{\min} + 2 = \opt_1^{\typCiii} + 2$.
             Moreover, $\opt(G_2^{\typCiii}) \leq \opt_2 + 3$ because we can turn $\OPT_2$ with at most three dummy edges of $D^{\typCiii}$ into a solution for $G_2^{\typCiii}$.  Thus, \eqref{eq:3vc-c3} gives
            \begin{align*}
                \reduce(G) &\leq \opt_1 - 2 + \alpha \cdot (\opt_2 + 3) + 4\varepsilon\abs{V(G_2^{\typCiii})} - 4 \\
                &\leq \alpha \cdot \opt(G) + 4\varepsilon\abs{V(G)} - 4 + (\alpha - 1)(2 - \opt_1) + \alpha \leq \alpha \cdot \opt(G) + 4\varepsilon\abs{V(G)} - 4 \ .
            \end{align*} 
            
            \item $(\typBii, \{\typA,\typBi,\typBii\})$. 
            In this case $\OPT_1^{\typBii}$ exists but $t^{\min} = \typCiii$. Thus, the definition of \Cref{alg:3-vertex-cuts} and \Cref{def:3vc-ties} imply that $\opt_1^{\typCiii} \leq \opt^{\typBii}_1 - 1 = \opt_1 - 1$.
            Moreover, $\opt(G_2^{\typCiii}) \leq \opt_2 + 2$ because we can turn $\OPT_2$ with at most two dummy edges of $D^{\typCiii}$ into a solution for $G_2^{\typCiii}$.  Thus, \eqref{eq:3vc-c3} gives
            \begin{align*}
            \reduce(G) &\leq \opt_1 - 1 + \alpha \cdot (\opt_2 + 2) + 4\varepsilon\abs{V(G_2^{\typCiii})} - 4  \\
            &\leq \alpha \cdot \opt(G) + 4\varepsilon\abs{V(G)} - 4 + (\alpha - 1)(1 - \opt_1) + \alpha \leq \alpha \cdot \opt(G) + 4\varepsilon\abs{V(G)} - 4 \ .
            \end{align*} 
            
            \item $(\typCi, \{\typA,\typBi,\typCi\})$. 
            In this case $\OPT_1^{\typCi}$ exists but $t^{\min} = \typCiii$. Thus, the definition of \Cref{alg:3-vertex-cuts} and \Cref{def:3vc-ties} imply that $\opt_1^{\typCiii} \leq \opt^{\typCi}_1 - 1 = \opt_1 - 1$.
            Moreover, $\opt(G_2^{\typCiii}) \leq \opt_2 + 2$ because we can turn $\OPT_2$ with at most two dummy edges of $D^{\typCiii}$ into a solution for $G_2^{\typCiii}$. Thus, \eqref{eq:3vc-c3} gives 
            \begin{align*}
            \reduce(G) &\leq \opt_1 - 1 + \alpha \cdot (\opt_2 + 2) + 4\varepsilon\abs{V(G_2^{\typCiii})} - 4 \\
            &\leq \alpha \cdot \opt(G) + 4\varepsilon\abs{V(G)} - 4 + (\alpha - 1)(1 - \opt_1) + \alpha \leq \alpha \cdot \opt(G) + 4\varepsilon\abs{V(G)} - 4 \ .
            \end{align*}

            \item $(\typCii, \{\typA,\typBi\})$. 
            In this case $\OPT_1^{\typCii}$ exists but $t^{\min} = \typCiii$. Thus, the definition of \Cref{alg:3-vertex-cuts} and \Cref{def:3vc-ties} imply that $\opt_1^{\typCiii} \leq \opt^{\typCii}_1 - 1 = \opt_1 - 1$.
            Moreover, $\opt(G_2^{\typCiii}) \leq \opt_2 + 1$ because we can turn $\OPT_2$ with at most one dummy edge of $D^{\typCiii}$ into a solution for $G_2^{\typCiii}$.  Thus, \eqref{eq:3vc-c3} gives
            \begin{align*}
            \reduce(G) &\leq \opt_1 - 1 + \alpha \cdot (\opt_2 + 1) + 4\varepsilon\abs{V(G_2^{\typCiii})} - 4 \\
            &\leq \alpha \cdot \opt(G) + 4\varepsilon\abs{V(G)} - 4 + (\alpha - 1)(1 - \opt_1) \leq \alpha \cdot \opt(G) + 4\varepsilon\abs{V(G)} - 4 \ .
            \end{align*}       
            
            \item $(\typCiii, \typA)$. Using $\opt_1^{\typCiii} = \opt_1$ and $\opt(G_2^{\typCiii}) \leq \opt_2$ as $\OPT_2$ is already feasible for $G_2^{\typCiii}$, \eqref{eq:3vc-c3} implies
            \begin{align*}
            \reduce(G) \leq \opt_1 + \alpha \cdot \opt_2 + 4\varepsilon\abs{V(G_2^{\typCiii})} - 4  \leq  \alpha \cdot \opt(G) + 4\varepsilon\abs{V(G)} - 4 \ .
            \end{align*}    
        \end{enumerate}
    \end{description}
    
    This completes the proof of the lemma.    
\end{proof}

\subsection{Proof of \Cref{lem:reduction-to-structured}}\label{sec:reduce-proofs}

In this final subsection, we prove \Cref{lem:reduction-to-structured}. The statement is implied by the following three lemmas.
We first show that $\Reduce(G)$ runs in polynomial time. The proof is similar to the corresponding proof in \cite{GargGA23improved}.

\begin{lemma}[running time]\label{lem:reduction-to-structured-runtime}
    For any 2-ECSS instance $G$, $\Reduce(G)$ runs in polynomial time in $|V(G)|$ if $\ALG(G)$ runs in polynomial time in $|V(G)|$. 
\end{lemma}

\begin{proof}
    Let $n = |V(G)|$ and $m = |V(G)|$, and $\sigma = n^2 + m^2$ be the size of the problem. We first argue that every non-recursive step of \Cref{alg:reduce,alg:2-vertex-cuts,alg:3-vertex-cuts} can be performed in time polynomial in $\sigma$.
    Clearly, all computations performed in \Cref{alg:reduce} can be performed in polynomial time in $\sigma$. In particular, it is possible to check in polynomial-time whether $G$ contains an $\alpha$-contractible subgraph with at most $\frac{4}{\varepsilon}$ vertices: Enumerate over all subsets $W$ of vertices of the desired size; for each such $W$, one can compute in polynomial time a minimum-size 2-ECSS $\OPT'$ of $G[W]$, if any such subgraph exists. Furthermore, one can compute in polynomial time a maximum cardinality subset of edges $F$ with endpoints in $W$ such that $G \setminus F$ is 2EC. Then, it is sufficient to compare the size of $\OPT'$, if exists, with $|E(G[W]) \setminus F|$.
    
    We now consider \Cref{alg:3-vertex-cuts}.
    We already argued that the computation in \Cref{reduce-3vc:compute-opt-types} can be executed in polynomial time in $\sigma$. Moreover, determining the $\typCii$ subcases in \Cref{reduce-3vc:c2-i,reduce-3vc:c2-ii,reduce-3vc:c2-iii} and selecting a desired $\OPT_1^{\typCii}$ solution in \Cref{reduce-3vc:c2-ii-sol-constr,reduce-3vc:c2-iii-sol-constr} and $\OPT_1^{\typCiii}$ solution in \Cref{reduce-3vc:c3-select-solution} can be done in constant time, as $G_1$ is of constant size in this case.
    Finding the corresponding sets of edges $F$ in \Cref{alg:3-vertex-cuts} can also be done in polynomial time in $\sigma$ by enumeration, because \Cref{lemma:3vc-both-large-F,lemma:3vc-b1-F,lemma:3vc-b2-F,lemma:3vc-c1-F,lemma:3vc-c2-i-F,lemma:3vc-c2-ii-F,lemma:3vc-c2-iii-F,lemma:3vc-c3-F} guarantee that there exists a desired set of edges $F$ of constant size.
    The same arguments can also be used to argue that  \cref{alg:2-vertex-cuts} can be implemented in polynomial time.

    To bound the total running time including recursive steps, let $T(\sigma)$ be the running time of \Cref{alg:reduce} as a function of $\sigma$. In each step, starting with a problem of size $\sigma$, we generate either a single subproblem of size $\sigma' < \sigma$ (to easily see this in \Cref{alg:2-vertex-cuts,alg:3-vertex-cuts}, recall that $2 \leq |V_1| \leq |V_2|$ and $7 \leq |V_1| \leq |V_2|$, respectively) or two subproblems of size $\sigma'$ and $\sigma''$ such that $\sigma',\sigma'' < \sigma$ and $\sigma' + \sigma'' \leq \sigma$. Thus, $T(\sigma) \leq \max\{T(\sigma'), T(\sigma') + T(\sigma'') \} + \mathrm{poly}(\sigma)$ with $\sigma',\sigma'' < \sigma$ and $\sigma' + \sigma'' \leq \sigma$. Therefore, $T(\sigma) \leq \sigma \cdot \mathrm{poly}(\sigma)$.
\end{proof}

Next, we show that $\Reduce(G)$ returns a feasible solution to the 2-ECSS problem on $G$.

\begin{lemma}[correctness]\label{lem:reduction-to-structured-correctness}
    For any 2-ECSS instance $G$, $\Reduce(G)$ returns a 2EC spanning subgraph of $G$.
\end{lemma}

\begin{proof}
The proof is by induction over the tuple $(\abs{V(G)}, \abs{E(G)})$. The base cases are given by \Cref{reduce:bruteforce,reduce:call-alg} of \Cref{alg:reduce}.
If we enter \Cref{reduce:bruteforce}, the claim trivially holds. If we call $\ALG$ in \Cref{reduce:call-alg}, the claim follows by the correctness of $\ALG$ for $(\alpha,\varepsilon)$-structured instances.

If the condition in \Cref{reduce:1vc}
applies, let $G_i = G[V_i \cup \{v\}]$. Note that both $G_1$ and $G_2$, which must be 2EC since $G$ is 2EC, each contain at most $|V(G)| - 1$ vertices, hence by induction hypothesis $\Reduce(G_1)$ and $\Reduce(G_2)$ are 2EC spanning subgraphs of $V_1 \cup \{v\}$ and $V_2 \cup \{v\}$, respectively. Thus, their union is a 2EC spanning subgraph of $G$.
If the conditions in \Cref{reduce:loop} or \Cref{reduce:irrelevant} apply, $G' \coloneq G \setminus \{e\}$ must be 2EC by \cref{lemma:self-loops-parallel-edges,lemma:irrelevant-edge}, respectively, and $|V(G')| = |V(G)|$ and $|E(G')| = |E(G)| - 1$. Hence, by induction hypothesis $\Reduce(G')$ is a 2EC spanning subgraph of $G'$, and thus also for $G$.
If the condition in \Cref{reduce:contractible} applies, the claim follows from \Cref{fact:decontract}.

It remains to consider the cases when \cref{alg:2-vertex-cuts} or \cref{alg:3-vertex-cuts} is executed.
Observe that in these cases, $G$ is 2VC since the condition of \cref{reduce:1vc} is checked before. Moreover, note that all recursive calls to $\Reduce$ made in \cref{alg:2-vertex-cuts} and \cref{alg:3-vertex-cuts} consider graphs with strictly fewer vertices than $G$, hence the induction hypothesis applies.
If \cref{alg:reduce} executes \cref{alg:2-vertex-cuts},
the correctness follows from the induction hypothesis, \cref{lemma:2vc-type-A-condition} and \cref{lemma:2vc-both-large-F,lemma:2vc-b-F,lemma:2vc-c-F}.
If \cref{alg:reduce} executes \Cref{alg:3-vertex-cuts}, 
the correctness follows from the induction hypothesis, \Cref{lemma:2vc-type-A-condition} and \Cref{lemma:3vc-b1-F,lemma:3vc-b2-F,lemma:3vc-c1-F,lemma:3vc-c2-i-F,lemma:3vc-c2-ii-F,lemma:3vc-c2-iii-F,lemma:3vc-c3-F}.
\end{proof}

The following lemma implies $\reduce(G) \leq (\alpha + 4\varepsilon) \opt(G)$ using the trivial bound $\opt(G) \geq |V(G)|$.

\begin{lemma}[approximation guarantee]\label{lem:reduction-to-structured-apx}
For any 2-ECSS instance $G$ it holds that
    \[
    \reduce(G) \leq \begin{cases}
        \opt(G) &\quad \text{if } \abs{V(G)} \leq \frac{4}{\varepsilon} \ , \text{ and} \\
        \alpha \cdot \opt(G) + 4\varepsilon \cdot \abs{V(G)} - 4 &\quad \text{if } \abs{V(G)} > \frac{4}{\varepsilon} \ . 
    \end{cases}
    \]
\end{lemma}
\begin{proof}
The proof is by induction over the tuple $(\abs{V(G)}, \abs{E(G)})$. The base cases are given by \Cref{reduce:bruteforce,reduce:call-alg} of \Cref{alg:reduce}.
If we enter \Cref{reduce:bruteforce}, the stated bound trivially holds. 
If the condition in \cref{reduce:loop} or \cref{reduce:irrelevant} applies, we consider a graph $G' \coloneq G \setminus \{e\}$ with the same number of vertices as $G$ and the same optimal size of a 2-ECSS by \Cref{lemma:self-loops-parallel-edges} and \cref{lemma:irrelevant-edge}, respectively. Thus, the claim follows from the induction hypothesis. If we call $\ALG$ in \Cref{reduce:call-alg}, the claim follows because we assume that $\ALG(G)$ returns an $\alpha$-approximate solution to the 2-ECSS problem on $G$ and $G$ is in this case $(\alpha,\varepsilon)$-structured by the correctness of $\Reduce$ (cf.\ \Cref{lem:reduction-to-structured-correctness}).

If the condition in \cref{reduce:1vc} applies, recall that by \cref{lemma:cut-vertex-egal} we have $\opt(G) = \opt(G_1) + \opt(G_2)$, where $G_i \coloneq G[V_i \cup \{v\}]$ for $i \in \{1,2\}$. We distinguish a few cases depending on the sizes of $V_1$ and $V_2$. Assume w.l.o.g.\ $|V_1| \leq |V_2|$. 
If $|V_2| \leq \frac{4}{\varepsilon}-1$, then $\reduce(G) = \opt(G_1) + \opt(G_2) = \opt(G) \leq \alpha \cdot \opt(G) + 4\varepsilon |V(G)| - 4$ since $\alpha \geq 1$ and $|V(G)| \geq \frac{4}{\varepsilon}$.
Otherwise, $|V_2| \geq \frac{4}{\varepsilon}$.
If $|V_1| \leq \frac{4}{\varepsilon}-1$, then by the induction hypothesis we have $\reduce(G) \leq \opt(G_1) + \alpha \opt(G_2) + 4\varepsilon |V(G_2)| - 4 \leq\alpha \opt(G) + 4\varepsilon |V(G)| - 4$ using $\alpha \geq 1$ and $|V(G_2)| \leq |V(G)|$.
Otherwise, $|V_2| \geq |V_1| \geq \frac{4}{\varepsilon}$, and the induction hypothesis gives
$\reduce(G) \leq \alpha \opt(G_1) + \alpha \opt(G_2) + 4\varepsilon (|V(G_1)| + |V(G_2)|) - 4 \leq \alpha \opt(G) + 4\varepsilon |V(G)| + 4\varepsilon - 8$, which implies the stated bound as $4\varepsilon \leq 4$.

If the condition in \cref{reduce:contractible} applies,
we have $\opt(G) = |\OPT(G) \cap {V(H)}^2| + |\OPT(G) \setminus {V(H)}^2| \geq \frac{1}{\alpha}|H| + \opt(G|H)$ by the definition of a contractible subgraph $H$ and observing that $\OPT(G) \setminus {V(H)}^2$ induces a feasible 2EC spanning subgraph of $G|H$.
If $|V(G|H)| \leq \frac{4}{\varepsilon}$ one has $\reduce(G) = |H| + \opt(G|H) \leq \alpha \opt(G) \leq \alpha \opt(G) + 4\varepsilon|V(G)| - 4$ since $|V(G)| \geq \frac{4}{\varepsilon}$ whenever \cref{alg:reduce} reaches \cref{reduce:contractible}. 
Otherwise, we conclude that $\reduce(G) \leq |H| + \alpha \opt(G|H) + 4\varepsilon |V(G|H)| - 4 \leq \alpha \opt(G) + 4\varepsilon |V(G)| - 4$ since $|V(G)| \geq |V(G|H)|$.

Finally, if \cref{alg:reduce} executes \cref{alg:2-vertex-cuts} or
\Cref{alg:3-vertex-cuts}, we have that $|V(G)| > \frac{4}{\varepsilon}$. Thus, in these cases the statement follows from the induction hypothesis and \cref{lemma:2vc-approx} or \cref{lemma:3vc-approx}, respectively.
\end{proof}

\section{Canonical 2-Edge Cover}\label{app:canonical}

This section is dedicated to the proof of \Cref{lem:canonical-cover:main}, which we first restate.

\lemmaCanonicalMain*

\begin{proof}
We first compute a triangle-free $2$-edge cover $H'$ (which can be done in polynomial time due to \Cref{lem:triangle-free-2-edge-cover}) and show how to modify it to obtain a canonical $2$-edge cover with the same number of edges.
    Since $H'$ is a triangle-free $2$-edge cover, each component is either a 2EC component that contains at least $4$ edges, or complex. 

    If $H'$ is canonical, we are done. Otherwise, we apply the following algorithm to $H'$: 
    While there exist sets of edges $F_A \subseteq E \setminus H'$ and $F_R \subseteq H'$ with $|F_A| \leq |F_R| \leq 2$ such that $H'' \coloneq (H' \setminus F_R) \cup F_A$ is a triangle-free $2$-edge cover and either 
    \begin{itemize}
        \item 
        $|H''| < |H'|$, or
        \item  $|H''| = |H'|$ and either
        \begin{itemize}
            \item $H''$ contains strictly fewer connected components than $H'$, or
            \item the number of components of $H''$ is equal to the number of components in $H'$ and $H''$ contains strictly fewer bridges than $H'$, or
            \item the number of components of $H''$ is equal to the number of components in $H'$, the number of bridges in $H''$ is equal to the number of bridges in $H'$, and $H''$ contains strictly fewer cut vertices in 2EC components than $H'$,
        \end{itemize} 
    \end{itemize} 
    we update $H' \coloneq H''$.

    Clearly, each iteration of this algorithm can be executed in polynomial time. Moreover, there can only be polynomially many iterations.
    Thus, let $H$ be equal to $H'$ when the algorithm terminates. 
    We prove that $H$ is a canonical 2-edge cover of $G$. 
    This proves the lemma because the algorithm does not increase the number of edges compared to a minimum triangle-free $2$-edge cover of $G$. 
    To this end, we consider the following cases.

    \textbf{Case 1:} $H$ contains a 2EC component $C$ such that $4 \leq |E(C)| \leq 7$ and $C$ is not a cycle on $|E(C)|$ vertices. We distinguish three exhaustive cases, and prove in any case that Case~1 cannot occur.
    
    \textbf{Case 1.1:}
    $C$ has a cut vertex $v$. Observe that $v$ must have degree $4$ in $C$.
    Furthermore, each connected component in $C \setminus \{v\}$ contains at least $2$ vertices, as otherwise $C$ contains parallel edges, a contradiction to $G$ being structured.
    There must be a connected component $K_1$ in $C \setminus \{v\}$ such that $K_1$ contains at most 2 vertices (and hence at most 1 edge). 
    This is true as otherwise each connected component has at least 2 edges in $C \setminus \{v\}$ and since there are at least $2$ connected components in $C \setminus \{v\}$, because $v$ is a cut vertex, $C$ must contain at least $8$ edges, a contradiction.
    Assume that there is a component $K_1$ with $2$ vertices and let $u_1$ and $u_2$ be these vertices.
    Note that the three edges $\{v u_1, u_1 u_2, u_2 v\}$ must be in $C$ as otherwise $G$ is not structured since it must contain parallel edges.
    If there exists an edge $u_1 w \in E \setminus H$ (or $u_2 w \in E$ by symmetry) for $w \in V \setminus V(C)$, then we can remove $v u_1$ from $H$ and add $u_1 w$ to obtain a triangle-free $2$-edge cover $H''$ with strictly fewer components and $|H| = |H''|$, a contradiction to the assumption that the algorithm terminated. 
    If there exists an edge $u_1 w \in E \setminus H$ (or $u_2 w \in E$ by symmetry) for $w \in V(C) \setminus \{v, u_1, u_2 \}$, then we can remove $v u_1$ from $H$ and add $u_1 w$ to obtain a triangle-free $2$-edge cover $H''$ with $|H''| = |H|$ in which $v$ is not a cut vertex anymore. Since we have not created a new cut vertex (and the number of components and number of bridges stays the same), this contradicts the assumption that the algorithm terminated.
    If neither of these two cases apply, we conclude $\deg_G(u_i) = 2$ for $i=1,2$. Thus, $v - u_1 - u_2 - v$ is a contractible cycle, a contradiction to $G$ being structured.

    \textbf{Case 1.2:}
    $C$ contains no 2-vertex cut (and, thus, $C$ is 3-vertex-connected). This implies that every vertex has degree at least 3 in $C$, as otherwise the neighboring vertices of a vertex of $C$ with degree $2$ form a $2$-vertex cut.
    But then it must hold that $|V(C)| = 4$, since $|V(C)| \leq 3$ is not possible with the above degree constraint and if $|V(C)| \geq 5$, then the number of edges in $C$ is at least $8$, a contradiction.
    But if $|V(C)| = 4$ and $C$ is $3$-vertex-connected, $C$ must be a complete graph on $4$ vertices. Therefore, there is an edge $e \in E(C)$ such that $H \setminus \{e\}$ is a triangle-free 2-edge cover, a contradiction to the assumption that the algorithm terminated. 

    \textbf{Case 1.3:} $C$ contains a 2-vertex cut $\{ v_1, v_2 \}$.
    Note that each connected component of $C \setminus \{v_1, v_2 \}$ must have an edge to both $v_1$ and $v_2$, as otherwise $C$ contains a cut vertex, and we are in Case 1.1.
    Hence, there are at most $3$ connected components in $C \setminus \{v_1, v_2 \}$, as otherwise the number of edges in $C$ is at least 8.
    
    First, assume that $C \setminus \{v_1, v_2 \}$ contains exactly $3$ connected components.
    Then there can be at most one connected component in $C \setminus \{v_1, v_2 \}$ containing 2 vertices, as otherwise the total number of edges in $C$ exceeds $7$.
    In any case, each vertex $u$ of some connected component in $C \setminus \{v_1, v_2 \}$ must have an edge $e_u$ to $\{v_1, v_2 \}$ as otherwise $C$ contains a cut vertex, and we are actually in Case 1.1.
    If there is a vertex $u$ of some connected component in $C \setminus \{v_1, v_2 \}$ that has an edge $u w$ in $G$ to some vertex $w \in V \setminus V(C)$, then we can replace $e_u$ with $uw$ in $H$ and thereby decrease the number of connected components, a contradiction to the algorithm.
    Hence, assume that this is not true. 
    Assume that there is exactly one connected component $K_1$ in $C \setminus \{v_1, v_2 \}$ containing 2 vertices $x_1$ and $x_2$. 
    Note that $K_1$ must consist of the single edge $x_1 x_2$ as otherwise we have at least $8$ edges in $C$.
    Let $K_2$ and $K_3$ be the other two connected components in $C \setminus \{v_1, v_2 \}$ consisting of single vertices $y$ and $z$, respectively. By the above arguments, we have that $C$ contains the edges $\{v_1 x_1, x_1 x_2, x_2 v_2, v_1 y, y v_2, v_1 z, z v_2\}$ (up to symmetry).
    If there is an edge $x_1 y$ in $G$ (or any of $\{x_1 z, x_2 y, x_2 z, y z\}$ by symmetry), then we can remove $v_2 y$ and $x_1 v_1$ from $H$ and add $x_1 y$ to obtain a triangle-free $2$-edge cover with strictly fewer edges than $H$, a contradiction to the algorithm. Hence, assume that none of the edges exists in $G$. 
    Now the vertices in $\{y, z\}$ have degree $2$ in $G$ and the vertices in $\{x_1, x_2 \}$ are only adjacent to $\{v_1, v_2, x_1, x_2 \}$. But then $C$ is contractible, a contradiction to $G$ being structured.
    The case that there is no component in $C \setminus \{v_1, v_2 \}$ containing 2 vertices, i.e., in which each component consists of exactly a single vertex, is similar.
    This finishes the case that in $C \setminus \{v_1, v_2 \}$ there are exactly 3 connected components.

    Hence, assume that $C \setminus \{v_1, v_2 \}$ contains exactly 2 connected components $K_1$ and $K_2$ and without loss of generality assume $|E(K_1)| \geq |E(K_2)|$.
    Recall that each connected component of $C \setminus \{v_1, v_2 \}$ must have an edge to both $v_1$ and $v_2$.
    Hence, $\{v_1, v_2 \}$ is incident to at least 4 edges in $C$. Thus, $|E(K_2)| \leq |E(K_1)| \leq 3$.
    If $|E(K_1)| = 3$, we prove that $K_1$ is a path on 4 vertices. 
    If $K_1$ is not a path on $4$ vertices, then $K_1$ is a triangle on 3 vertices or a star on 4 vertices (one center and 3 leaves). 
    If $K_1$ is a triangle, there must be an edge $e$ in $C$ such that $H \setminus \{e\}$ is a triangle-free $2$-edge cover, a contradiction to the algorithm.
    If $K_1$ is a star on 4 vertices, then there must be 3 edges from $K_1$ to $\{v_1, v_2 \}$, as otherwise $H$ is not a 2-edge cover. But then $C$ contains at least 8 edges, a contradiction.
    Hence, if $|E(K_1)| = 3$, then $K_1$ is a path on 4 vertices.
    But then one leaf of $K_1$ must be adjacent to $v_1$ in $C$ and the other leaf of $K_1$ must be adjacent to $v_2$ in $C$ (otherwise, there is a cut vertex and we are in Case 1.1). Since $K_2$ is a single vertex in this case and must be incident to both $v_1$ and $v_2$, $C$ is a cycle on $7$ vertices, a contradiction.
    Hence, we have that $|E(K_2)| \leq |E(K_1)| \leq 2$. 
    Since $G$ does not contain parallel edges, $K_1$ must be a simple path on $i$ vertices for some $i \in \{1, 2, 3\}$. The same is true for $K_2$.
    By similar reasoning as before we can observe that the leaves of $K_1$ and $K_2$, respectively, must be incident to $v_1$ and $v_2$ (or, if $K_1$ or $K_2$ is a single vertex, then it must be incident to both $v_1$ and $v_2$) and therefore we observe that $C$ is a simple cycle, a contradiction.

    \textbf{Case 2:}
There is some complex component $C$ of $H$ containing a pendant block $B$ with less than 6 edges.
    First, assume there is a pendant block $B$ with exactly 3 edges.
    In this case, $B$ is a cycle $b_1 - b_2 - b_3 - b_1$ (since there are no parallel edges as $G$ is structured), and assume that $b_1$ is incident to the unique bridge $f$ in $C$.
    However, since $B$ is a triangle and $G$ is structured, there must be an edge $e = b_2 w$ incident to $b_2$ in $G$ with $w \notin \{b_1, b_3 \}$ (by invoking the 3-matching \Cref{lem:3-matching} to $V(B)$ and $V \setminus V(B)$).
    But then we can remove $b_1 b_2$ and add $e$ to $H$ to obtain $H''$ and observe that  $H''$ is a triangle-free 2-edge cover.
    If $w \in V \setminus V(C)$, the number of connected components in $H''$ is strictly less compared to $H$, and otherwise, that is, $w \in V(C)$, $H''$ has strictly fewer bridges compared to $H$ (and the number of connected component stays the same). 
    This is a contradiction to the algorithm. 

    Next, assume there is a pendant block $B$ with exactly 4 edges.
    In this case, $B$ is a cycle $b_1 - b_2 - b_3 - b_4 - b_1$ (since there are no parallel edges as $G$ is structured), and assume that $b_1$ is incident to the unique bridge $f$ in $C$.
    If $b_2$ or $b_4$ are incident to an edge $e$ incident to some $w \in V \setminus V(B)$, then we can remove either $b_1 b_2$ or $b_1 b_4$ and add $e$ to $H$ to obtain $H''$ and observe that $H''$ is a triangle-free 2-edge cover.
    If $w \in V \setminus V(C)$, the number of connected components in $H''$ is strictly less compared to $H$, and otherwise, that is, $w \in V(C)$, $H''$ has strictly fewer bridges compared to $H$ (and the number of connected component stays the same). 
    This is a contradiction to the algorithm. 
    Hence, this is not the case. But then $\{b_1, b_3 \}$ is a non-isolating 2-vertex, a contradiction to $G$ being structured.
    
    Finally, assume there is a pendant block $B$ with exactly 5 edges. 
Note that $B$ must contain either 4 or 5 vertices.
    If $B$ contains exactly $4$ vertices then $B$ must contain a simple cycle on 4 vertices with an additional edge $e$ being a chord of this cycle, which is redundant, and, thus, a contradiction to the algorithm.
Hence, assume that $B$ contains exactly $5$ vertices, i.e., $B$ is a cycle $b_1 - b_2 - b_3 - b_4 - b_5 - b_1$, and assume that $b_1$ is incident to the unique bridge $f$ in $C$.
    If either $b_2$ (or $b_5$ by symmetry) is incident to an edge $e$ with neighbor in $V \setminus V(B)$ we can remove $b_1 b_2$ from $H$ and add $e$ to obtain a triangle-free $2$-edge cover $H''$. 
    If $e$ is incident to some vertex in $V \setminus V(C)$, the number of connected components in $H''$ is strictly less compared to $H$, and otherwise, that is, $e$ is incident to some vertex in $V(C) \setminus V(B)$, $H''$ has strictly fewer bridges compared to $H$ (and the number of connected component stays the same). 
    This is a contradiction to the algorithm. Thus, in the following we can assume that $b_2$ and $b_5$ are only adjacent to vertices of $V(B)$.
    Since $G$ is structured, $b_3$ and $b_4$ both must be incident to some edge $e_3$ and $e_4$ going to some vertex in $V \setminus V(B)$, respectively, as otherwise $B$ is contractible.
    Moreover, there must be the edge $b_2 b_5$ in $G$, as otherwise $B$ is contractible.
    But then we can remove $b_2 b_3$ and $b_1 b_5$ from $H$ and add $b_2 b_5$ and $e_3$ to obtain  a triangle-free $2$-edge cover $H''$.
    If $e_3$ is incident to some vertex $V \setminus V(C)$, the number of connected components in $H''$ is strictly less compared to $H$, and otherwise, that is, $e_3$ is incident to some vertex $V(C) \setminus V(B)$, $H''$ has strictly fewer bridges compared to $H$ (and the number of connected component stays the same). 
    This is a contradiction to the algorithm.

   \textbf{Case 3:}
    There is some complex component $C$ of $H$ containing a non-pendant block $B$ that contains less than $4$ edges, that is, $B$ is a triangle on the vertices $b_1 - b_2 - b_3$ (otherwise $B$ contains parallel edges, a contradiction to $G$ being structured).
    In this case, observe that either there must be an edge $e$ in $B$ such that $H \setminus \{e\}$ is a triangle-free $2$-edge cover, a contradiction to the algorithm, or all edges in $H$ incident to $V(B)$ and $V(G) \setminus V(B)$ are incident to one vertex, say $b_1$.
    But then, applying \Cref{lem:3-matching}, we know that there must be an edge $b_2 w \in E(G) \setminus E(H)$, where $w \in V(G) \setminus V(B)$.
    But then we can remove $b_1 b_2$ from $H$ and add $b_2 w$ to obtain a triangle-free $2$-edge cover $H''$.
    If $w \in V \setminus V(C)$, the number of connected components in $H''$ is strictly less compared to $H$, and otherwise, that is, $w \ in V(C) \setminus V(B)$, $H''$ has strictly fewer bridges compared to $H$ (and the number of connected component stays the same). 
    This is a contradiction to the algorithm.
    
    Finally, observe that if $B$ contains $4$ edges but less than $4$ vertices, then there must be some parallel edge in $B$, a contradiction to the fact that $G$ is structured.
\end{proof}

\section{Bridge Covering}
\label{app:bridge-covering}

In this section, we prove \Cref{lem:bridge-covering-start2}, which is our main lemma for bridge covering.

\lemmaBridgeCoveringMain*

We prove this lemma by exhaustively applying the following lemma, which is essentially proved in~\cite{GargGA23improved} in Appendix~D.1.

\begin{lemma}[Lemma 2.10 in~\cite{GargGA23improved}]
    \label{lem:bridge-covering-main-iterative}
    Let $H$ be a canonical $2$-edge cover of some structured graph $G$ containing at least one bridge.
    There is a polynomial-time algorithm that computes a canonical 2-edge cover $H'$ of $G$ such that the number of bridges in $H'$ is strictly less than the number of bridges in $H$ and $\cost(H') \leq \cost(H)$.
\end{lemma}

For completeness we restate the proof of \Cref{lem:bridge-covering-main-iterative} from~\cite{GargGA23improved} here but adapt the proof to our notation wherever needed. Specifically, note that our definition of structured graphs is stronger than in \cite{GargGA23improved}, and we prove a smaller approximation guarantee, hence have a different credit value.

Let $C$ be any connected component of $H$ containing at least one bridge. Thus, $C$ is complex (cf.~\Cref{def:complex}).
Let $G_C$ be the multi-graph obtained from $G$ by contracting each block $B$ of $C$ and each connected component $C''$ of $H$ other than $C$ into a single node (cf.~\cref{def:complex} for the definition of block). 
Let $T_C$ be the tree in $G_C$ induced by the bridges of $C$: we call the vertices of $T_C$ corresponding
to blocks \emph{block nodes}, and the remaining vertices of $T_C$ \emph{lonely vertices}.
Observe that the leaves of $T_C$ are necessarily block nodes (otherwise $H$ would not be a 2-edge cover).
We refer to vertices of $G_C$ that arise from blocks or components as \emph{nodes} and refer to vertices of $G_C$ that correspond to vertices in $G$ as vertices.
If it is not clear whether a vertex of $G_C$ is a node or a vertex, we call it a vertex.

At a high level, we will transform $H$ into a new solution $H'$ containing a component $C'$ spanning the vertices of $C$ (and possibly some vertices of components from $H$ distinct from $C$). 
Furthermore, no new bridge is created and at least one bridge $e$ of $C$ is not a bridge of $C'$ (intuitively, the bridge $e$ gets covered).
During the process of \emph{covering} bridges we want to maintain that the new solution $H'$ is canonical, i.e., each 2EC component is an $i$-cycle, for $4 \leq i \leq 7$ or contains at least $8$ edges.
Furthermore, for every complex component, each of its pendant blocks contains at least $6$ edges and $6$ vertices and each of its non-pendant blocks contains at least $4$ edges and $4$ vertices.
We remark that each component of the initial (canonical) $2$-edge cover $H$ that is not 2EC contains at least $12$ vertices, and we only possibly merge together components in this stage of the construction. 
As a consequence, the new solution is also canonical.

A \emph{bridge-covering path} $P_C$ is any path in $G_C \setminus E(T_C)$ with its (distinct) endpoints $u$ and $v$ in $T_C$, and the remaining (internal) vertices outside $T_C$. 
Notice that $P_C$ might consist of a single edge, possibly parallel to some edge in $E(T_C)$. 
Augmenting $H$ along $P_C$ means adding the edges of $P_C$ to $H$, hence obtaining a new $2$-edge cover
$H'$. 
Notice that $H'$ obviously has fewer bridges than $H$: in particular all the bridges of $H$ along the $u$-$v$ path in $T_C$ are not bridges in $H'$ (we also informally say that such bridges are removed), and the bridges of $H'$ are a subset of the bridges of $H$. 
If we are able to show that we can find a bridge-covering path such that the resulting solution satisfies $\cost(H') \leq \cost(H)$, we are done: $H'$ now is canonical, has fewer bridges, and satisfies $\cost(H') \leq \cost(H)$. This satisfies the requirements of \Cref{lem:bridge-covering-main-iterative}.

Hence, it remains to analyze $\cost(H')$. 
Suppose that the distance between $u$ and $v$ in $T_C$ is $br$ (i.e., the path contains $br$ many bridges) and such a path contains $bl$ many blocks. Then the number of edges w.r.t.\ $H$ grows by $|E(P_C )|$. 
The number of credits w.r.t.\ $H$ decreases by at least $\frac{1}{4} br + bl + |E(P_C)| - 1$ since we remove $br$ bridges, $bl$ blocks, and $|E(P_C)| - 1$ components (each one having at least one credit). 
However, the number of credits also grows by 1 since we create a new block $B'$,  which needs $1$ credit, or a new 2EC component $C'$, which needs $1$ additional credit w.r.t.\ the credit of $C$.
Altogether $\cost(H) - \cost(H') \geq \frac{1}{4} br + bl - 2$. 
We say that $P_C$ is \emph{cheap} if the latter quantity is non-negative, and \emph{expensive} otherwise. 
In particular, $P_C$ is cheap if it involves at least $2$ block nodes or $1$ block node and at least $4$
bridges. 
Notice that a bridge-covering path, if at least one such path exists, can be computed in polynomial time; specifically, one can use breadth first search after truncating $T_C$.

Before proving \Cref{lem:bridge-covering-main-iterative}, we need the following two technical lemmas. 
We say that a vertex $v \in V(T_C) \setminus \{u \}$ is \emph{reachable} from $u \in V(T_C)$ if there exists a bridge-covering path between $v$ and $u$. 
Let $R(W)$ be the vertices in $V(T_C) \setminus W$ reachable from some vertex in $W \subseteq V(T_C)$, and let us use $R(u) \coloneq R(\{u\})$ for $u \in V(T_C)$. 
Notice that $v \in R(u)$ if and only if $u \in R(v)$.

\begin{lemma}
    \label{lem:bridge-covering-helper1}
    Let $e = xy \in E(T_C)$ and let $X_C$ and $Y_C$ be the two sets of vertices of the two trees obtained from $T_C$ after removing the edge $e$, where $x \in X_C$ and $y \in Y_C$. 
    Then $R(X_C)$ contains a block node or $R(X_C) \setminus \{y\}$ contains at least 2 lonely vertices.
\end{lemma}

\begin{proof}
    Let us assume that $R(X_C)$ contains only lonely vertices, as otherwise the claim holds. 
    Let $X$ be the vertices in $G_C$ that are connected to $X_C$ after removing $Y_C$, and let $Y$ be the remaining vertices in $G_C$. Let $X'$ and $Y'$ be the vertices in $G$ corresponding to $X$ and $Y$, respectively.
    Let $e'$ be the edge in $G$ that corresponds to $e$ and let $y'$ be the vertices of $Y'$ that is incident to $e'$. 
    In particular, if $y$ is a lonely vertex then $y' = y$, and otherwise $y'$ belongs to the block of $C$ corresponding to $y$. 
    
    Observe that both $X_C$ and $Y_C$ (hence $X$ and $Y$) each contain at least one (pendant) block node, hence both $X'$ and $Y'$ contain at least 6 vertices\footnote{Pendant blocks contain at least 6 vertices since $H$ is canonical, and the bridge-covering stage does not create smaller pendant blocks.}. 
    Therefore, we can apply the $3$-matching \Cref{lem:3-matching} to $X'$ and obtain a matching $M = \{ u_1' v_1', u_2' v_2', u_3' v_3' \}$ between $X'$ and $Y'$, where $u'_i \in X'$ and $v_i' \in Y'$ for $i \in \{ 1, 2, 3 \}$.
     Let $u_i$ and $v_i$ be the vertices in $G_C$ corresponding to $u'_i$ and $v'_i$, respectively.
     We remark that $v_i \in Y_C$: indeed otherwise $v_i$ would be connected to $X_C$ in $G_C \setminus Y_C$, contradicting the definition of $X$. 
     Let us show that the $v_i$'s are all distinct. 
     Assume by contradiction that $v_i = v_j$ for $i \neq j$. 
     Since $v'_i \neq v_j'$ (since $M$ is a matching) and $v'_i, v_j'$ are both associated with $v_i$, this mean that $v_i$ is a block node in $R(X_C)$, a contradiction since we assumed that $R(X_C)$ only contains lonely vertices.

    For each $u_i$ there exists a path $P_{w_i, u_i}$ in $G_C \setminus E(T_C)$ between $u_i$ and some $w_i \in X_C$ (possibly $w_i = u_i$).
    Observe that $P_{w_i, u_i} \cup \{ u_i v_i \}$  is a bridge-covering path between $w_i \in X_C$ and $v_i \in Y_C$ unless $u_i v_i = xy$ (recall that the edges $E(T_C)$ cannot be used in a bridge-covering path). 
    In particular, since the $v_i$'s are all distinct, at least two such paths are bridge-covering paths from $X_C$ to distinct (lonely) vertices of $Y_C \setminus \{ y \}$, implying $|R(X_C) \setminus \{y \}| \geq 2$, and thus proving our claim.
\end{proof}

Let $T_C(u, v)$ denote the path in $T_C$ between vertices $u$ and $v$.

\begin{lemma}
    \label{lem:bridge-covering-helper2}
    Let $b$ and $b'$ be two pendant (block) nodes of $T_C$. 
    Let $u \in R(b)$ and $u' \in R(b')$ be vertices of $V (T_C) \setminus \{b, b' \}$. 
    Suppose that $T_C(b, u)$ and $T_C(b', u')$ both contain some vertex $w$ (possibly $w = u = u'$) and $|E(T_C(b, u)) \cup E(T_C(b', u'))| \geq 4$. 
    Then in polynomial time one can find a 2-edge cover $H'$ satisfying the conditions of \Cref{lem:bridge-covering-main-iterative}.
\end{lemma}

\begin{proof}
Let $P_{bu}$ (resp. $P_{b'u'}$) be a bridge-covering path between $b$ and $u$ (resp., $b'$ and $u'$). 
Suppose that $P_{bu}$ and $P_{b'u'}$ share an internal vertex. 
Then there exists a (cheap) bridge-covering path between $b$ and $b'$, and we can find the desired $H'$ by the previous discussion (in particular, since there is a bridge-covering path containing the 2 block nodes $b$ and $b'$). 
So next assume that $P_{bu}$ and $P_{b'u'}$ are internally vertex disjoint. 
Let $H' \coloneq H \cup E(P_{bu}) \cup E(P_{b'u'})$. 
All the vertices and bridges induced by $E(P_{bu}) \cup E(P_{b'u'}) \cup E(T_C(b, u)) \cup E(T_C(b', u'))$ become part of the same block or 2EC component $C'$ of $H'$. 
Furthermore $C'$ contains at least $4$ bridges of $C$ and the two blocks $B$ and $B'$ corresponding to $b$ and $b'$, resp. 
One has $|H'| = |H| + |E(P_{bu})| + |E(P_{b'u'})|$. 
Moreover, $\credit(H') \leq \credit(H) + 1 - (|E(P_{bu})| - 1) - (|E(P_{b'u'})| - 1) - 2 - \frac{1}{4} \cdot 4$. 
In the latter inequality, the $+1$ is due to the extra credit required by $C'$, the $-2$ due to the removed blocks $B$ and $B'$ (i.e., blocks of $H$ that are not blocks of $H'$), and
the $-\frac{1}{4} \cdot 4$ due to the removed bridges. 
Altogether $\cost(H) - \cost(H') \geq \frac{1}{4} \cdot 4 + 2 -3 = 0$.
Hence, we observe that $H'$ has fewer bridges than $H$ and $H$ is indeed a canonical 2-edge cover satisfying $\cost(H') \leq \cost(H)$, proving the lemma.
\end{proof}

We are now ready to prove \Cref{lem:bridge-covering-main-iterative}.

\begin{proof}[Proof of \Cref{lem:bridge-covering-main-iterative}]
If there exists a cheap bridge-covering path $P_C$ (condition that we can check in
polynomial time), we simply augment $H$ along $P_C$ hence obtaining the desired $H'$.
Thus we next assume that no such path exists.
Let $P' = b, u_1, \ldots, u_\ell$ in $T_C$ be a longest path in $T_C$ (interpreted as a sequence of vertices). 
Notice that $b$ must be a leaf of $T_C$, hence a block node (corresponding to some leaf block $B$ of $C$). 
Let us consider $R(b)$. Since by assumption there is no cheap bridge-covering path, $R(b)$ does not contain any block node.
Hence $|R(b) \setminus \{u_1\} | \geq 2$ and $R(b)$ contains only lonely vertices by \Cref{lem:bridge-covering-helper1} (applied to $xy = bu_1$).

We next distinguish a few subcases depending on $R(b)$. 
Let $V_i$, $i \geq 1$, be the vertices in $V(T_C) \setminus V(P')$ such that their path to $b$ in $T_C$ passes through $u_i$ and not through $u_{i+1}$. 
Notice that $\{ V(P'), V_1, \ldots, V_\ell \}$ is a partition of $V(T_C)$. 
We observe that any vertex in $V_i$ is at distance at most $i$ from $u_i$ in $T_C$ as otherwise $P'$ would not be a longest path in $T_C$. 
We also observe that as usual the leaves of $T_C$ in $V_i$ are block nodes.
This implies that (a) all the vertices in $V_1$ are block nodes, and all the vertices in $V_2$ are block nodes or are lonely non-pendant vertices at distance $1$ from $u_2$ in $T_C$.

\textbf{Case 1:} There exists $u \in R(b)$ with $u \notin \{ u_1, u_2, u_3 \} \cup V_1 \cup V_2$. By definition there exists a bridge-covering path between $b$ and $u$ containing at least $4$ bridges, hence cheap.
This is excluded by the previous steps of the construction.

\textbf{Case 2:} There exists $u \in R(b)$ with $u \in V_1 \cup V_2$. 
Since $u$ is not a block node, by (a) $u$ must be a lonely non-pendant vertex in $V_2$ at distance $1$ from $u_2$. 
Furthermore, $V_2$ must contain at least one pendant block node $b'$ adjacent to $u$.
Consider $R(b')$. 
By the assumption that there are no cheap bridge-covering paths and \Cref{lem:bridge-covering-helper1} (applied to $xy = b'u$), $|R(b') \setminus \{u \}| \geq 2$. 
In particular, $R(b')$ contains at least one lonely vertex $u' \notin \{b', b, u \}$. The tuple $(b, b', u, u')$ satisfies the conditions of \Cref{lem:bridge-covering-helper2} (specifically, both $T_C(b, u)$ and $T_C(b', u')$ contain $w = u_2$), hence we can obtain the desired $H'$.

\textbf{Case 3:} $R(b) \setminus \{u_1 \} = \{u_2, u_3 \}$. 
Recall that $u_2$ and $u_3$ are lonely vertices. 
We distinguish 2 subcases:

\textbf{Case 3a:} $V_1 \cup V_2 \neq \emptyset$. 
Take any pendant (block) node $b' \in V_1 \cup V_2$, say $b' \in V_i$. 
Let $\ell'$ be the vertex adjacent to $b'$. 
By the assumption that there are no cheap bridge-covering paths and \Cref{lem:bridge-covering-helper1} (applied to $xy = b'\ell'$), $R(b') \setminus \{ \ell' \}$ has cardinality at least 2 and contains only lonely vertices.
Choose any $u' \in R(b') \setminus \{ \ell' \}$. Notice that $u' \notin V_1$ by (a) (but it could be a lonely vertex in $V_2$ other than $\ell'$).
The tuple $(b, b', u_3, u')$ satisfies the conditions of \Cref{lem:bridge-covering-helper2} (in particular both $T_C(b, u_3)$ and $T_C(b', u')$ contain $w = u_2$), hence we can compute the desired $H'$.

\textbf{Case 3b:} $V_1 \cup V_2 = \emptyset$. By \Cref{lem:bridge-covering-helper1} (applied to $xy = u_1 u_2$) the set $R( \{b, u_1 \})$ contains a block node or $R(\{ b, u_1 \} ) \setminus \{ u_2 \}$ contains at least $2$ lonely vertices. 
Suppose first that $R( \{ b, u_1 \} )$ contains a block node $b'$. 
Notice that $b' \notin R(b)$ by the assumption that there are no cheap bridge-covering paths, hence $u_1 \in R(b')$. 
Notice also that $b' \notin \{u_2, u_3 \}$ since those are lonely vertices. 
Thus, the tuple $(b, b', u_2, u_1)$ satisfies the conditions of \Cref{lem:bridge-covering-helper2} (in particular both $T_C(b, u_2)$ and $T_C(b', u_1)$ contain $w = u_2$), hence we can obtain the desired $H'$.

The remaining case is that $R(\{ b, u_1 \} ) \setminus \{ u_2 \}$ contains at least 2 lonely vertices. 
Let us choose $u' \in R(\{ b, u_1 \} ) \setminus \{ u_2 \}$ with $u' \neq u_3$. 
Let $P_{b u_2}$ (resp., $P_{u' u_1}$) be a bridge-covering path between $b$ and $u_2$ (resp., $u'$ and $u_1$). 
Notice that $P_{b u_2}$ and $P_{u' u_1}$ must be internally vertex disjoint, as otherwise $u' \in R(b)$, which is excluded since $R(b) \subseteq \{u_1, u_2, u_3 \}$ by assumption.
Consider $H' \coloneq H \cup E(P_{b u_2}) \cup E(P_{u' u_1}) \setminus \{u_1 u_2\}$. 
Notice that $H'$ has fewer bridges than $H$. 
One has $|H'| = |H| + |E(P_{b u_2})| + |E(P_{u' u_1})| -1$, where the $-1$ comes from the removal of the edge $u_1 u_2$. 
Furthermore, $\credit(H') \leq \credit(H) + 1 - (|E(P_{bu_2})| -1) - (|E(P_{u'u_1})| -1) - \frac{1}{4} \cdot 4 -1$, 
where the $+1$ comes from the extra credit needed for the block or 2EC component $C'$ containing $V(B)$, 
the final $-1$ from the removed block $B$, and the $-\frac{1}{4} \cdot 4$ from the at least $4$ bridges removed from $H$ 
(namely a bridge corresponding to $b u_1$, the edges $u_1 u_2$ and $u_2 u_3$, and one more bridge incident to $u'$). 
Altogether $\cost(H') \leq \cost(H)$.
It can be easily checked that $H'$ is also canonical and hence $H'$ satisfies the lemma statement.
\end{proof}

\end{document}